\documentclass[english]{myjflart}	
\usepackage[utf8]{inputenc}
\usepackage[T1]{fontenc}

\usepackage[utf8]{inputenc}

\usepackage{enumerate}
\usepackage{array}
\usepackage{tikz-cd}

\usepackage{mathtools}
\usepackage{stmaryrd}
\usepackage{cmll}

\usepackage[%
  full,		
]{optional}

\author{Colin Riba}
\author{Solal Stern}
\authorrunning{Riba and Stern}

\affil{  ENS Lyon, Universit{\'e} de Lyon, LIP%
  \footnote{UMR 5668 CNRS ENS Lyon UCBL INRIA}
}

\newcommand\mytitle{Liveness Properties in Geometric Logic for Domain-Theoretic Streams}

\usepackage{macros}
\newcommand{\fntext}{}
\newcommand{\fn}{\footnote{\fntext}}

\newcommand\colin[1]{\opt{draft}{{\color{red}\sffamily #1}}}

\jfla{35}{2024}

\title{\mytitle}

\begin{document}

%


\maketitle

\begin{abstract}
We devise a version of Linear Temporal Logic ($\LTL$) on a denotational
domain of streams.
We investigate this logic in terms of domain theory, (point-free) topology
and geometric logic.
This yields the first steps toward
an extension of the ``Domain Theory in Logical Form'' paradigm
to temporal liveness properties.

We show that the negation-free formulae of $\LTL$ induce
sober subspaces of streams,
but that this is in general not the case in presence of negation.
We propose a direct, inductive, translation of negation-free $\LTL$ to geometric logic.
This translation reflects the approximations
used to compute the usual fixpoint representations of $\LTL$ modalities.

As a motivating example,
we handle a natural input-output
specification for the usual $\filter$ function on streams.
\end{abstract}

\section{Introduction}
\label{sec:intro}

%

We are interested in input-output
properties of higher-order programs that handle
infinite 
data,
such as streams or non-wellfounded trees.
Consider for instance the usual $\filter$ function
\[
\begin{array}{r c l}
  \filter
& :
& (\Base \to \Bool)
  ~\longto~
  \Stream \Base
  ~\longto~
  \Stream \Base
\\


  \filter\ p\ (a \Cons x)
& =
& \term{if}~ (p\ a)
  ~\term{then}~ a \Cons (\filter\ p\ x)
  ~\term{else}~ (\filter\ p\ x)
\end{array}
\]

\noindent
where $\Stream\Base$ stands for the type of streams on $\Base$.
Assume $p : \Base \to \Bool$ is a total function that
tests for a property $P$.
If $x$ is a stream on $\Base$,
then $(\filter\ p\ x)$ retains those elements of $x$ which satisfy $P$.
The stream produced by $(\filter\ p\ x)$ is thus only partially defined,
unless $x$ has infinitely many elements satisfying $P$.

Logics like $\LTL$, $\CTL$ or the modal $\mu$-calculus are
widely used to formulate, on infinite objects,
safety and liveness properties (see e.g.~\cite{hr07chapter,bs07chapter}).
Safety properties state that
some ``bad'' event will not occur,
while liveness properties 
specify that ``something good'' will happen 
(see e.g.~\cite{bk08book}).
One typically uses
temporal modalities like $\Box$
(\emph{always})
or $\Diam$
(\emph{eventually})
to write properties of streams and specifications of programs over such data.

\renewcommand\fntext{In the setting of~\cite{jr21esop},
we would assume $\Base = \sum_{i=1}^n \one$,
with $\Phi_i$ representing the image of the $i$th injection.}

A possible specification for 
$\filter$
asserts that $(\filter\ p \ x)$ is a totally defined stream whenever
$x$ is a totally defined stream with infinitely many elements satisfying $P$.
We express this with the temporal modalities $\Box$ and $\Diam$.
Let $\Base$ be finite, and assume given,
for each $a$ of type $\Base$,
a formula $\Phi_a$ which holds on $b : \Base$
exactly when $b$ equals $a$.\fn\ 
Then
$\Box \bigvee_a \Phi_a$
selects those streams on $\Base$
which are totally defined.
The formula $\Box\Diam P$ expresses that a stream
has infinitely many elements satisfying $P$.
We can thus state that for all streams $x : \Stream\Base$,
\begin{equation}
\label{eq:intro:spec}
\begin{array}{c}
  \text{$x$ satisfies $\Box\bigvee_a \Phi_a$ and $\Box\Diam P$}
  \quad\longimp\quad
  \text{$(\filter\ p\ x)$ satisfies $\Box\bigvee_a \Phi_a$}
\end{array}
\end{equation}


The question we address in this paper is the following.
Having in mind that a stream (as opposed to e.g.\ an integer)
is inherently an infinite object,
what do we mean exactly by ``the stream $x$ satisfies $\Box\Diam P$''?
In our view, the above specification for $\filter$
should hold for any stream whatsoever, and not only
for those definable in a given programming language.

This leads us to investigate temporal properties
on infinite datatypes at the level of denotational semantics.
Logics on top of domains are known since quite a long time.
Our reference is the paradigm of
``Domain Theory in Logical Form'' (DTLF) \cite{abramsky91apal}
(see also~\cite{zhang91book}),
which allows one to systematically generate a logic from a 
domain representing a type.
These logics are actually obtained by Stone duality,
which
is at the core of a
rich interplay between domain theory, logic and (point-free) topology.
This area is presented under various perspectives in a number of sources.
We refer to~\cite{abramsky91apal,ac98book}
and (e.g.)
\cite{johnstone82book,vickers89book,vickers07chapter,goubault13book,gg23book}.
Some key ideas are put at work in~\cite{cz00csl}.


However, logics on domains given by Stone duality
are usually restricted to safety properties.
To our knowledge, there is no systematic investigation
of liveness properties, such as the ones used
in the specification for $\filter$ above.

This paper reports on preliminary works, mostly based on an internship
of the second author during summer 2023.
We devise a version of the logic $\LTL$
on a domain of streams $\I{\Stream\Base}$
determined by the recursive type equation
$\Stream \Base \cong \Base \times \Stream\Base$.
Each formula $\Phi$ of $\LTL$ yields a subset $\I\Phi \sle \I{\Stream\Base}$.
We investigate such $\LTL$-definable subsets
in terms of domain theory, of (point-free) topology
and of a logic called geometric logic.


Our first step is to view domains as topological spaces, so as to
benefit from the rich notion of subspace.
For instance,
(with $\Base$ finite)
the set of $\omega$-words
$\Base^\omega = \I{\Box \bigvee_a \Phi_a}$
turns out to be a
discrete sub-poset of $\I{\Stream\Base}$.
But as a subspace of $\I{\Stream\Base}$,
it becomes equipped with its usual product topology
(in the sense of e.g.~\cite{pp04book}).
We observe that $\LTL$ formulae without negation
induce subspaces of $\I{\Stream\Base}$
which are sober,
but that this may fail in presence of negation.
The notion of sobriety originates from point-free topology, and has become 
quite important for
the general (point-set) topology of domains (see e.g.~\cite{goubault13book}).

%
%
%
%

We then turn to geometric logic.
The idea is roughly the following. 
DTLF
rests on the fact that
finite approximations in a domain
can be represented in a 
propositional logic
generated from the topology of the domain.
But this is too weak to handle infinitary properties
such as those definable with the modalities $\Box$ and $\Diam$.
On the other hand, the sobriety of $\I\Phi$ means that we can
reason using an abstract notion of approximation induced by the subspace
topology.
Geometric logic is an infinitary propositional logic
which allows for concrete representations of topologies.
We provide a direct, inductive, translation of negation-free $\LTL$
to a geometric logic based on the domain $\I{\Stream\Base}$.
This translation reflects the approximations
used to compute the usual fixpoint representations of $\Box,\Diam$.
This shows that for the negation-free fragment, the semantics
of $\LTL$ can be concretely represented by approximations which
live in a natural extension of DTLF for the domain $\I{\Stream\Base}$.

We also 
check
that our translation of negation-free $\LTL$
indeed
conveys the good approximations to prove that the denotation of $\filter$
meets the specification~\eqref{eq:intro:spec} above.

\renewcommand\fntext{Actually, as well as
e.g.~\cite{nukt18lics,su23popl}, despite a fundamentally different
approach (see~\S\ref{sec:conc}).}

Let us finally mention the scientific context of this work.
It is undecidable whether a given higher-order program satisfies a 
given input-output temporal property written with formulae
of the modal $\mu$-calculus~\cite{ktu10popl}.
A previous work with the first author
provided a refinement type system for proving such properties~\cite{jr21esop}.
This type system handles the alternation-free modal $\mu$-calculus
on (finitary) polynomial types, which includes $\LTL$.
But it is based on guarded recursion
and does not allow for non-productive functions such as $\filter$.
We ultimately target a similar refinement type systems
for a language based on $\FPC$
(which extends Plotkin's seminal $\PCF$~\cite{plotkin77tcs} with
recursive types, see e.g.~\cite{pierce02book}).
We think that the present work is a significant step in this direction.
On the one hand,
DTLF allows for reasoning on denotations using (finitary)
type systems~\cite{abramsky91apal}.
On the other hand,
it has been advocated in~\cite{kt14lics} that a form of oracle
is needed to handle liveness properties in type systems.
And indeed, \cite{jr21esop}
incorporates
such oracles in a notion of ``iteration term'',
which in fact makes the system infinitary.\fn\
We think that our representation of negation-free $\LTL$ in geometric
logic can lead to an infinitary type system which extends~\cite{abramsky91apal},
and whose infinitary part
can be simulated using iteration terms.
%

\paragraph{Organization of the paper.}
The preliminary~\S\ref{sec:prelim} introduces background on domain theory,
and the logic $\LTL$ on 
$\I{\Stream\Base}$.
The (point-free) topological approach is presented in~\S\ref{sec:frames},
and~\S\ref{sec:geom} is devoted to geometric logic.
\opt{full}{Section~\ref{sec:free} deals with a deduction system for geometric
logic, in connection with the notion of spatiality. }%
The specification of $\filter$ is discussed in~\S\ref{sec:scott}.
We conclude in~\S\ref{sec:conc}.
%
\opt{short}{Proofs are available in the full version~\cite{rs23full},
which also contains additional material on deduction for geometric logic.}%
\opt{long}{Proofs are available in the Appendices,
which also contain additional material on deduction for geometric logic.}%
\opt{full}{Proofs are available in the Appendices.}


\colin{TODO FOR SUBMISSION:
\begin{itemize}
\item Enable final option (for final version).
\end{itemize}}

%
%
%

\section{A Linear Temporal Logic on a Domain of Streams}
\label{sec:prelim}

Let $\Base$ be a set.
A (finite) \emph{word} on $\Base$ is an element of $\Base^*$.
$\Base^\omega$
is the set of \emph{$\omega$-words} on $\Base$, i.e.\ the set of all
functions $\sigma \colon \NN \to \Base$.
We write $\word \sle \wword$ when 
$\word \in \Base^*$ is a prefix of $\wword \in \Base^* \cup \Base^\omega$.
The concatenation of 
$\word \in \Base^*$ with $\wword \in \Base^* \cup \Base^\omega$
is denoted $\word \cdot \wword$ or $\word \wword$.
Given $\iword \in \Base^\omega$ and $k \in \NN$,
we let $\iword\restr k \in \Base^\omega$
be the $\omega$-word with $(\iword\restr k)(n) = \iword(k+n)$
for all $n \in \NN$.
For instance, $\iword\restr 0$ is $\iword$, while
\(
  \iword\restr 1
  =
  \iword(1) \cdot \iword(2) \cdots \iword(n+1) \cdots
\)
is $\iword$ deprived from its first letter.

\subsection{Domains}
\label{sec:prelim:domain}

The basic idea of domain theory is to represent a type
by partial order $(X,\leq_X)$ thought about as an ``information order''.
The intuition is that $x \leq_X y$ means that
$y$ has ``more information'' than $x$, or that
$x$ is ``less defined'' than $y$.
Domains are often required to have a least element
(representing plain divergence),
and are always asked to be stable under certain supremums
(so that infinite objects can be thought about as limits of their finite approximations).
Our presentation mostly follows~\cite[\S 1]{ac98book}.
See also~\cite{abramsky91apal,goubault13book}.

\paragraph{Dcpos and Cpos.}
Let $(X,\leq)$ be a partial order (or \emph{poset}).
An \emph{upper bound} of a subset $\SP \sle X$
is an element $x \in X$ such that $(\forall s \in \SP)(s \leq x)$.
A \emph{least upper bound} (or \emph{supremum}, \emph{sup}) of $\SP$
is an upper bound $\ell$ of $\SP$ such that 
$\ell \leq x$ for every 
upper bound $x$ of $\SP$.
The sup of $\SP$ is unique whenever it exists, 
and is usually denoted $\bigvee \SP$.
The notion of \emph{greatest lower bound}
(or \emph{infimum}, \emph{inf}) is defined dually.
A subset $D \sle X$ is \emph{directed} if $D$ is non-empty and
for every $x,y \in D$,
there is some $z \in D$ such that $x \leq z$ and $y \leq z$.
%


We say that $(X,\leq)$ is a \emph{dcpo} if every directed $D \sle X$
has a sup $\bigvee D \in X$.
A \emph{cpo} is a dcpo with a least element (usually denoted $\bot$).
Note that each set $\Base$ is a dcpo for the discrete order
(in which $x$ is comparable with $y$ if, and only if, $x = y$).
However, such a dcpo $\Base$ is not a cpo unless $\Base$ is a singleton.

\begin{exem}[Flat Domains]
Given a set $\Base$,
the \emph{flat domain} $\I\Base$ is the disjoint union $\{\bot\} + \Base$
equipped with the partial order $\leq_{\I\Base}$, where
$x \leq_{\I\Base} y$
iff $x = y$ or ($x = \bot$ and $y \in \Base$).


It is easy to see that
$(\I\Base,\leq_{\I\Base})$ is a cpo
whose directed subsets have at most
one element from $\Base$.
For instance, the domain $\I\Bool$ can be represented by the following Hasse diagram.
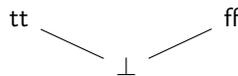
\[
\begin{tikzcd}[row sep=tiny]
  \term{tt}
&
& \term{ff}
\\
& \bot
  \arrow[dash]{ul}
  \arrow[dash]{ur}
\end{tikzcd}
\]
\end{exem}
\paragraph{Scott-Continuous Functions.}
Let $X = (X,\leq_X)$ and $Y = (Y,\leq_Y)$ be dcpos.
A function $f \colon X \to Y$ is \emph{Scott-continuous}
if $f$ is monotone
($x \leq_X x'$ implies $f(x) \leq_Y f(x')$)
and if moreover $f$ preserves directed sups, in the sense that
for each directed $D \sle X$, we have
\[
\begin{array}{l l l}
  f(\bigvee D)
& =
& \bigvee \left\{f(d) \mid d \in D \right\}
\end{array}
\]

\noindent
We write $\CPO$ (resp.\@ $\DCPO$) for the category with
cpos (resp.\@ dcpos) as objects and with
Scott-continuous functions as morphisms. 
We say that $f \in \CPO\funct{X,Y}$ is strict if $f(\bot_X) = \bot_Y$.
A non-strict monotone map between flat domains is necessarily constant.

Given dcpos $X = (X,\leq_X)$ and $Y = (Y,\leq_Y)$, the set of
Scott-continuous functions $\DCPO\funct{X,Y}$
is itself a dcpo w.r.t.\ the \emph{pointwise order}
\[
\begin{array}{l l l}
  f \leq_{\DCPO\funct{X,Y}} g
& \text{iff}
& \forall x \in X,~
  f(x) \leq_{Y} g(x)
\end{array}
\]


\noindent
If $Y$ is actually a cpo, then $\DCPO\funct{X,Y}$ is a cpo whose
least element is the constant function
$x \in X \mapsto \bot_Y \in Y$, where $\bot_Y$ is the least element of $Y$.

\begin{exem}[Streams]
Let $\Base$ be a set.
We let
$\I{\Stream\Base}$, the cpo of \emph{streams over $\Base$},
be $\DCPO\funct{\NN,\I\Base}$ with $\NN$ discrete.
We unfold this important example.
Since $\NN$ is discrete, $\I{\Stream\Base}$ actually consists
of the set $\I\Base^\omega$ equipped with the partial order
\[
\begin{array}{l !{\quad\text{iff}\quad} l}
  \stream \leq_{\I{\Stream\Base}} \streambis
& \forall n \in \NN,~ \stream(n) \leq_{\I\Base} \streambis(n)
\end{array}
\]

\noindent
A set $D \sle \I{\Stream\Base}$ is directed if, and only if,
$D$ is non-empty and
each $D(n) = \left\{ x(n) \mid x \in D\right\}$
has at most one element from $\Base$.
Then $\bigvee D \in \I{\Stream\Base}$
takes $n \in \NN$ to
the largest element of $D(n)$.
The least element of $\I{\Stream\Base}$ is the
stream
$\bot^\omega$ of constant value $\bot \in \I\Base$.

Note that $\I{\Stream\Base}$ has ``partially defined'' elements.
Besides the least element $\bot^\omega$,
we have e.g.\ the stream $\word \cdot \bot^\omega$
(which agrees with $\word \in \Base^*$ and then is $\bot$ at all sufficiently
large positions) or $(a \cdot \bot)^\omega$ (which is $a$ at all even positions,
and is $\bot$ everywhere else).
The $\omega$-words on $\Base$ are precisely those streams
$\stream \in \I{\Stream\Base}$
which never take
the value $\bot$. Such streams are called \emph{total}.
Note that if $\stream$ is total, then
\[
\begin{array}{*{3}{l}}
  \stream
& =
& \bigvee
  \left\{
    \word \cdot \bot^\omega \mid
   \text{$\word \in \Base^*$ and $\word \sle \stream$}
  \right\}
\end{array}
\]

%
%
\end{exem}

\begin{remark}
\label{rem:prelim:stream:domain-equation}
The cpo $\I{\Stream\Base}$ is 
the usual solution in the category $\CPO$
of the
\emph{domain equation}
\[
\begin{array}{l l l}
  X
& \cong
& \I\Base
  \times
  X
\end{array}
\]

\noindent
(where $\I\Base \times X$ is equipped with the pointwise order),
see e.g.~\cite[Theorem 7.1.10 and Proposition 7.1.13]{ac98book}.
In particular, the 
constructor $(- \Cons -)$ of the type $\Stream\Base$
is interpreted as the isomorphism
taking
$(a,\stream) \in \I\Base \times \I{\Stream\Base}$
to
$a \cdot \stream \in \I{\Stream\Base}$,
with inverse $\stream \mapsto (\stream(0), \stream\restr 1)$.
Note that $\I{\Stream\Base}$ differs from the usual \emph{Kahn domain}
$\Base^* \cup \Base^\omega$
(see e.g.~\cite[Definition 3.7.5 and Example 5.4.4]{vickers89book}
or~\cite[\S 7.4]{dst19book}, see also~\cite{vvk05concur}).
\end{remark}
\begin{remark}
\label{rem:prelim:filter}
Each 
$f \colon X \to_\CPO X$
has a \emph{least fixpoint}
$\term{Y}(f) \deq \bigvee_{n \in \NN} f^n(\bot) \in X$.
In particular, $\filter$ is interpreted as the
Scott-continuous function $\I\filter$
taking
$p \colon \I\Base \to_{\CPO} \I\Bool$
to the least fixpoint of the following function $f_p$,
where $X$ is the cpo
$\I{\Stream\Base} \to_{\CPO} \I{\Stream\Base}$.
\[
\begin{array}{*{7}{l}}
  f_p
& \deq
& \lambda g.\lambda \stream.~
  \term{if}~ p (\stream(0))
  ~\term{then}~ \stream(0) \cdot g(\stream \restr 1)
  ~\term{else}~ g(\stream \restr 1)
& :
& X
& \longto_{\CPO}
& X
\end{array}
\]
%
\end{remark}

%
%

\paragraph{Algebraicity.}
Among the many good properties of $\I{\Stream\Base}$,
algebraicity is the crucial one in this work.
This property is not used right away, but will be the
main assumption of various statements later on.

\renewcommand\fntext{Finite elements are called \emph{compact}
in~\cite{ac98book}.}
Let $(X,\leq)$ be a dcpo.
We say that $x \in X$ is \emph{finite}\fn\ if for every directed
$D \sle X$ such that $x \leq \bigvee D$, there is some $d \in D$
such that $x \leq d$. 
We say that $X$ is \emph{algebraic} if for every $x \in X$,
the set $\{d \in X \mid \text{$d$ finite and $d \leq x$} \}$
is directed and has sup $x$.
Each discrete or flat dcpo is algebraic.


\begin{exem}[Streams]
\label{ex:prelim:stream:algebraic}
The cpo $\I{\Stream\Base}$ is algebraic, and its finite
elements admit a particularly simple description.
The \emph{support} of $x \in \I{\Stream\Base}$ is
the set $\supp(x)$ of ``defined letters'' of $x$:
\[
\begin{array}{l l l}
  \supp(x)
& \deq
& \left\{ n \in \NN \mid x(n) \neq \bot \right\}
\end{array}
\]

\noindent
We say that a stream $\stream$ has finite support
when $\supp(\stream)$ is a finite set.
For instance, given a finite word $\word \in \Base^*$
and $n \in \NN$,
the stream $\bot^n \cdot \word\cdot \bot^\omega$ has finite support.
On the other hand, total streams, as well as e.g. $(a\cdot \bot)^\omega$,
do not have finite support.

For each $x \in \I{\Stream\Base}$,
the set 
\(
 \left\{
   d \mid
   \text{$d$ of finite support and $d \leq_{\I{\Stream\Base}} \stream$}
  \right\}
\)
is directed and has sup $\stream$.
Moreover, 
the finite elements of $\I{\Stream\Base}$
are exactly those of finite support.
%
\end{exem}

\subsection{Linear Temporal Logic ($\LTL$)}
\label{sec:prelim:ltl}

\paragraph{Syntax and Semantics.}
Let $\Base$ be a set.
The formulae of $\LTL = \LTL(\Base)$ 
are given by 
\[
\begin{array}{l l l}
  \Phi,\Psi
& \bnf
& a
  \gs
  \True
  \gs
  \False
  \gs
  \Phi \land \Psi
  \gs
  \Phi \lor \Psi
  \gs
  \lnot \Phi
  \gs
  \Next \Phi
  \gs
  \Phi \Ushort \Psi
  \gs
  \Phi \Wshort \Psi
\end{array}
\]

\noindent
where $a \in \Base$.
Hence, besides pure propositional logic, $\LTL(\Base)$
has \emph{atomic formulae} $a \in \Base$, and \emph{modalities}
$\Next \Phi$ (read ``next $\Phi$''),
$\Phi \Ushort \Psi$ (read ``$\Phi$ until $\Psi$'')
and $\Phi \Wshort \Psi$ (read ``$\Phi$ weak until $\Psi$'' or ``$\Phi$ unless $\Psi$''). 

The $\LTL$ formulae over $\Base$ are usually interpreted
on $\omega$-words over $\Base$, see e.g.~\cite[\S 5]{bk08book}.
The interpretation of modalities actually implicitly relies on the bijection
$\Base^\omega \cong \Base \times \Base^\omega$.
We similarly rely on the isomorphism
$\I{\Stream\Base} \cong_{\CPO} \I\Base \times \I{\Stream\Base}$
for interpreting $\LTL(\Base)$ formulae in $\I{\Stream\Base}$.
We define $\MI\Phi \sle \I{\Stream\Base}$ by induction on $\Phi$. 
The propositional connectives of $\LTL$ are interpreted using the usual
Boolean algebra structure of the powerset $\Po(\I{\Stream\Base})$.
For $a \in \Base$,
we let $\I{a} \deq \{ \stream \in \I{\Stream\Base} \mid \stream(0) = a\}$.
The modalities are interpreted as follows.
\[
\begin{array}{r c l}

  \I{\Next \Phi}
& \deq
& \left\{ \stream \in \I{\Stream\Base} \mid
  \stream \restr 1 \in \I\Phi
  \right\}
\\


  \I{\Phi \Ushort \Psi}
& \deq
& \left\{
  \stream \in \I{\Stream\Base} \mid
  \exists i \in \NN,~
  \stream \restr 0,\dots, \stream \restr(i-1) \in \I\Phi
  ~\text{and}~
  \stream \restr i \in \I{\Psi}
  \right\}
\\


  \I{\Phi \Wshort \Psi}
& \deq
& \left\{
  \stream \in \I{\Stream\Base} \mid
  \forall i \in \NN,~ \stream \restr i \in \I\Phi
  \right\}
  \cup
  \I{\Phi \Ushort \Psi}

\end{array}
\]


\noindent
We say that $\stream \in \I{\Stream\Base}$ \emph{satisfies} a formula
$\Phi$ (notation $\stream \forces \Phi$) when $\stream \in \MI\Phi$.
It is often convenient to decompose 
$\I{\Next \Phi}$ as $\I\Next(\I\Phi)$,
where
$\I\Next \colon \Po(\I{\Stream\Base}) \to \Po(\I{\Stream\Base})$
takes $\SP$ to
\(
\left\{ \stream \in \I{\Stream\Base} \mid \stream\restr 1 \in \SP \right\}
\).
The modalities $\Ushort$ and $\Wshort$ may not be easy to grasp.
Given $\LTL$ formulae $\Phi$ and $\Psi$,
we let
\[
\begin{array}{r !{~~} r c l !{~~} r !{=} l}
  \text{``eventually $\Psi$''}
& \Diam \Psi
& \deq
& \True \Ushort \Psi
& \big(
  \I{\Diam\Psi}
& \left\{
  \stream \in \I{\Stream\Base} \mid
  \exists i \in \NN,~ \stream \restr i \in \I\Psi
  \right\}
  \big)
\\

  \text{``always $\Phi$''}
& \Box \Phi
& \deq
& \Phi \Wshort \False
& \big(
  \I{\Box\Phi}
& \left\{
  \stream \in \I{\Stream\Base} \mid
  \forall i \in \NN,~ \stream \restr i \in \I\Phi
  \right\}
  \big)
\end{array}
\]

\begin{exem}
\label{ex:ltl:base}
Consider a stream $\stream \in \I{\Stream\Base}$.
\begin{enumerate}[(1)]
\item
We have
$\stream \forces \Next a$ if, and only if, $\stream(1) = a$.
For instance, $\bot a \bot^\omega \forces \Next a$
but $a \bot^\omega \not\forces \Next a$.

\item
We have $\stream \forces \Diam a$ if, and only if,
$\stream(i) = a$ for some $i \in \NN$.
For instance, $\bot^n a \bot^\omega \forces \Diam a$ for every $n \in \NN$.
But $b^\omega \not\forces \Diam a$ if $b \neq a$.

\item
\label{ex:ltl:base:total}
We have $\stream \forces \Box a$ if, and only if, $\stream = a^\omega$.
If $\Base$ is finite, then $\stream$ is total
iff
$\stream \forces \Box \bigvee_{a \in \Base} a$.

\item
We have $\stream \forces \Box\Diam a$ if, and only if,
$\stream(i) = a$ for infinitely many $i \in \NN$.
E.g.\ $(\bot a)^\omega \forces \Box\Diam a$.

\item
We have
$\stream \forces \Diam\Box a$ if, and only if,
$\stream(i) = a$ for ``ultimately all $i \in \NN$''.
This means that for some $n \in \NN$, we have $\stream(i) = a$
for all $i \geq n$.
For instance,
$\bot^n a^\omega \forces \Diam\Box a$
for all $n \in \NN$.
But $(\bot a)^\omega \not\forces \Diam\Box a$.
\end{enumerate}
\end{exem}

Say that $\Phi$ and $\Psi$
are (logically) \emph{equivalent},
notation $\Phi \equiv \Psi$, if $\I\Phi = \I\Psi$.
$\LTL$ has many redundancies w.r.t.\ logical equivalence.
Besides the usual De Morgan laws,
we have e.g.
\[
\begin{array}{l !{\equiv} l !{\qquad} l !{\equiv} l !{\qquad} l !{\equiv} l}
  \lnot \Box \Phi
& \Diam \lnot \Phi

& \lnot \Diam \Phi
& \Box \lnot \Phi

& \Phi \Wshort \Psi
& (\Phi \Ushort \Psi) \lor \Box \Phi
\end{array}
\]

\begin{remark}
\label{rem:ltl:fix}
The modalities $\Ushort$ and $\Wshort$ are also ``De Morgan'' duals,
in the following sense.
Given $\Phi$ and $\Psi$, it is well-known that
$\I{\Phi \Ushort \Psi}$ and $\I{\Phi \Wshort \Psi}$
are respectively the least and the greatest fixpoint of the
(monotone) map on $\Po(\I{\Stream\Base})$
taking $\SP$
to $\I\Psi \cup \left( \I\Phi \cap \I\Next(\SP) \right)$.
See e.g.~\cite[Lemmas 5.18 and 5.19]{bk08book}.
But $\Po(\I{\Stream\Base})$ is a complete atomic Boolean algebra, and
given a monotone endo-function $f$
on such a Boolean algebra,
the least and greatest fixpoints of $f$ are related by
$\lfp(f) = \lnot \gfp(b \mapsto \lnot f(\lnot b))$
and
$\gfp(f) = \lnot \lfp(b \mapsto \lnot f(\lnot b))$.
%
%
\end{remark}

\paragraph{Negation-Free $\LTL$.}
Our main positive results only hold for the negation-free fragment of $\LTL$.
An $\LTL$ formula is  \emph{negation-free}
(\emph{n.-f.})
if it contains no negation ($\lnot(-)$).
Hence, the negation-free formulae of $\LTL$ are generated by the above
grammar for $\LTL$, but without the production $\lnot \Phi$.

\begin{exem}
\label{ex:ltl:nf}
All formulae of Example~\ref{ex:ltl:base} are negation-free.
Moreover, 
the negation-free fragment is closed under $\Box(-)$ and $\Diam(-)$.


Assume $\Base$ is finite.
For any $\SP \sle \Base$,
there is a negation-free formula $\Psi_\SP$ such that
$\stream \forces \Psi_\SP$ iff $\stream(0) \in \SP$.
It follows that for any Scott-continuous $p \colon \I\Base \to \I\Bool$,
there is a negation-free formula $\Psi_p$ such that
$\stream \forces \Box\Diam\Psi_p$ if, and only if,
$\stream$ has infinitely many elements satisfying $p$.
\end{exem}

Most
redundancies of $\LTL$ mentioned above disappear in the
negation-free fragment.
This is why we have chosen this set of connectives from the start.
In negation-free $\LTL(\Base)$,
all connectives have a De Morgan dual.
But negated atomic formulae ($\lnot a$ for $a \in \Base$)
are 
not available.
Hence, in contrast with positive normal forms
(see e.g.~\cite[Definition 5.20]{bk08book}),
negation is \emph{not} definable in negation-free $\LTL$.
This positive character 
is reflected in the following fundamental fact,
proved by induction on formulae.

\begin{lemm}
\label{lem:ltl:up}
If $\Phi$ is n.-f.\ then $\I\Phi$
is upward-closed 
(if $\stream \in \I\Phi$ and $\stream \leq_{\I{\Stream\Base}} \streambis$
then $\streambis \in \I\Phi$).
\end{lemm}




\begin{coro}
\label{cor:ltl:dcpo}
Let $\Phi$ be negation-free.
Then $\I\Phi$ is closed in $\I{\Stream\Base}$ under directed sups.
Moreover,
the inclusion $\I\Phi \emb \I{\Stream\Base}$
is a Scott-continuous order-embedding.
\end{coro}

Hence, $\I\Phi$ is a sub-dcpo of $\I{\Stream\Base}$ when $\Phi$ is n.-f.
But this may not give much information on $\I\Phi$.
For instance, 
$\Base^\omega = \I{\Box\bigvee_{a \in \Base} a}$
(Example~\ref{ex:ltl:base}(\ref{ex:ltl:base:total}), $\Base$ finite)
is a discrete dcpo.
%
Building on Lemma~\ref{lem:ltl:up},
we are going to exhibit much more structure on 
such inclusions $\I\Phi \emb \I{\Stream\Base}$.
But before, we note that
Lemma~\ref{lem:ltl:up} and Corollary~\ref{cor:ltl:dcpo} may fail
in presence of negation.

\begin{exem}
\label{ex:ltl:counter}
Consider the formula $\lnot \Box a$.
Note that $a^\omega \not\forces \lnot \Box a$.
But for every finite $d \leq_{\I{\Stream\Base}} a^\omega$,
we have $d \forces \lnot \Box a$.
Hence $\I{\lnot \Box a}$ is not upward-closed.
Moreover, $\{\text{$d$ finite} \mid d \leq_{\I{\Stream\Base}} a^\omega\}$
is a directed subset of $\I{\lnot \Box a}$ which has no sup in $\I{\lnot \Box a}$.
Hence $\I{\lnot \Box a}$ is not a dcpo w.r.t.\ the restriction of
$\leq_{\I{\Stream\Base}}$.
\end{exem}

\section{The Topological Approach}
\label{sec:frames}

We shall now look at inclusions
$\I\Phi \emb \I{\Stream\Base}$ 
from a topological perspective.
We recall in~\S\ref{sec:frames:spaces}
that the categories $\CDCPO$ can be embedded in the category $\Top$
of topological spaces and continuous functions.
The highlight is that $\Top$ has a 
much richer notion of substructures (called subspaces)
than $\CDCPO$.

Actually, when looking at (d)cpos as topological spaces,
the notion of sobriety
from point-free (or ``element-free'') topology
comes to the front.
Ample mathematical justifications for
the importance of sober spaces in domain theory
are gathered in~\cite{goubault13book}.
We shall content ourselves with
more informal motivations in~\S\ref{sec:frames:frames}.
In~\S\ref{sec:frames:sub}, we abstractly prove that $\I\Phi$ 
induces a sober subspace of $\I{\Stream\Base}$
when $\Phi$ is negation-free.
\opt{full,long}{This will be refined to concrete representations
in~\S\ref{sec:geom} and~\S\ref{sec:free}, using geometric logic.}%
\opt{short}{This will be refined to concrete representations
in~\S\ref{sec:geom} (and also~\cite[\S 5]{rs23full}), using geometric logic.}

\subsection{Topological Spaces}
\label{sec:frames:spaces}

\noindent
A \emph{topological space} is a pair $(X,\Open(X))$ of a set $X$
and a collection $\Open = \Open(X)$ of subsets of $X$, called \emph{open sets}.
$\Open$
is called a \emph{topology} on $X$,
and is asked to be stable under arbitrary unions and under finite intersections.
In particular, $\emptyset$ and $X$ are open in $X$
(respectively as the empty union and the empty intersection).

A set $C \sle X$ is \emph{closed} if its complement $X \setminus C$ is open.
Closed sets are stable under finite unions and arbitrary
intersections.
Hence, any $\SP \sle X$ is contained in a least closed
set $\clos \SP  \sle X$. 
Each space $(X,\Open)$ is equipped with a \emph{specialization} (pre)order
$\leq_\Open$ on $X$, defined as
\[
\begin{array}{l l l}
  x \leq_\Open y
& \text{iff}
& (\forall U \in \Open)
  \left(
  x \in U
  ~\longimp~
  y \in U
  \right)
\end{array}
\]

\noindent
Given $x \in X$, we have
$\clos{\{x\}} = \down x \deq \{ y \in X \mid y \leq_{\Open} x\}$
(see e.g.~\cite[Lemma 4.2.7]{goubault13book}).
A topology $\Open$ 
is \emph{$T_0$} when $\leq_\Open$ is a partial order
(see e.g.~\cite[Proposition 4.2.3]{goubault13book}).

Given spaces $(X,\Open(X))$ and $(Y,\Open(Y))$,
a function $f \colon X \to Y$ is \emph{continuous}
when its inverse image $f^{-1} \colon \Po(Y) \to \Po(X)$
restricts to a function $\Open(Y) \to \Open(X)$,
i.e.\ when $f^{-1}(V) \in \Open(X)$ for all $V \in \Open(Y)$.
We write $\Top$ for the category of topological spaces and
continuous functions.
An \emph{homeomorphism} is an isomorphism in $\Top$.

\paragraph{The Scott Topology.}
The following is well-known.
See e.g.~\cite[\S 1.2]{ac98book} or~\cite[\S 4]{goubault13book}.


Let $(X,\leq_X)$ be a dcpo.
A subset $U \sle X$ is \emph{Scott-open} if $U$ is upward-closed,
and if moreover $U$ is inaccessible by directed sups,
in the sense that if $\bigvee D \in U$
with $D \sle X$ directed,
then $D \cap U \neq \emptyset$.
This equips $X$ with a $T_0$ topology, called the \emph{Scott topology},
whose specialization order coincides with $\leq_X$.
Note that $C \sle X$ is Scott-closed precisely when $C$ is downward-closed
and stable under directed sups.

\begin{exem} 
\label{ex:topo:stream}
When $(X,\leq)$ is algebraic, the sets
$\up d \deq \{x \in X \mid d \leq x\}$ with $d$ finite
form a sub-basis for the Scott topology.
For instance,
the Scott-open subsets of $\I{\Stream\Base}$ are arbitrary unions
of sets of the form
\[
\begin{array}{l l l}
  \up d
& =
& \left\{
  x \in \I{\Stream\Base} \mid
  \forall i \in \supp(d),~
  x(i) = d(i)
  \right\}
\end{array}
\]

\noindent
with $\supp(d)$ finite.
In particular, given $x \in U$ with $U \sle \I{\Stream\Base}$ Scott-open,
there is a \emph{finite} set $\{i_1,\dots,i_k \} \sle \NN$ 
such that
\(
  \left\{ y \mid y(i_1)=x(i_1),\dots,y(i_k)=x(i_k)
  \right\}
  \sle U
\).

Beware that $\I\Phi$ may not be an open
nor a closed subset of $\I{\Stream\Base}$, even when $\Phi$ is negation-free.
Consider for instance
$\Base^\omega = \I{\Box\bigvee_{a \in \Base} a}$
(with $\Base$ non-empty and finite),
which contains $a^\omega$ but no finite $d \leq_{\I{\Stream\Base}} a^\omega$.
\end{exem}

A function $f \colon X \to Y$ between dcpos is
Scott-continuous precisely when $f$ is continuous w.r.t.\
the Scott-topologies on $X$ and $Y$.
It follows that $\DCPO$ and $\CPO$ are full subcategories of $\Top$.
From now on, we shall mostly look at $\CDCPO$ in this way.
Unless stated otherwise, dcpos will always be equipped with their Scott
topology.

\paragraph{Subspaces.}
Our motivation for moving from $\CDCPO$ to $\Top$
is that $\Top$ has a rich notion of subspace.
We refer to~\cite[\S 1.2]{bbt20book}.
Given a 
space $(X,\Open)$ and a subset $P \sle X$,
the \emph{subspace topology} on $P$ is
\[
\begin{array}{l l l}
  \Open\restr P
& \deq
& \left\{
  U \cap P \mid U \in \Open
  \right\}
\end{array}
\]

\noindent
The subspace topology on $P$ makes the inclusion function
$\incl \colon P \emb X$ continuous.
It is the ``best possible'' topology on $P$
in the following sense:
given a space $(Y,\Open(Y))$,
a function $f \colon Y \to P$ is continuous if, and only if,
the composition $\incl \comp f \colon Y \to X$ is continuous.
\[
\begin{tikzcd}[row sep=tiny]
& P
  \arrow[hookrightarrow]{dr}{\incl}
\\
  Y
  \arrow{ur}{f}
  \arrow{rr}[below]{\incl \comp f}
&
& X
\end{tikzcd}
\]

\begin{exem}
\label{ex:topo:omegaword}
Generalizing Example~\ref{ex:ltl:base}(\ref{ex:ltl:base:total})
and Lemma~\ref{lem:ltl:up},
$\Base^\omega$ is a discrete sub-dcpo of $\I{\Stream\Base}$.

On the other hand, 
the subspace topology
$\Open(\I{\Stream\Base}) \restr \Base^\omega$
is the usual \emph{product topology}
on $\Base^\omega$
(see e.g.~\cite[\S III]{pp04book} or~\cite{kechris95book}).
The sets of the form $\Base^\omega \cap \up d$
(with $d$ finite) form a sub-basis for
the subspace topology.
In fact,
its opens
are 
unions of sets of the form
$\{ \iword \in A^\omega \mid \iword(i_1) = a_1, \dots, \iword(i_k) = a_k \}$,
where
$i_1,\dots,i_k \in \NN$ and $a_1,\dots,a_k \in \Base$ with $k \geq 0$.
\end{exem}


\subsection{The Element-Free Setting}
\renewcommand\fntext{Dually, the knowledge of the whole $\omega$-word
$\iword$ may be needed to testify that $\iword \notin U$.}
\label{sec:frames:frames}
The topological setting comes with an intrinsic notion of approximation.

Consider for instance $\omega$-words $\iword \in \Base^\omega$
(Example~\ref{ex:topo:omegaword}).
Similarly as with streams in Example~\ref{ex:topo:stream},
given an $\omega$-word $\iword$ and an open $U$, if $\iword$ belongs to $U$,
then this fact is witnessed
by the knowledge of a finite number of elements of $\iword$.\fn\
We view the opens $U$ such that $\iword \in U$ as 
approximations of $\iword$.

Given a space $(X,\Open)$, we are interested in 
describing the elements of $X$ by 
their approximations, represented as suitable sets of opens $\Filt \sle \Open$.
This is the realm of \emph{element-free} (or \emph{point-free}) \emph{topology}.
Its central objects, called
\emph{frames} (or \emph{locales}),
abstract away from the elements of spaces, and only retain
the lattice structure of open sets.
Besides \cite{goubault13book},
we refer to~\cite{johnstone82book,johnstone83bams,vickers89book,pp12book,pp21book}.
\paragraph{Frames.}
A \emph{complete lattice} is a poset having all sups and all infs.
But recall (from e.g.~\cite[Theorem 2.31]{dp02book})
that a poset has all sups if, and only if, it has all infs.
Hence, we can see complete lattices indifferently as posets with all sups or
as posets with all infs.

A \emph{frame} is a poset $L$ with all sups (and thus all infs),
and which satisfies the following \emph{frame distributive law}:
for all $\SP \sle L$ and all $a \in L$,
\[
\begin{array}{l l l}
  a \wedge \bigvee \SP
& =
& \bigvee \left\{ a \wedge s \mid s \in \SP \right\}
\end{array}
\]


\renewcommand\fntext{Consider e.g.\ a finite (and thus complete) non-distributive
lattice, see~\cite[Example 4.6(6)]{dp02book}.}
\noindent
Not every complete lattice is a frame.\fn\
But every $(\Po(X),\sle)$ is a frame,
and so is the two-elements poset
$\two \deq \{0 \leq 1\}$.

Given frames $L$ and $K$,
a \emph{frame morphism} $f \colon L \to K$
is a function which preserves all sups and all \emph{finite} infs.
Note that frame morphisms are automatically monotone.
We write $\Frm$ for the category of frames and frame morphisms.

\begin{exem}
Let $(X,\Open(X))$ be a space.
Then $\Open(X)$ has all sups and they are given by unions.
Hence $\Open(X)$ is a complete lattice.
Beware that the inf in $\Open(X)$ of an arbitrary $\SP \sle \Open(X)$
is in general not its intersection $\bigcap \SP$,
but the \emph{interior} of $\bigcap \SP$
(the largest open set contained in $\bigcap \SP$).
However, \emph{finite} infs in $\Open(X)$ are given by intersections,
and $\Open(X)$ is a frame.
In fact, the topologies on a given set $X$ correspond exactly to
the sub-frames of $\Po(X)$.


Moreover, if $f \colon X \to Y$ is continuous,
then its inverse image
$f^{-1}$ restricts to a function $\Open(Y) \to \Open(X)$.
This function is actually a frame morphism
$\Open(f) \in \Frm\funct{\Open(Y), \Open(X)}$.
In other words, the operation $(X,\Open(X)) \mapsto \Open(X)$
extends to a functor $\Open$ from the category $\Top$ to $\Frm^\op$,
the opposite of $\Frm$.
The category $\Frm^\op$ is the category of \emph{locales}.
\end{exem}

\paragraph{The Space of Points.}
We see a frame $L$ as a collection of formal approximations.
Suitable subsets of $L$ describe ``converging'' sets of formal approximations,
and constitute the elements of a space, the space of points of $L$.
The idea is as follows.
Given a space $X$ and $x \in X$, let
\[
\begin{array}{l l l}
  \Filt_x
& \deq 
& \left\{
  U \in \Open(X) \mid x \in U
  \right\}
\end{array}
\]

%
%

\noindent
Note the following properties of $\Filt_x$
w.r.t.\ the frame structure of $\Open(X)$.
First,
$\Filt_x$ is stable under finite intersections
($x \in X$, and $x \in  U \cap V$ iff $x \in U$ and $x \in V$).
Second, given $\SP \sle \Open(X)$ with $\bigcup \SP \in \Filt_x$,
we have
$U \in\Filt_x$ for some $U \in \SP$
(if $x \in \bigcup \SP$ then $x \in U$ for some $U \in \SP$).
Hence, the characteristic function of $\Filt_x \sle \Open(X)$
is a frame morphism $\Open(X) \to \two$.

A \emph{point} of a frame $L$ is an element of $\pt(L) \deq \Frm\funct{L,\two}$.
We shall always identify a point $\Filt \in \Frm\funct{L,\two}$
with the set $\{a \in L \mid \Filt(a) = 1 \}$.
Given $a \in L$, let
\[
\begin{array}{l l l}
  \ext(a)
& \deq
& \left\{ \Filt \in \pt(L) \mid a \in \Filt \right\}
\end{array}
\]

\noindent
The function
$\ext \colon L \to \Po(\pt(L))$ is a frame morphism
(see e.g.~\cite[Lemma II.1.6]{johnstone82book}).
In particular, its image is a sub-frame of $\Po(\pt(L))$,
and is thus a topology $\Open(\pt(L))$ on $\pt(L)$.
The \emph{space of points} of $L$ is $(\pt(L),\Open(\pt(L))$.

The operation $L \mapsto \pt(L)$ extends to a functor
$\pt \colon \Frm^\op \to \Top$
which is right adjoint to $\Open$ (see e.g.~\cite[Theorem II.1.4]{johnstone82book}).
The action of $\pt \colon \Frm^\op \to \Top$
on a frame morphism $f \colon L \to K$
is the continuous function $\pt(K) \to \pt(L)$
which takes $\Filt \in \Frm\funct{K, \two}$ to $\Filt \comp f \in\Frm\funct{L, \two}$.
The unit at $X \in \Top$ of the adjunction $\Open \adj \pt$
is the continuous function $\eta_X \colon X \to \pt(\Open(X))$
taking $x$ to $\Filt_x$
(see e.g.~\cite[\S II.1.6]{johnstone82book}).

\paragraph{Sober Spaces.}
\renewcommand\fntext{Actually,
\cite[Lemma II.1.6(ii)]{johnstone82book} states
that each $T_2$ space is sober.}
The function $\eta_X \colon X \to \pt(\Open(X))$
continuously maps the space $X$ to its space of
``converging formal approximations'' $\pt(\Open(X))$.
But $\pt(\Open(X))$ may not correctly represent $X$.
A space $X$ is \emph{sober} if $\eta_X$ is a bijection
(in which case $\eta_X$ is automatically an homeomorphism,
see \cite[\S II.1.6]{johnstone82book}).

Given a frame $L$, the space of points $\pt(L)$ is always sober
(\cite[Lemma II.1.7]{johnstone82book}).
Hence (by functoriality of $\Open$ and $\pt$),
if $X$ is homeomorphic to $\pt(L)$, then $X$ is sober as well.
It follows from~\cite[Lemma II.1.6(ii)]{johnstone82book}
that $\Base^\omega$ is sober for its product topology.\fn\
But not every dcpo is sober (\cite[II.1.9]{johnstone82book}).
For the following, see e.g.~\cite[Theorem VII.2.6]{johnstone82book}.

\begin{prop} 
\label{prop:frames:dcposober}
Algebraic dcpos are sober.
\end{prop}

\begin{remark}
\label{rem:frames:scott}
In fact, a sober space
is always $T_0$, and is moreover a dcpo w.r.t.\ its specialization order
(\cite[Lemmas II.1.6(i) and II.1.9]{johnstone82book}).
This provides a functor $\Frm^\op \to \DCPO$
which is actually right adjoint to the composite
$\DCPO \emb \Top \to \Frm^\op$,
yielding the \emph{Scott adjunction} of~\cite{diliberti22act}.

In particular, a sober dcpo is completely determined by the specialization
order of its space of points.
On the other hand, beware that the composite
$\Frm^\op \to \DCPO \to \Frm^\op$
may loose a lot of structure 
(e.g.\ it takes the product topology on $\Base^\omega$
to the discrete (Scott) topology).
\end{remark}

\subsection{Sobriety of Subspaces}
\label{sec:frames:sub}

Proposition~\ref{prop:frames:dcposober} and Remark~\ref{rem:frames:scott}
imply that for an algebraic dcpo $X$, the topological notion of
approximation coincides with the domain-theoretic one.
But we are interested in subspaces of the algebraic cpo $\I{\Stream\Base}$.
Discussing the sobriety of such subspaces
involves going 
further into the point-free setting.
While the main results of this~\S\ref{sec:frames:sub} are important for this paper,
the technical developments
\opt{full,long}{are used again only in~\S\ref{sec:free}.}%
\opt{short}{are used again only in~\cite[\S 5]{rs23full}.}%


Let $(X,\Open)$ be a space, and consider some $P \sle X$.
The subspace inclusion $\incl \colon (P,\Open\restr P) \emb (X,\Open)$
induces
the surjective frame morphism
$\ladj\incl \deq \Open(\incl) \colon \Open \quot \Open\restr P$
which
takes $U \in \Open$ to $(U \cap P) \in \Open\restr P$.
The following is a handy reformulation of sobriety for $(P,\Open\restr P)$.

\begin{lemm}
\label{lem:frames:sub:soberfilt}
Assume that $(X,\Open)$ is sober. Then the following are equivalent.
\begin{enumerate}[(i)]
\item
$(P,\Open\restr P)$ is sober.

\item
For each $x \in X$, we have $x \in P$ if, and only if,
$\Filt_x = \Filtbis \comp \ladj\incl$
for some $\Filtbis \in \pt(\Open\restr P)$.
\[
\begin{tikzcd}[row sep=tiny]
  \Open
  \arrow[twoheadrightarrow]{rr}{\ladj\incl}
  \arrow{dr}[below]{\Filt_x}
&
& \Open\restr P
  \arrow[dashed]{dl}{\Filtbis}
\\
& \two
\end{tikzcd}
\]
\end{enumerate}
\end{lemm}

Let $L$ be a frame.
A \emph{quotient frame} of $L$ is an isomorphism-class of
surjective frame morphisms $L \quot K$.
We are going to discuss
an abstract but mathematically powerful
representation of the quotient frame $\Open \quot \Open \restr P$.
We use tools from~\cite[\S II.2]{johnstone82book} and~\cite[\S VI.1]{pp12book}
on the dual notion of \emph{sub-locale}.

Everything starts from Galois connections and related adjointness properties,
for which we refer to~\cite[7.23--7.34]{dp02book}.
Fix a frame morphism $f \colon L \to K$.
Since $f \colon L \to K$ preserves all sups, it has an
\emph{upper adjoint} $\radj f \colon K \to L$.
This means that
for all $a \in L$ and all $b \in K$,
\[
  f(a) \leq_K b
  \qquad
  \text{if, and only if,}
  \qquad
  a \leq_L \radj f(b)
\]

\noindent
The pair $(f,\radj f)$ thus forms a Galois connection,
and $\radj f$ is (uniquely) determined by
\[
\begin{array}{l l l}
  \radj f(b)
& =
& \bigvee_L \left\{ a \mid f(a) \leq_K b \right\}
\end{array}
\]

\noindent
The function $\radj f$
is in general not a frame morphism, but it always preserves all infs.

The composition $j \deq \radj f \comp f \colon L \to L$
is a \emph{nucleus} in the sense of~\cite[\S II2.2]{johnstone82book}:
we have (i) $j(a \wedge a') = j(a) \wedge j(a')$,
(ii) $a \leq j(a)$
and (iii), $j(j(a)) \leq j(a)$.
Nuclei are monotone and idempotent.
If $j \colon L \to L$ is a nucleus, then
the set 
$L_j \deq \left\{ a \in L \mid j(a) = a\right\}$
of \emph{$j$-fixpoints}
is a frame and $j \colon L \to L_j$ is a frame morphism
(\cite[Lemma II.2.2]{johnstone82book}).
Note that the finite infs in $L_j$ are those of $L$.
But the sup of $\SP \sle L_j$ in $L_j$ is 
$j(\bigvee_L \SP)$.

%

Consider now a subspace inclusion
$\incl \colon (P,\Open\restr P) \emb (X,\Open)$.
Following~\cite[\S VI.1.1]{pp12book},
we write $\widetilde P$ for the frame of $j$-fixpoints,
where $j \deq \radj\incl \comp \ladj\incl$ and $\radj\incl$ is the upper
adjoint of $\ladj\incl$.
\[
\begin{tikzcd}[column sep=large]
  \Open
  \arrow[twoheadrightarrow]{r}{\ladj\incl}
& \Open\restr P
  \arrow{r}{\radj\incl}
& \Open
\end{tikzcd}
\]


\noindent
We rely on the following description of the nucleus $j$
(from which~\cite[\S VI.1.1]{pp12book}
gives an explicit representation of $j$-fixpoints that we shall not use directly).

\begin{remark}
\label{rem:frames:sub:nucleus}
Given an open $U \in \Open$ of $X$, we have
\[
\begin{array}{l l l}
  j(U)
& =
& \bigcup \left\{ V \in \Open \mid
  V \cap P = U \cap P
  \right\}
\end{array}
\]
\end{remark}

%

The proof of~\cite[Theorem II.2.3]{johnstone82book}
gives Lemma~\ref{lem:frames:sub:quot} below.
Recall that order-isomorphisms preserve all existing sups and infs
(\cite[Lemma 2.27(ii)]{dp02book}).

\begin{lemm}
\label{lem:frames:sub:quot}
The function $\radj\incl \colon \Open\restr P \to \Open$
co-restricts to a frame isomorphism $\radj\incl \colon \Open\restr P \to \widetilde P$.
\end{lemm}

We use
$j = \radj\incl \ladj\incl \colon \Open \quot \widetilde P$
to represent the quotient frame induced by the subspace inclusion
$\incl \colon (P,\Open\restr P) \emb (X,\Open)$.
The frame $\widetilde P$ turns out to be a good tool for studying the sobriety
of $(P,\Open \restr P)$.
Following~\cite[VI.1.3]{pp12book}, given $x \in X$
we let $\widetilde x \deq X \setminus \clos{\{x\}} = X \setminus \down x$.
We shall now see that it is useful
to characterize when
$\widetilde x \in \widetilde P$
(i.e.\ when $j(\widetilde x) = \widetilde x$).

\begin{remark}
\label{rem:frames:sub:widetilde}
Given $x \in X$ and $U \in \Open$, we have $U \sle \widetilde x$
if, and only if, $x \notin U$.
\end{remark}

\begin{lemm}
\label{lem:frames:sub:widefilt}
Let $x \in X$ and
$\widetilde\Filt_x \deq \{ U \in \widetilde P \mid x \in U \}$.
Then
$\widetilde\Filt_x \in \pt(\widetilde P)$
if and only if
$\widetilde x \in \widetilde P$.
\end{lemm}

\begin{prop}
\label{prop:frames:sub:tilde}
Let $P \sle X$ with $(X,\Open)$ sober.
Then the following are equivalent.
\begin{enumerate}[(i)]
\item
$(P,\Open\restr P)$ is sober.

\item
\label{item:frames:sub:tilde:tilde}
For each $x \in X$, we have $x \in P$ if, and only if,
$\widetilde x \in \widetilde P$.
\end{enumerate}
\end{prop}

\noindent
In condition~(\ref{item:frames:sub:tilde:tilde}) above,
we actually always have
$\widetilde x \in \widetilde P$ when $x \in P$
(see \cite[VI.1.3.1]{pp12book}).



Proposition~\ref{prop:frames:sub:tilde}
will yield a general sufficient condition for the sobriety
of $(P,\Open \restr P)$ (Theorem~\ref{thm:frames:sub:upsober} below),
from which we will obtain the case of negation-free $\LTL$
(Corollary~\ref{cor:frames:sub:upsober}).

One further step into the point-free setting gives us sharper results.
A space
$(X,\Open)$ is $T_D$ when for each $x \in X$,
there is some open $U \in \Open$ such that $x \in U$ and
$(U \setminus \{x\}) \in \Open$. 
See~\cite[\S I.2]{pp12book}.
It is shown in~\cite[Proposition VI.1.3.1]{pp12book}
that if $X$ is a (possibly not sober) $T_D$ space, then
condition~(\ref{item:frames:sub:tilde:tilde})
of Proposition~\ref{prop:frames:sub:tilde}
holds for \emph{any} $P \sle X$.
It follows that if $X$ is sober \emph{and} $T_D$,
then each $P \sle X$ induces a sober subspace.

Consider now the case of a sober space $(X,\Open)$ which is not $T_D$.
Hence, there is some $x \in X$ such that for all open $U$
with $x \in U$, the set $U \setminus \{x\}$ is not open.

\begin{lemm}
\label{lem:frames:sub:notsober}
Let $x \in X$ as above and set $P \deq X \setminus \{x\}$.
Then $(P,\Open\restr P)$ is not sober.
\end{lemm}

\renewcommand\fntext{Actually, each $T_1$ space is $T_D$ (\cite[\S I.2.1]{pp12book}).}
It follows from~\cite[\S I.2.1]{pp12book}
that $\Base^\omega$ is $T_D$
for the product topology.\fn\
But $\I{\Stream\Base}$ is not $T_D$, unless $\Base = \emptyset$.
Consider $a^\omega \in \I{\Stream\Base}$.
Then any Scott-open $U$ containing $a^\omega$ contains
also some finite $d \leq_{\I{\Stream\Base}} a^\omega$.
Hence $U \setminus \{a^\omega\}$ is not upward-closed
and thus not Scott-open.

\begin{coro}
\label{cor:frames:sub:notsober}
$\I{\lnot\Box a} = \I{\Stream\Base} \setminus \{a^\omega\}$
is not a sober subspace of $\I{\Stream\Base}$.
\end{coro}


Let $(X,\Open)$ be sober, and let $P \sle X$
be upward-closed for $\leq_\Open$. 
Assume $x \notin P$.
Then $P \setminus \down x = P$,
and thus $P \cap \widetilde x = P$.
Hence
\(
  j(\widetilde x)
  =
  \bigcup \left\{
  V \in \Open \mid V \cap P = P
  \right\}
\),
so $j(\widetilde x) = X$
and $\widetilde x \notin \widetilde P$.
It follows that $x \in P$ precisely when $\widetilde x \in \widetilde P$,
and Proposition~\ref{prop:frames:sub:tilde} gives the following.

\begin{theo}
\label{thm:frames:sub:upsober}
If $(X,\Open)$ is sober and $P \sle X$ is upward-closed for $\leq_\Open$,
then $(P,\Open\restr P)$ is sober.
\end{theo}

\begin{coro}
\label{cor:frames:sub:upsober}
If $(X,\Open)$ is a sober dcpo and if $P \sle X$ is upward-closed,
then $(P,\Open\restr P)$ is sober.
In particular, if $\Phi$ is negation-free, then $\I\Phi$
is a sober subspace of $\I{\Stream\Base}$.
\end{coro}

The importance we give to the negation-free fragment of $\LTL$
ultimately rests
on Corollaries~\ref{cor:frames:sub:notsober} and~\ref{cor:frames:sub:upsober}.
But the frame $\widetilde{P}$ seems too abstract to be used concretely.

\section{Geometric Logic}
\label{sec:geom}

Geometric logic is an infinitary propositional logic which describes frames.
Very roughly, the idea is that if a theory $\th{Th}$ in geometric logic
represents a frame $L$, then the models of $\th{Th}$ can be organized in
a space which is homeomorphic to $\pt(L)$, the space of points of $L$.

We will represent the (sober) space $\I{\Stream\Base}$ by a theory
$\th{T}\I{\Stream\Base}$
in geometric logic.
Further, for each negation-free formula $\Phi$ of $\LTL$,
we shall (inductively) devise a theory $\th{T}\I\Phi$
such that $\th T\I{\Stream\Base} \cup \th T\I\Phi$
represents the (sober) subspace induced by $\I\Phi \emb \I{\Stream\Base}$.
This will provide a concrete presentation of the corresponding quotient frame.

Our approach to geometric logic
here
is not the usual one, as presented in e.g.~\cite{vickers07chapter}
(see also~\cite[\S D]{johnstone02book}).
\opt{full,long}{The relations between the two approaches are discussed
in~\S\ref{sec:free}.}%
\opt{short}{The relations between the two approaches are discussed
in~\cite[\S 5]{rs23full}.}


\subsection{Geometric Theories}
\label{sec:geom:th}

%
%

\paragraph{Formulae and Valuations.}
Let $\At$ be a set of \emph{atomic propositions}.
The \emph{conjunctive} 
and the \emph{geometric} formulae 
over $\At$
are respectively defined as
\[
\begin{array}{l l l !{\qquad\text{and}\qquad} l l l}
  \gamma,\gamma'
  \in \Conj(\At)
& \bnf
& p
  \gs
  \true
  \gs
  \gamma \land \gamma'

& \varphi,\psi,\theta \in \Geom(\At)
& \bnf
& 
  \bigvee \SP
\end{array}
\]


\noindent
where $p \in \At$ and $\SP \sle \Conj(\At)$.
%
A \emph{valuation} of $\At$ is a function $\nu \colon \At \to \two$.
Given $\chi \in \Conj(\At) \cup \Geom(\At)$,
the \emph{satisfaction relation} $\nu \models \chi$
is defined by
\[
\begin{array}{l c l !{\qquad} l c l}
  \nu
  \models
  \true
&
&

& \nu
  \models
  \gamma \land \gamma'
& \text{iff}
& \text{$\nu \models \gamma$ and $\nu \models \gamma'$}
\\

  \nu \models p
& \text{iff}
& \nu(p) = 1


& \nu \models \bigvee \SP
& \text{iff}
& \text{there exists $\gamma \in \SP$ such that $\nu \models \gamma$}

\end{array}
\]



%

We let $\false$ be the geometric formula $\bigvee \emptyset$.
We may write
$\gamma$ for the geometric formula $\bigvee\{\gamma\}$.
Given conjunctive formulae $(\gamma_i \mid i \in I)$,
we write $\bigvee_{i \in I} \gamma_i$ 
for the geometric formula $\bigvee \{\gamma_i \mid i \in I \}$.
Note that
$\bigvee_{i \in I} \gamma_i = \bigvee_{j \in J} \gamma'_j$
if there is a bijection $f \colon I \to J$
with $\gamma_i = \gamma'_{f(i)}$.

\begin{remark}
\label{rem:geom:conn}
There is no primitive notion of conjunction or disjunction on geometric formulae,
but they can be defined.
Given $(\varphi_i \mid i \in I)$ with
$\varphi_i = \bigvee \left\{ \gamma_{i,j} \mid j \in J_i \right\}$,
we define $\bigvee_{i \in I} \varphi_i$ to be the geometric formula
$\bigvee \left\{ \gamma_{i,j} \mid \text{$i \in I$ and $j \in J_i$}\right\}$.
We then have
$\nu \models \bigvee_{i \in I} \varphi_i$
iff
$\nu \models \varphi_i$ for some $i \in I$.

Similarly, given $\varphi = \bigvee_{i \in I} \gamma_i$
and $\psi = \bigvee_{j \in J} \gamma'_j$,
we define
$\varphi \land \psi \deq \bigvee\{ \gamma_i \land \gamma'_j \mid (i,j) \in I \times J \}$.
Then
$\nu \models \varphi \land \psi$
iff
$\nu \models \varphi$ and $\nu \models \psi$.
\end{remark}

\paragraph{Sequents and Theories.}
A \emph{sequent} over $\At$ is a pair $\psi \thesis \varphi$
of geometric formulae $\varphi,\psi \in \Geom(\At)$.
A valuation $\nu$ of $\At$
is a \emph{model}
of $\psi \thesis \varphi$ 
if $\nu \models \psi$ implies $\nu \models \varphi$.
Note that $\nu$ is a model of the sequent $\varphi \thesis \false$
if, and only if,
$\nu \not\models \varphi$.

A \emph{geometric theory} over $\At$ is a set $\th{T}$ of sequents over $\At$.
A valuation $\nu$ of $\At$ is a \emph{model} 
of $\th{T}$
if $\nu$ is a model of all the sequents of $\th{T}$.
We write $\Mod(\th{T})$ for the set of models of $\th{T}$.

An \emph{antecedent-free sequent} has
the form $\true \thesis \varphi$, and is denoted $\thesis \varphi$
(or even $\varphi$) for short.
An \emph{antecedent-free theory} consists of antecedent-free sequents only.

\paragraph{Algebraic Dcpos.}
Algebraic dcpos have a natural representation by
geometric theories.
This relies on the following well-known facts, for which we
refer to~\cite[\S 1.1]{ac98book}.

An \emph{ideal} on a poset $(P,\leq)$ is a subset $J \sle P$
which is downward-closed and directed.
The set $\Idl(P)$ of ideals on $P$ is an algebraic dcpo for inclusion.
The operation $P \mapsto \Idl(P)$ is left adjoint
to the forgetful functor
from $\DCPO$
to the category of posets and monotone maps.

Let $X$ be a dcpo, and let 
$\Fin(X)$ be 
its sub-poset of finite elements.
Given an ideal $J \in \Idl(\Fin(X))$, since $J$ is directed
we have $\bigvee J \in X$.
For the following, see e.g.~\cite[Proof of Proposition 1.1.21(2)]{ac98book}.


\begin{lemm} 
\label{lem:geom:algdcpo}
Let $X$ be an algebraic dcpo.
The function $J \in \Idl(\Fin(X)) \mapsto \bigvee J \in X$
is an order-isomorphism.
Its inverse
$X \to \Idl(\Fin(X))$ takes $x$ to $\{ d \in \Fin(X) \mid d \leq x \}$.
\end{lemm}

We represent an algebraic dcpo $X = (X,\leq_X)$ by a geometric theory $\th T(X)$
over $\At = \Fin(X)$.
The theory $\th T(X)$ consists of 
$\thesis \bigvee \Fin(X)$
together with all the sequents
\[
\begin{array}{c}

  d \thesis d'
\quad\text{(if $d' \leq_X d$)}

\qquad\quad

%

  d \land d'
  \thesis
  \bigvee \left\{
  d'' \in \Fin(X) \mid \text{$d \leq_X d''$ and $d' \leq_X d''$}
  \right\}
\end{array}
\]


\noindent
where $d,d' \in \Fin(X)$.
Note that $J \sle \Fin(X)$ is an ideal if, and only if,
its characteristic function $\Fin(X) \to \two$ is a model of $\th T(X)$.
Combined with Lemma~\ref{lem:geom:algdcpo},
this yields Proposition~\ref{prop:geom:algdcpo} below.
If $x \in X$,
let
$\nu(x) \colon \At \to \two$
be the characteristic function
of $\{d \in \Fin(X) \mid d \leq_X x\}$.

\begin{prop}
\label{prop:geom:algdcpo}
The map $x \mapsto \nu(x)$ is a bijection $X \to \Mod(\th T(X))$.
\end{prop}



When $X$ is actually an algebraic cpo, 
let $\th T_\bot(X)$ be the theory obtained from
$\th T(X)$
by replacing the antecedent-free sequent $\thesis \bigvee \Fin(X)$
with $\thesis \bot_X$
(where $\bot_X \in \Fin(X)$ is the least element of $X$).
The theories $\th T(X)$ and $\th T_\bot(X)$ have exactly the same models.
Hence Proposition~\ref{prop:geom:algdcpo} also holds with $\th T_\bot(X)$
in place of $\th T(X)$.


\begin{exem}[Streams]
\label{ex:geom:streams}
We further simplify the theory 
representing the cpo $\I{\Stream\Base}$.


Consider $\stream,\streambis \in \I{\Stream\Base}$
such that $\up\stream \cap \up\streambis \neq \emptyset$
(i.e.\ such that 
$\stream,\streambis \leq_{\I{\Stream\Base}} \streamter$
for some $\streamter \in \I{\Stream\Base}$).
Then, since $\I\Base$ is a flat cpo, we have
$\stream(n) = \streambis(n)$
for all $n \in \supp(\stream) \cap \supp(\streambis)$.
It follows that the set $\{\stream,\streambis\}$
has a sup $\stream \vee_{\I{\Stream\Base}} \streambis$ in $\I{\Stream\Base}$.
Note that
$\stream \vee_{\I{\Stream\Base}} \streambis$ is finite whenever so are $\stream$ and
$\streambis$.

We can thus represent $\I{\Stream\Base}$ with the following theory
$\th T\I{\Stream\Base}$,
where $d,d' \in \Fin(\I{\Stream\Base})$.
\[
\begin{array}{r !{\thesis} l @{\quad} l !{\qquad\qquad} r !{\thesis} l @{\quad} l}
  d
& d'
& \text{(if $d' \leq_{\I{\Stream\Base}} d$)}

& d \land d'
& \false
& \text{(if $\up d \cap \up d' = \emptyset$)}
\\

& \bot^\omega
& 

& d \land d'
& d \vee_{\I{\Stream\Base}} d'
& \text{(if $\up d \cap \up d' \neq \emptyset$)}
\end{array}
\]
\end{exem}

\begin{remark}
\label{rem:spectral}
\renewcommand\fntext{Note that $\I{\Stream\Base}$ is always compact,
since any Scott-open containing $\bot^\omega$ contains the whole of $\I{\Stream\Base}$.
This contrasts with the product topology on $\Base^\omega$,
which is compact iff $\Base$ is finite
(see e.g.~\cite[\S III.3.5]{pp04book}).}
Note that the 
theory $\th T(\I{\Stream\Base})$ of Example~\ref{ex:geom:streams}
only involves finite geometric formulae.
Actually, this amounts to the fact
that $\I{\Stream\Base}$
is \emph{spectral} in its Scott topology
(see~\cite[Corollary 7.48 and Definition 6.2]{gg23book}).\fn\
Roughly, being a spectral space means that the topology
can be generated from a distributive lattice (as opposed to a frame).
See e.g.~\cite[Proposition 3.26 and Theorem 6.1]{gg23book} for details.

Spectral cpos include those known as ``SFP'' domains
(see e.g.~\cite[\S 2.2]{abramsky91apal}).
SFP domains are also called ``bifinite'' domains
(see e.g.~\cite[Definition 5.2.2 and Theorem 5.2.7]{ac98book}).
They are stable under most common domain operations
and have solutions for recursive type equations
(see e.g.~\cite[\S 2.2]{abramsky91apal}).

\opt{short}{\renewcommand\fntext{In the terminology
of~\cite[\S 5]{rs23full},
this is because frames induced by distributive lattices are always
spatial (see e.g.~\cite[Theorem II.3.4]{johnstone82book}).}}
\opt{full,long}{\renewcommand\fntext{In the terminology of~\S\ref{sec:free},
this is because frames induced by distributive lattices are always
spatial (see e.g.~\cite[Theorem II.3.4]{johnstone82book}).}}
In the paradigm of ``Domain Theory in Logical Form'',
spectral domains are particularly important
because the logic of the underlying distributive lattice
can be incorporated into finitary type systems~\cite{abramsky91apal}
(see also~\cite[\S 10.5]{ac98book} for a simple instance).\fn
\end{remark}

\subsection{The Sober Space of Models}
\label{sec:geom:spaces}

Given a geometric theory $\th T$ over $\At$,
we will now equip $\Mod(\th T)$
with a sober topology induced by a quotient of $\Geom(\At)$.
In particular, this will extend the bijection of Proposition~\ref{prop:geom:algdcpo}
to an homeomorphism between an algebraic dcpo $X$ and the space
of models $\Mod(\th T(X))$.

In view of Example~\ref{ex:geom:streams},
we may have $\Mod(\th{T}) = \Mod(\th{U})$
for different theories $\th{T}$ and $\th{U}$ over $\At$.
The topology on $M= \Mod(\th{T})$
depends on $M$ (and $\At$),
but not on the theory $\th{T}$ such that $M= \Mod(\th{T})$.
Fix a set of atomic propositions $\At$
and let $M$ be a set of the form $\Mod(\th{T})$ for some theory $\th{T}$
over $\At$.
Define
\[
  \gmod_M \colon \Geom(\At) \to \Po(M)
  ,\quad
  \varphi \mapsto \left\{ \nu \in M \mid \nu \models \varphi \right\}
\]

\noindent
Let $\Open(M) \sle \Po(M)$ be the image of $\gmod_M$.
We have
$\gmod_M(\true) = M$,
and
(via Remark~\ref{rem:geom:conn}),
\[
\begin{array}{c !{\quad\text{and}\quad} c}
  \gmod_M(\varphi \land \psi)
  =
  \gmod_M(\varphi) \cap \gmod_M(\psi)

& \gmod_M\left(\bigvee_{i \in I} \varphi_i \right)
  =
  \bigcup_{i \in I} \gmod_M(\varphi_i)
\end{array}
\]

%

\noindent
It follows that $(\Open(M),\sle)$ is stable
under the sups and the finite infs of $(\Po(M),\sle)$.
Hence $(\Open(M),\sle)$ is a sub-frame of $(\Po(M),\sle)$,
and $(M,\Open(M))$ is a topological space.

Given a theory $\th{T}$,
we write $\Mod(\th{T})$ for $(\Mod(\th{T}),\Open(\Mod(\th{T})))$,
the \emph{space of models of $\th{T}$}.
Since the space $\Mod(\th{T})$ only depends on the models of $\th{T}$,
the following directly applies to the theory $\th T\I{\Stream\Base}$
of Example~\ref{ex:geom:streams}.


\begin{prop}
\label{prop:geom:algdcpo:scott}
Let $X$ be an algebraic dcpo.
The bijection $x \mapsto \nu(x)$
of Proposition~\ref{prop:geom:algdcpo} extends to an homeomorphism from
$X$ to $\Mod(\th T(X))$.
\end{prop}

Spaces of models are always sober.
To this end,
given $M$ as above,
we quotient $\Geom(\At)$ under the preorder $\preceq_M$
with $\varphi \preceq_M \psi$ iff $\gmod_M(\varphi) \sle \gmod_M(\psi)$.
The relation $\sim_M$ of \emph{$M$-equivalence}
on $\Geom(\At)$
is defined
as $\varphi \sim_M \psi$ iff
$\varphi \preceq_M \psi$ and $\psi \preceq_M \varphi$
(i.e.\ $\gmod_M(\varphi) = \gmod_M(\psi)$).
We let $\class{\varphi}_M$ be the $\sim_M$-class of $\varphi$,
and $\Geom(\At)/M$ be the set of
$\sim_M$-classes of geometric formulae.
We write $\leq_M$ for the partial order on
$\Geom(\At)/M$ induced by the preorder $\preceq_M$
(see e.g.~\cite[\S 2.3.1]{goubault13book}).

The function $\gmod_M$ yields an order-isomorphism
$(\Geom(\At)/M,\leq_M) \to (\Open(M),\sle)$.
Since order-isomorphisms preserve all existing sups and infs
(\cite[Lemma 2.27(ii)]{dp02book}),
we obtain

\begin{lemm}
\label{lem:geom:spaces:frame}
$(\Geom(\At)/M, \leq_M)$
is a frame with greatest element $\class\true_M$, and
\[
\begin{array}{r c l}
  \class{\bigvee_{i \in I} \gamma_i}_M
  \wedge
  \class{\bigvee_{j \in J} \gamma'_j}_M
& =
& \class{\bigvee \left\{
  \gamma_i \land \gamma'_j \mid \text{$i \in I$ and $j \in J$}
  \right\}}_M
\\

  \bigvee_{i \in I} \class{\bigvee_{j \in J_i} \gamma_{i,j}}_M
& =
& \class{\bigvee \left\{
  \gamma_{i,j} \mid \text{$i \in I$ and $j \in J_i$}
  \right\}}_M
\end{array}
\]
\end{lemm}

\begin{theo}
\label{theo:geom:pt}
Let $\th{T}$ be a geometric theory over $\At$.
The function taking 
$\nu \in \Mod(\th T)$
to
$\{ \class\varphi_{\Mod(\th{T})} \mid \nu \models \varphi \}$
is an homeomorphism
from $\Mod(\th T)$ to $\pt(\Geom(\At)/\Mod(\th{T}))$.
\end{theo}

\begin{coro}
\label{cor:geom:pt}
Let $\th{T}$ be a geometric theory.
The space $\Mod(\th{T})$ is sober.
\end{coro}

\paragraph{Subspaces.}
Consider a (sober) space $(X,\Open)$.
Assume that $X$ is represented by a geometric theory $\th{T}$ over $\At$,
in the sense that $X$ is homeomorphic to the space $\Mod(\th{T})$.
Given a subset $P \sle X$, there might be
a theory $\th{U}$ over $\At$ such that
the bijection $X \cong \Mod(\th{T})$ restricts to a bijection
$P \cong \Mod(\th{T} \cup \th{U})$.
In this case, Proposition~\ref{prop:geom:sub} below
implies that the subspace $(P,\Open\restr P)$
is homeomorphic to the space $\Mod(\th{T} \cup \th{U})$,
so that $(P,\Open\restr P)$ is sober
and
$\Open\restr P$ is isomorphic to $\Geom(\At)/\Mod(\th{T} \cup \th U)$. 
In such situations, 
we write
$\Mod_{\th{T}}(\th{U})$ for the space $\Mod(\th{T} \cup \th{U})$.

\begin{prop}
\label{prop:geom:sub}
Given geometric theories $\th{T}$ and $\th{U}$ on $\At$,
the space $\Mod_{\th{T}}(\th{U})$ is
\emph{equal}
to the subspace induced by the inclusion $\Mod(\th{T} \cup \th{U}) \sle \Mod(\th{T})$.
\end{prop}

\begin{exem}
\label{ex:geom:sub}
Let $X$ be an algebraic dcpo, and let $P \deq \up y$ for some fixed $y \in X$.
Given $x \in X$, we have $x \in P$ precisely when $\nu(x)$
is a model of $\th U(y) \deq \{ \thesis d \mid d \leq_X y \}$.
Hence, the subspace induced by $P \sle X$
is homeomorphic to $\Mod_{\th T(X)}(\th U(y))$.

Assume now $X = \I{\Stream\Base}$, and let $a \in \Base$.
Since the subspace induced by the $\LTL$ formula $\lnot\Box a$ is not sober
(Corollary~\ref{cor:frames:sub:notsober}), it follows from
Proposition~\ref{prop:geom:sub} and Corollary~\ref{cor:geom:pt}
that there is \emph{no} geometric theory $\th{T}$ over $\Fin(\I{\Stream\Base})$
such that for all $\stream \in \I{\Stream\Base}$,
we have $\stream \in \I{\lnot \Box a}$ iff $\nu(\stream) \in \Mod(\th{T})$.
\end{exem}
\subsection{Operations on Theories}
\label{sec:geom:op}


We shall now see that for each negation-free formula $\Phi$ of $\LTL$,
there is a geometric theory $\th T\I\Phi$ such that
for all streams $\stream$,
we have $\stream \in \I\Phi$ iff $\nu(\stream) \in \Mod(\th T\I\Phi)$
(where $\nu(\stream)$ is as in
Propositions~\ref{prop:geom:algdcpo} and~\ref{prop:geom:algdcpo:scott}).
This may not be possible if $\Phi$ contains negations (Example~\ref{ex:geom:sub}).

To this end, we devise operations on theories which represent
unions and intersections of sets of models.
Given theories $(\th T_i \mid i \in I)$ over $\At$,
we let
$\bigcurlywedge_{i \in I} \th T_i \deq \bigcup_{i \in I} \th T_i$.
Then we have
\[
\begin{array}{l l l}
  \Mod\left(\bigcurlywedge_{i \in I} \th T_i
  \right)
& =
& \bigcap_{i \in I} \Mod(\th T_i)
\end{array}
\]

\noindent
Intersections of sets of models can thus be represented by unions of theories.
It is more difficult to devise an operation on theories
for \emph{unions} of sets of models.
%
A solution is provided by the following crucial construction.
Let $(\th T_i \mid i \in I)$ be theories, all over $\At$,
with $\th T_i = \{ \psi_{i,j} \thesis \varphi_{i,j} \mid j \in J_i\}$.
\begin{enumerate}[(1)]
\item
If $I$ is finite, we let
\(
  \bigcurlyvee_{i \in I} \th T_i
  \deq
  \left\{
  \bigwedge_{i \in I} \psi_{i,f(i)}
  \thesis
  \bigvee_{i \in I} \varphi_{i,f(i)}
  \mid
  f \in \prod_{i \in I} J_i
  \right\}
\).

\item
If $I$ is infinite, and all $\th T_i$'s are antecedent-free,
\(
  \bigcurlyvee_{i \in I} \th T_i
  \deq
  \left\{
  \thesis
  \bigvee_{i \in I} \varphi_{i,f(i)}
  \mid
  f \in \prod_{i \in I} J_i
  \right\}
\).
\end{enumerate}

Note that if $(\th T_i \mid i \in I)$ consists of countably many countable
(antecedent-free) theories,
then $\bigcurlywedge_{i \in I} \th T_i$ is always countable
while $\bigcurlyvee_{i \in I} \th T_i$
may be uncountable.

\begin{prop}
\label{prop:geom:op}
In both cases above, we have
(using the Axiom of Choice when $I$ is infinite)
\[
\begin{array}{l l l}
  \Mod\left(
  \bigcurlyvee_{i \in I} \th T_i
  \right)
& =
& \bigcup_{i \in I} \Mod(\th T_i)
\end{array}
\]
\end{prop}


\begin{exem}
\label{ex:geom:op}
Let $X$ be an algebraic dcpo, and let $P \sle X$ be upward-closed.
Then $P = \bigcup_{y \in P} \up y$.
Hence, given $x \in X$, we have $x \in P$
exactly when $\nu(x)$ is a model of $\bigcurlyvee_{y \in P}\th U(y)$,
where $\th U(y)$ is as in Example~\ref{ex:geom:sub}.
%
\end{exem}

\subsubsection{Translation of Negation-Free $\LTL$ Formulae}
\label{sec:geom:trans}
\opt{short}{\renewcommand\fntext{In fact, when $\Base$ is countable,
if $P \sle \I{\Stream\Base}$ induces a sober subspace,
then this subspace is homeomorphic to $\Mod_{\th T\I{\Stream\Base}}(\th U)$
for some (abstractly given) theory $\th U$
(see~\cite[\S 5]{rs23full}).}}
\opt{full,long}{\renewcommand\fntext{In fact,
we shall see in Remark~\ref{rem:free:nucleus} (\S\ref{sec:free})
that when $\Base$ is countable,
if $P \sle \I{\Stream\Base}$ induces a sober subspace,
then this subspace is homeomorphic to $\Mod_{\th T\I{\Stream\Base}}(\th U)$
for some (abstractly given) theory $\th U$.}}
Recall from Lemma~\ref{lem:ltl:up} that if $\Phi$ is a negation-free formula
of $\LTL(\Base)$, then $\I\Phi$ is upward-closed in $\I{\Stream\Base}$.
Example~\ref{ex:geom:op} thus provides a geometric theory over $\Fin(\I{\Stream\Base})$
for $\I\Phi$.\fn\ 
But we shall get more information by 
explicitly defining a geometric theory $\th T\I\Phi$
by induction on $\Phi$.
Actually, it is even better to work with a stratified presentation
of negation-free $\LTL$.

Our stratification of negation-free $\LTL$ formulae is based on the
following expected fact.
Recall from~\S\ref{sec:prelim:ltl}
the map $\I\Next$ taking $\SP \in \Po(\I{\Stream\Base})$
to $\{ \stream \mid \stream \restr 1 \in \SP\} \in \Po(\I{\Stream\Base})$.

\begin{lemm}
\label{lem:geom:trans:base}
Fix set $\Base$.
\begin{enumerate}[(1)]
\item
\label{item:geom:trans:base:next}
The function $\I\Next \colon \Po(\I{\Stream\Base}) \to \Po(\I{\Stream\Base})$
preserves all unions and all intersections.

\item
\label{item:geom:trans:base:fix}
Given $\LTL$ formulae $\Phi, \Psi$,
let $H_{\Phi,\Psi}$
take $\SP \in \Po(\I{\Stream\Base})$ to $\I\Psi \cup (\I\Phi \cap \I\Next(\SP))$.
Then
we have
$\I{\Phi \Ushort \Psi} = \bigcup_{n \in \NN} H^n_{\Phi,\Psi}(\I\False)$
and
$\I{\Phi \Wshort \Psi} = \bigcap_{n \in \NN} H^n_{\Phi,\Psi}(\I\True)$.

%
\end{enumerate}
\end{lemm}
\begin{figure}[!t]
\[
\begin{array}{r r l @{~~}c@{~~} l @{~~}c@{~~} l @{~~}c@{~~} l @{~~}c@{~~} l}
  \Phi_0,\Psi_0
& \bnf
& \False
& |
& \True
& |
& a
\\
& |
& \Phi_0 \lor \Psi_0
& |
& \Phi_0 \land \Psi_0
& |
& \Next \Phi_0
\\
\\

  G(\Base) \ni
  \Phi_1,\Psi_1
& \bnf
& \Phi_0
\\
& |
& \Phi_1 \lor \Psi_1
& |
& \Phi_1 \land \Psi_1
& |
& \Next \Phi_1
& |
& \Phi_1 \Ushort \Psi_1
\\
\\

  G_\delta(\Base) \ni
  \Phi_2,\Psi_2
& \bnf
& 
  \multicolumn{3}{l}{\Phi_1 \in G(\Base)}
\\
& |
& \Phi_2 \lor \Psi_2
& |
& \Phi_2 \land \Psi_2
& |
& \Next \Phi_2
& |
& \Phi_2 \Wshort \Psi_2
\\
\\

  \text{n.-f.\ $\LTL(\Base)$} \ni
  \Phi_3,\Psi_3
& \bnf
& 
  \multicolumn{3}{l}{\Phi_2 \in G_\delta(\Base)}
\\
& |
& \Phi_3 \lor \Psi_3
& |
& \Phi_3 \land \Psi_3
& |
& \Next \Phi_3
& |
& \Phi_3 \Ushort \Psi_3
& |
& \Phi_3 \Wshort \Psi_3
\end{array}
\]
\caption{Stratified grammar for negation-free $\LTL(\Base)$, where $a \in \Base$.}
\label{fig:geom:trans}
\end{figure}

Figure~\ref{fig:geom:trans} presents a
stratified grammar for negation-free $\LTL(\Base)$.
We let $G = G(\Base)$ be the set of all formulae $\Phi_1,\Psi_1$
from the second layer in Figure~\ref{fig:geom:trans}.
$G_\delta= G_\delta(\Base)$ consists of formulae $\Phi_2,\Psi_2$
from the third layer.
The negation-free $\LTL(\Base)$ formulae are the
those
from the last layer.

\begin{exem}
\label{ex:geom:trans:class}
Recall from~\S\ref{sec:prelim:ltl} that
$\Diam \Psi = (\True \Ushort \Psi)$
and
$\Box \Phi = (\Phi \Wshort \False)$.
Hence, $G$ is closed under $\Diam(-)$ and $G_\delta$ is closed under $\Box(-)$.
But $G_\delta$ is (crucially) \emph{not} closed under $\Diam(-)$.
In particular, we have $\Next a, \Diam a \in G$ 
and $\Box a, \Box \Diam a \in G_\delta$.
On the other hand, the negation-free formula $\Diam \Box a$ is not a $G_\delta$
formula.

When looking at $\Diam \Psi$ and $\Box \Phi$
via Lemma~\ref{lem:geom:trans:base}(\ref{item:geom:trans:base:fix}),
it is convenient
to simplify the functions $H_{\True,\Psi}$ and $H_{\Phi,\False}$
to 
respectively
$\I{\Psi} \cup \I\Next(-)$ and $\I\Phi \cap \I\Next(-)$.
This amounts to
restate Lemma~\ref{lem:geom:trans:base}(\ref{item:geom:trans:base:fix})
as
\(
  \I{\Diam\Psi}
  =
  \bigcup_{m \in \NN} \I\Psi \cup \I{\Next\Psi} \cup \dots \cup \I{\Next^m\Psi}
\),
and
similarly for $\I{\Box\Phi}$.
%
%

%
\end{exem}

\begin{remark}
\label{rem:geom:trans:class}
The interpretations of
formulae from $G$ or $G_\delta$ have the expected topological complexity.
Namely,
if $\Phi_1 \in G$, then $\I{\Phi_1}$ is Scott-open in $\I{\Stream\Base}$.
If $\Phi_2 \in G_\delta$, then $\I{\Phi_2}$ is a countable intersection
of Scott-opens (i.e.\ a $G_\delta$ subset of $\I{\Stream\Base}$).
%
\end{remark}

This stratification of negation-free $\LTL$ allows for a stratified
translation to geometric theories.
In fact, each $\LTL$ formula $\Phi_1 \in G$ can be translated
to a single geometric formula $\th F\I{\Phi_1}$,
with $\th F\I{\Phi_0}$ finite when $\Phi_0$ is from the first layer.
Formulae $\Phi$ from the last two layers will be translated to
antecedent-free theories $\th T\I{\Phi}$,
with $\th T\I{\Phi_2}$ countable when $\Phi_2 \in G_\delta$.

Fix a set $\Base$ and let $\At \deq \Fin(\I{\Stream\Base})$.
We devise operations on geometric formulae and theories
which mimic the action of $\I\Next(-)$ on $\Po(\I{\Stream\Base})$.
%
%
We begin with geometric formulae.
The idea is that given $\stream \in \I{\Stream\Base}$
and $d \in \Fin(\I{\Stream\Base})$,
we have $d \leq_{\I{\Stream\Base}} \stream\restr 1$
exactly when $(\bot \cdot d) \leq_{\I{\Stream\Base}} \stream$.
The geometric formula $\Next\varphi$
is then defined by propagating the stream operation $d \mapsto \bot \cdot d$
in $\varphi$.
We set
$\Next d \deq \bot \cdot d$ and
\[
\begin{array}{l !{\deq} l !{\qquad} l !{\deq} l !{\qquad} l !{\deq} l}
  \Next\true
& \true

& \Next (\gamma \land \gamma')
& (\Next \gamma) \land (\Next \gamma')

& \Next \bigvee_{i \in I} \gamma_i
& \bigvee_{i \in I} \Next \gamma_i
\end{array}
\]



\noindent
Given a theory $\th{Th}$ over $\At$,
we let
\(
  \Next \th{Th}
  \deq
  \left\{
  \Next \psi \thesis \Next \varphi
  \mid (\psi \thesis \varphi) \in \th{Th}
  \right\}
\).
Note that $\Next \th{Th}$ is antecedent-free whenever so is $\th{Th}$.
Recall the map $x \mapsto \nu(x)$ of
Propositions~\ref{prop:geom:algdcpo} and~\ref{prop:geom:algdcpo:scott}.

\begin{lemm}
\label{lem:geom:trans:nextgeom}
Let $\stream \in \I{\Stream\Base}$.
\begin{enumerate}[(1)]
\item
We have
$\nu(\stream) \models \Next \varphi$ if, and only if,
$\nu(\stream \restr 1) \models \varphi$.

\item
We have
$\nu(\stream) \in \Mod(\Next\th{Th})$ if, and only if,
$\nu(\stream \restr 1) \in \Mod(\th{Th})$.
\end{enumerate}
\end{lemm}

We now define a geometric formula $\th F\I{\Phi_1}$ over $\At$
by induction on $\Phi_1 \in G$:
\[
\begin{array}{l !{\deq} l !{\qquad} l !{\deq} l !{\qquad}}
  \th F\I{a}
& a \cdot \bot^\omega
& \multicolumn{2}{l}{}
\\

  \th F\I{\True}
& \true

& \th F\I{\Phi_1 \land \Psi_1}
& \th F\I{\Phi_1} \land \th F\I{\Psi_1}
\\

  \th F\I{\False}
& \false

& \th F\I{\Phi_1 \lor \Psi_1}
& \th F\I{\Phi_1} \lor \th F\I{\Psi_1}
\\

  \th F\I{\Next\Phi_1}
& \Next \th F\I{\Phi_1}

& \th F\I{\Phi_1 \Ushort \Psi_1}
& \bigvee_{n \in \NN} \th H^n_{\th F\I{\Phi_1},\th F\I{\Psi_1}}(\false)
\end{array}
\]

%
%
%
%
%

\noindent
where $\th H_{\varphi,\psi}(\theta) \deq \psi \lor (\varphi \land \Next(\theta))$.
(We silently included the case of $\Phi_0$ from the first layer.)

\begin{lemm}
\label{lem:geom:trans:cor}
Let $\Phi_1 \in G$.
Given $\stream \in \I{\Stream\Base}$,
we have
$\stream \in \I{\Phi_1}$
if, and only if,
$\nu(\stream) \models \th F\I{\Phi_1}$.
\end{lemm}

Finally,
the antecedent-free theory $\th T\I{\Phi_3}$
is defined by induction on $\Phi_3$ as follows:
\[
\begin{array}{l !{\deq} l !{\qquad} l !{\deq } l}
  \th T\I{\Phi_1}
& 
  \{\thesis \th F\I{\Phi_1} \}

& \th T\I{\Next\Phi_3}
& \Next \th T\I{\Phi_3}
\\

  \th T\I{\Phi_3 \land \Psi_3}
& \th T(\I{\Phi_3}) \curlywedge \th T(\I{\Psi_3})

& \th T\I{\Phi_3 \Wshort \Psi_3}
& \bigcurlywedge_{n \in \NN}
  \th{TH}^n_{\th T\I{\Phi_3},\th T\I{\Psi_3}}(\{\thesis \true \})
\\

  \th T\I{\Phi_3 \lor \Psi_3}
& \th T(\I{\Phi_3}) \curlyvee \th T(\I{\Psi_3})

& \th T\I{\Phi_3 \Ushort \Psi_3}
& \bigcurlyvee_{n \in \NN}
  \th{TH}^n_{\th T\I{\Phi_3},\th T\I{\Psi_3}}(\{\thesis \false \})
\end{array}
\]

%
%
%
%

\noindent
where
$\th{TH}_{\th T, \th U}(\th V) \deq \th U \curlyvee (\th T \curlywedge \Next \th V)$.
(We silently included the case of $\Phi_2 \in G_\delta$.)

\begin{theo}
\label{thm:geom:trans:cor}
Let $\Phi$ be negation-free.
For $\stream \in \I{\Stream\Base}$,
we have
$\stream \in \I{\Phi}$
if, and only if,
$\nu(\stream) \in \Mod(\th T\I{\Phi})$.
\end{theo}

\begin{remark}
\label{rem:geom:trans:card}
A direct inspection reveals that $\th F\I{\Phi_0}$ is indeed
a finite geometric formula when $\Phi_0$ is from the first layer.
Similarly, the geometric theory $\th T\I{\Phi_2}$ contains only countably-many
sequents when $\Phi_2 \in G_\delta$.
%
\end{remark}

\begin{remark}
\label{rem:geom:trans:equiv}
Recall that $\LTL$ formulae $\Phi,\Psi$ are \emph{equivalent},
notation $\Phi \equiv \Psi$, when $\I\Phi = \I\Psi$.
The following standard equivalences are obtained
similarly as in~\cite[\S 5.1.4]{bk08book}.
%
\[
\begin{array}{l !{\equiv} l !{\qquad} l !{\equiv} l !{\qquad} l !{\equiv} l}
  \Next\False
& \False

& \Next(\Phi \lor \Psi)
& \Next\Phi \lor \Next\Psi

& \Next(\Phi \Ushort \Psi)
& (\Next\Phi) \Ushort (\Next\Psi)
\\

  \Next\True
& \True

& \Next(\Phi \land \Psi)
& \Next\Phi \land \Next \Psi

& \Next(\Phi \Wshort \Psi)
& (\Next\Phi) \Wshort (\Next\Psi)
\end{array}
\]

\noindent
Hence, up to equivalence, we can push
the $\Next$'s to atoms $a \in \Base$.
In particular, we may assume that $\Next$
occurs only in first layer's formulae of the form $\Next^n a$.
\end{remark}

\begin{exem}
\label{ex:geom:trans:infty}
Let $\Phi_0$ be an $\LTL$ formula from the first layer in Figure~\ref{fig:geom:trans}.
Up to equivalence (Remark~\ref{rem:geom:trans:equiv}),
we can assume that $\Phi_0$ is in disjunctive normal form,
and actually that $\Phi_0$ is a disjunction of conjunctions
of formulae of the form $\Next^n a$.
Then $\th F\I{\Phi_0}$ is simply a disjunction of conjunctions
of atomic propositions of the form
$(\bot^n \cdot a \cdot \bot^\omega) \in \Fin(\I{\Stream\Base})$.

Consider the formula $\Phi_1 \deq \Diam\Phi_0 \in G$.
Recall that $\Diam\Phi_0 = (\True \Ushort \Phi_0)$,
and note that
$\th H_{\true,\varphi}(\theta) = \varphi \lor (\true \land \Next\theta)$.
Up to the replacement of $\true \land \Next\theta$
by $\Next\theta$,
we get that $\th F\I{\Phi_1}$ is the geometric formula
\(
  \bigvee_{m \in \NN}
  \left(
  \th F\I{\Phi_0}
  \lor
  \Next \th F\I{\Phi_0} 
  \lor
  \dots
  \lor
  \Next^m \th F\I{\Phi_0} 
  \right)
\).
This mirrors
the formulation of
Lemma~\ref{lem:geom:trans:base}(\ref{item:geom:trans:base:fix}) in
Example~\ref{ex:geom:trans:class}.

We turn to the formula $\Phi_2 \deq \Box\Diam\Phi_0 \in G_\delta$.
We have $\Box\Phi_1 = (\Phi_1 \Wshort \False)$
and we simplify $\th H_{\varphi,\false}(\theta)$ to $\varphi \land \Next\theta$.
The theory $\th T\I{\Phi_2}$ then consists of all 
the sequents
\[
\begin{array}{l l}
  \thesis
& \bigwedge_{n \leq N}
  \bigvee_{m \in \NN}
  \Next^n \th F\I{\Phi_0}
  \lor
  \Next^{n+1} \th F\I{\Phi_0}
  \lor
  \dots
  \lor
  \Next^{n+m} \th F\I{\Phi_0}
\end{array}
\]


\noindent
where $N$ ranges over $\NN$.
This mirrors the fact that
$\I{\Box\Diam\Phi_0}$ is the set of those streams $\stream \in \I{\Stream\Base}$
such that for each $n \in \NN$, there is some $m \in \NN$
with $\stream\restr(n+m) \in \I{\Phi_0}$.
\end{exem}

\begin{exem}
\label{ex:geom:trans:uall}
Continuing Example~\ref{ex:geom:trans:infty},
we now consider the case of $\Phi_3 \deq \Diam\Box a$ with $a \in \Base$.
For this more involved example,
we allow ourselves some simplifications that we deliberately avoided
in Example~\ref{ex:geom:trans:infty}.
Namely, for $\th T\I{\Box\Phi}$ and $\th T\I{\Diam\Psi}$
we take respectively
$\bigcurlywedge_{n \in \NN} \Next^n \th T\I\Phi$
and
$\bigcurlyvee_{m \in \NN} \Next^m \th T\I\Psi$.
Then we have
\[
\begin{array}{r c l}
  \th T\I{\Box a}
& =
& \left\{
  \Next^n \th F\I{a}
  \mid n \in \NN
  \right\}
\\

  \th T\I{\Diam\Box a}
& =
& \left\{
  \bigvee_{m \in \NN}
  \Next^{m+f(m)} \th F\I{a}
  \mid f \colon \NN \to \NN
  \right\}
\end{array}
\]

\noindent
The uncountable theory $\th T\I{\Diam\Box a}$ relies on
the (classical) choice principle behind Proposition~\ref{prop:geom:op}.
It expresses that given
a stream $\stream \in \I{\Stream\Base}$, we have $\stream \notin \I{\Diam\Box a}$
if, and only if,
there exists a function $f \colon \NN \to \NN$
such that $\stream(m+f(m)) \neq a$ for all $m \in \NN$.
In particular, if $\stream \notin \I{\Diam\Box a}$,
then the function $g \colon m \mapsto m+f(m)$ finds arbitrary large
$n = g(m)$ such that $\stream(n) \neq a$.
\end{exem}

Consider a negation-free $\Phi \in \LTL(\Base)$.
We see the subset $\I\Phi \sle \I{\Stream\Base}$
as a subspace rather than as a sub-dcpo
(cf.\ Example~\ref{ex:topo:omegaword}).
%
This subspace 
$\I\Phi = (\I\Phi, \Open(\I{\Stream\Base})\restr \I\Phi)$
is always sober (Corollary~\ref{cor:frames:sub:upsober}).
With geometric logic, we gained a description of
the subspace $\I\Phi$
as the space of models of the theory $\th T\I{\Phi}$.
Namely, the subspace $\I\Phi$
is homeomorphic to
$\Mod_{\I{\Stream\Base}}(\th T\I\Phi)$
(Theorem~\ref{thm:geom:trans:cor} and Proposition~\ref{prop:geom:sub}),
while the frame
$\Open(\I{\Stream\Base}) \restr \I\Phi$
is isomorphic to
$\Geom(\At)/\Mod_{\I{\Stream\Base}}(\th T\I\Phi)$
(Theorem~\ref{theo:geom:pt}).
Note that in the latter, two geometric formulae are equivalent
exactly when they have the same $\th T\I\Phi$-models.
Hence geometric formulae generate the frame  
$\Open(\I{\Stream\Base}) \restr \I\Phi$.
Moreover, we have seen in 
Examples~\ref{ex:geom:trans:infty} and~\ref{ex:geom:trans:uall}
concrete cases in which the theory $\th T\I\Phi$
explicitly represents approximations of $\I\Phi$.


However, a limitation of this approach is that the frame 
$\Geom(\At)/\Mod_{\I{\Stream\Base}}(\th T\I\Phi)$
is defined by purely semantic means.
\opt{full,long}{We discuss this in~\ref{sec:free} below, }%
\opt{short}{We discuss this in~\cite[\S 5]{rs23full}, }%
which gives a complete deduction system for $\th T\I\Phi$ in the $G_\delta$ case.
We now comment on potential extensions to $\LTL$ with negation.

\begin{remark}
Say that an $\LTL$ formula $\Phi$ is an \emph{$F$ formula}
if $\Phi$ is the negation of a $G$ formula.
The \emph{$F_\sigma$} formulae are the negations of the $G_\delta$ ones.
For instance, $\lnot a$ (with $a \in \Base$) is a simple non-trivial $F$ formula,
while $\lnot \Box a \equiv \Diam \lnot a$ is an $F_\sigma$ formula.
It follows from Remark~\ref{rem:geom:trans:class}
that $F$ formulae induce Scott-\emph{closed} subsets of $\I{\Stream\Base}$,
and that the $F_\sigma$ ones induce countable unions of Scott-closed
sets (i.e.\ $F_\sigma$ sets).

Now, if $\Phi = \lnot \Phi_1$ with $\Phi_1 \in G$,
then the subspace $\I\Phi$ is represented by the geometric theory
$\{ \th F\I{\Phi_1} \thesis \false \}$.
Hence Theorem~\ref{thm:geom:trans:cor} extends to $F$ formulae
(so that Proposition~\ref{prop:geom:sub} and
Theorem~\ref{theo:geom:pt} can be applied in this case).
But beware that this 
does not hold in general
for $F_\sigma$ formulae,
since the subspace $\I{\lnot \Box a} = \I{\Diam \lnot a}$ is not representable 
in geometric logic (in the sense of Example~\ref{ex:geom:sub}).
In particular, there is no geometric theory $\th{T}$
such that
$\Mod(\th{T}) = \bigcup_{m \in \NN} \Mod\{ \Next^m \th F\I{a} \thesis \false \}$,
and Proposition~\ref{prop:geom:op} does not extend to infinitely many
arbitrary theories.
\end{remark}

\opt{full}{
\section{Free Frames and Spatiality}
\label{sec:free}

Our approach to geometric logic in~\S\ref{sec:geom}
focuses on \emph{spaces} of \emph{models}.
However, the literature rather considers geometric logic
as a formal way to present frames by \emph{generators and relations}.
A customary tool for this is the notion of \emph{congruence preorder}
(see~\cite[\S 4 and \S 6.1--2]{vickers89book}
and~\cite[\S 3]{heckmann15mscs}).
%
It is folklore
that congruence preorders can be presented using an (infinitary)
deduction system for geometric logic
(in the spirit of e.g.~\cite[\S 2.2]{vickers07chapter}
and~\cite[D1.1.7(m) and \S D1.3]{johnstone02book}).

\begin{figure}[t]
\[
\begin{array}{c}
\rn{Th}~
\dfrac{}
  {\varphi \thesis_{\th T} \psi}
~\text{(if $(\varphi \thesis \psi) \in \th T$)}

\qquad\quad

\rn{Ax}~
\dfrac{}
  {\varphi \thesis_{\th T} \varphi}

\qquad\quad

\rn{Cut}~
\dfrac{\varphi \thesis_{\th T} \theta
  \qquad
  \theta \thesis_{\th T} \psi}
  {\varphi \thesis_{\th T} \psi}

\\\\

\rn{$\land$-L$_1$}~
\dfrac{}
  {\varphi \land \psi \thesis_{\th T} \varphi}

\qquad\qquad

\rn{$\land$-L$_2$}~
\dfrac{}
  {\varphi \land \psi \thesis_{\th T} \psi}

\qquad\qquad

%
%
%

\rn{$\land$-R}~
\dfrac{\theta \thesis_{\th T} \varphi
  \qquad
  \theta \thesis_{\th T} \psi}
  {\theta \thesis_{\th T} \varphi \land \psi}

\\\\

\rn{$\true$-R}~
\dfrac{}
  {\varphi \thesis_{\th{T}} \true}

\qquad\quad

\rn{$\bigvee$-L}~
\dfrac{\text{for all $i \in I$, $\varphi_i \thesis_{\th T} \psi$}}
  {\bigvee_{i \in I}\varphi_i \thesis_{\th T} \psi}

\qquad\quad

\rn{$\bigvee$-R}~
\dfrac{}
  {\varphi_i \thesis_{\th T} \bigvee_{i \in I} \varphi_i}

%
\end{array}
\]
\caption{Deduction rules for geometric logic over a theory $\th T$.%
\label{fig:free}}
\end{figure}

Let $\th{T}$ be a geometric theory over $\At$.
The \emph{deduction relation} $\thesis_{\th{T}}$ on $\Geom(\At)$
is defined by the rules in Figure~\ref{fig:free},
using the constructs of Remark~\ref{rem:geom:conn}.
The relation $\thesis_{\th{T}}$ is a preorder.
Similarly as in~\S\ref{sec:geom:spaces},
we quotient $\Geom(\At)$ under the relation $\sim_{\th{T}}$
of \emph{$\th{T}$-equivalence} defined as
$\varphi \sim_{\th{T}} \psi$ iff $\varphi \thesis_{\th{T}} \psi$
and $\psi \thesis_{\th{T}} \varphi$.
We write $\class\varphi_{\th{T}}$ for the $\sim_{\th{T}}$-class of $\varphi$
and $\Geom(\At)/\th{T}$ for the set of $\sim_{\th{T}}$-classes.
We let
$\leq_{\th T}$ be the partial order on $\Geom(\At)/\th{T}$
induced by $\thesis_{\th{T}}$.
In contrast with~\cite[\S D1.3 and D1.4.14]{johnstone02book}
(see also~\cite[\S 2.2]{vickers07chapter}),
we do not need to enforce frame distributivity in Figure~\ref{fig:free},
since it is hardwired in the constructs of Remark~\ref{rem:geom:conn}.



\begin{lemm}
\label{lem:free:frame}
$(\Geom(\At)/\th{T}, \leq_{\th{T}})$ is a frame.
\end{lemm}

\paragraph{Spatiality.}
How does $\Geom(\At)/\th{T}$ relate to the frame
$\Geom(\At)/\Mod(\th{T})$ of~\S\ref{sec:geom:spaces}?
To answer this question, we compare the preorders $\thesis_{\th{T}}$
and $\preceq_{\Mod(\th{T})}$.
First, an easy induction on derivations proves the following
\emph{soundness} property:
\[
\begin{array}{l l l}
  \psi \thesis_{\th{T}} \varphi
& \longimp
& \psi \preceq_{\Mod(\th{T})} \varphi
\end{array}
\]

The converse implication is a form of \emph{completeness}:
if every model of $\th{T}$ is a model of
the sequent
$\psi \thesis \varphi$,
then $\psi \thesis_{\th{T}} \varphi$ is derivable.
This may fail in geometric logic, by lack of spatiality.
%
%
A frame $L$ is said to be \emph{spatial} when given $a \not\leq b$ in $L$,
there is a point $\Filt \in \pt(L)$ such that $a \in \Filt$ and $b \notin \Filt$.
A frame is spatial precisely when it is isomorphic to 
$\Open(X)$ for some space $X$, but not every frame is spatial
(\cite[\S II.1.5]{johnstone82book}).


Assume now that $\Geom(\At)/\th{T}$ is spatial,
and let $\psi \not\thesis_{\th{T}} \varphi$.
Hence $\class\psi_{\th{T}} \not\leq_{\th{T}} \class\varphi_{\th{T}}$
and there is a point $\Filt$
of $\Geom(\At)/\th{T}$
such that $\class\psi_{\th{T}} \in \Filt$ and $\class\varphi_{\th{T}} \notin \Filt$.
Let $\nu \colon \At \to \two$ take $p \in \At$ to $1$ iff
$\class{p}_{\th{T}} \in \Filt$.

\begin{lemm}
\label{lem:free:filtmod}
Let $\Filt$ and $\nu$ as above.
Then for every $\theta \in \Geom(\At)$,
we have $\nu \models \theta$ if, and only if, $\class\theta_{\th{T}} \in \Filt$.
In particular, $\nu$ is a model of $\th{T}$ with
$\nu \models \psi$ and $\nu \not\models \varphi$.
\end{lemm}

Hence, when $\Geom(\At)/\th{T}$ is spatial,
the preorders $\thesis_{\th{T}}$ and $\preceq_{\Mod(\th{T})}$ coincide,
and $\Geom(\At)/\th{T}$ is thus isomorphic
to $\Geom(\At)/\Mod(\th{T})$.
Since $\Geom(\At)/\Mod(\th{T})$ is always isomorphic
to the frame of opens of the space $\Mod(\th{T})$
(Theorem~\ref{theo:geom:pt}),
we have the following.

\begin{theo}
\label{thm:free:equiv}
Let $\th{T}$ be a geometric theory over $\At$.
Then the following are equivalent.
\begin{enumerate}[(i)]
\item
The frame $\Geom(\At)/\th{T}$ is spatial.

\item
Deduction in $\thesis_{\th{T}}$ is complete
($\psi \thesis_{\th{T}} \varphi$ if, and only if,
$\psi \preceq_{\Mod(\th{T})} \varphi$).

\item
The frame $\Geom(\At)/\th{T}$ is isomorphic to
$\Geom(\At)/\Mod(\th{T})$.
\end{enumerate}
\end{theo}

\paragraph{Free Frames and Congruence Preorders.}
In view of Theorem~\ref{thm:free:equiv},
it is interesting to know when a frame $\Geom(\At)/\th{T}$ is spatial.
We discuss this using the following notions.
Write $\th{E}$ for the empty theory.

Let $\Forget \colon \Frm \to \Set$ be the forgetful functor.
A \emph{free frame} on $\At$ is the data of a frame $L(\At)$
and of a function $\eta_\At \colon  \At \to \Forget L(\At)$
such that
for each frame $L$
and each function $f \colon \At \to \Forget L$,
there is a unique frame morphism
$\ladj f \colon L(\At) \to L$
making the following diagram to commute.
\[
\begin{tikzcd}[row sep=tiny, column sep=large]
  \At
  \arrow{rr}{f}
  \arrow{dr}[below]{\eta_{\At}}
&
& \Forget L
\\
& \Forget L(\At)
  \arrow[dashed]{ur}[below]{\Forget \ladj f}
\end{tikzcd}
\]

\begin{prop}
\label{prop:free:frame:free}
$(\Geom(\At)/\th{E},\leq_{\th{E}})$
(with the function $p \in \At \mapsto \class{p}_{\th{T}}$)
is a free frame on $\At$.
\end{prop}


\noindent
The operation $\At \mapsto \Geom(\At)/\th{E}$
thus yields a left adjoint to $\Forget \colon \Frm \to \Set$
(see e.g.~\cite[Theorem IV.1.2]{maclane98book}).
Actually,
the set of Scott-opens
of the cpo $\Po(\At)$ is also a free frame on $\At$.
See~\cite[Lemma VII.4.9]{johnstone82book}
(see also~\cite[Theorem 3.1]{heckmann15mscs}
and~\cite[C1.1.4 and C4.1.6]{johnstone02book}).
Since free frames are unique up to isomorphism,
it follows that the frame $\Geom(\At)/\th{E}$ is spatial.

%

It is customary to present a frame by quotienting a free frame
under a congruence preorder.
We refer to~\cite[\S 3.4]{heckmann15mscs}.
A \emph{congruence preorder} on a frame $(L,\leq)$ is a preorder
$\preceq$ on $L$ such that ${\leq} \sle {\preceq}$ and
such that
for each arbitrary (resp.\ finite) $\SP \sle L$,
we have $\bigvee \SP \preceq b$ whenever $a \preceq b$ for all $a \in \SP$
(resp.\ $b \preceq \bigwedge \SP$ whenever $b \preceq a$ for all $a \in \SP$).
%
Given a congruence preorder $\preceq$ on $L$,
let $a \sim b$ iff $a \preceq b$ and $b \preceq a$.

Each binary relation $R$ on $L$ is contained in a least congruence preorder
$\clos R$ on $L$.


\begin{prop}
\label{prop:free:quot}
Given geometric theories $\th{T}, \th{U}$ over $\At$,
let
\(
  R
  \deq
  \{ (\class\varphi_{\th{T}} , \class\psi_{\th{T}}) 
  \mid (\varphi \thesis \psi) \in \th{T} \cup \th{U}\}
\).
Then $\Geom(\At)/(\th{T} \cup \th{U})$
is isomorphic to the quotient of $\Geom(\At)/\th{T}$
by $\sim_{\clos R}$.
\end{prop}

In particular, $\Geom(\At)/\th{T}$
is isomorphic to the quotient of the free frame $\Geom(\At)/\th{E}$
under $\sim_{\clos R}$,
where
\(
  R
  =
  \{ (\class\varphi_{\th{E}} , \class\psi_{\th{E}}) 
  \mid (\varphi \thesis \psi) \in \th{T}\}
\).

\begin{exem}[{\cite[D1.1.7(m) and D1.4.14]{johnstone02book}}]
\label{ex:free:frame}
Each frame $(L,\leq_L)$
is isomorphic to the frame
$\Geom(\At)/\th T(L)$,
where $\th T(L)$ is
the following 
theory over
$\At \deq \{ \underline a \mid a \in L\}$.
\[
\begin{array}{r !{\thesis} l @{\quad} l !{\qquad} r !{\thesis} l}
  \underline{a}
& \underline{b}
& \text{(if $a \leq_L b$)}

&
& \underline{\top_L}
\\

  \underline{\bigvee_L \SP}
& \bigvee \left\{ \underline{a} \mid a \in \SP\right\}
& \text{(if $\SP \sle L$)}

& \underline{a} \land \underline{b}
& \underline{a \land_L b}

\end{array}
\]
\end{exem}


Hence, any frame is a quotient of a free frame
by a geometric theory.
It then follows
from~\cite[\S II.1.5]{johnstone82book}
that $\Geom(\At)/\th{T}$ may not be spatial.
On the positive side, we have

\begin{theo}[{\cite[Corollary 3.15]{heckmann15mscs}}]
Let $R$ be a \emph{countable} binary relation on a free frame $L(\At)$.
Then the quotient of $L(\At)$ under $\sim_{\clos R}$ is a spatial frame.
\end{theo}

\begin{coro}
\label{cor:free:count}
Let $\th{T}$ be a geometric theory over $\At$.
If $\th{T}$ is countable, then the frame $\Geom(\At)/\th{T}$ is spatial.
In particular, if $\Phi$ is a $G_\delta$ formula of $\LTL(\Base)$,
then $\Geom(\At)/\th{T}\I\Phi$ is spatial,
where $\At = \Fin(\I{\Stream\Base})$.
\end{coro}

Hence, when $\Phi$ is a $G_\delta$ formula of $\LTL$,
deduction in $\thesis_{\th T\I\Phi}$ completely axiomatizes
$\Mod(\th T\I\Psi)$, and the frame $\Geom(\At)/\th{T}\I\Phi$ is
isomorphic to $\Geom(\At)/\Mod(\th{T}\I\Phi)$.

We must warn the reader on the following subtle points,
which actually motivate the explicit constructions of~\S\ref{sec:geom:trans}.

\begin{remark}
\label{rem:free:nucleus}
Consider a space $(X,\Open)$.
It follows from Example~\ref{ex:free:frame}
that there is a 
theory $\th T_\Open$ over $\At_\Open \deq \{ \underline U \mid U \in \Open \}$
such that
$\Open$ is isomorphic to $\Geom(\At_\Open)/\th T_\Open$
(so that $\Geom(\At_\Open)/\th T_\Open$ is spatial).
But beware that when $X$ is an algebraic dcpo, the theory
$\th T \deq \th T(X)$ of~\S\ref{sec:geom:th}
is over the set of atomic propositions $\At \deq \Fin(X)$,
which differs from $\At_\Open$.
In particular, the theories $\th T$ and $\th T_\Open$ differ,
and there is \emph{a priori} no reason for $\Open$ to be
isomorphic to $\Geom(\At)/\th T$.

An algebraic dcpo $X$ is said to be \emph{$\omega$-algebraic}
when $\Fin(X)$ is (at most) countable.
For instance, $\I{\Stream\Base}$ is $\omega$-algebraic
precisely when $\Base$ is (at most) countable.

Assume $X$ is $\omega$-algebraic.
Set $\At \deq \Fin(X)$ and $\th T \deq \th T(X)$ as above.
The theory $\th T$ is countable, and Corollary~\ref{cor:free:count} applies.
Hence, 
the frame $\Geom(\At)/\th T$
is spatial and thus isomorphic to the frame $\Open$
(Theorem~\ref{thm:free:equiv},
Theorem~\ref{theo:geom:pt}
and Proposition~\ref{prop:geom:algdcpo:scott}).

Consider now a subset $P \sle X$.
Recall from~\S\ref{sec:frames:sub} that the 
quotient frame
$\Open \quot \Open\restr P$
can be represented as the frame of $j$-fixpoints
for a nucleus $j$ on $\Open$.
It is then a consequence of 
Proposition~\ref{prop:free:quot}
and~\cite[\S 3.4]{heckmann15mscs}
that the frame 
$\Open\restr P$ is isomorphic to
$\Geom(\At)/(\th T \cup \th U)$,
where
\[
\begin{array}{l l l}
  \th U
& \deq
& \left\{
  \varphi \thesis \psi
  \mid
  \class\varphi_{\th T} \leq_{\th T} j(\class\psi_{\th T})
  \right\}
\end{array}
\]

\noindent
In particular, the frame $\Geom(\At)/(\th T \cup \th U)$ spatial
and thus isomorphic to $\Geom(\At)/\Mod_{\th T}(\th U)$.
But beware that this does \emph{not} imply that the \emph{space}
$(P,\Open\restr P)$ is represented by the space of models
$\Mod_{\th{T}}(\th U)$, unless $(P,\Open\restr P)$ is sober,
since in this case we have
\[
  (P,\Open\restr P)
  ~\cong~
  \pt(\Open\restr P)
  ~\cong~
  \pt(\Geom(\At)/(\th T \cup \th U))
  ~\cong~
  \pt(\Geom(\At)/\Mod_{\th T}(\th U))
  ~\cong~
  \Mod_{\th{T}}(\th U)
\]

In the case of streams $\I{\Stream\Base}$ (with $\Base$ countable),
it follows that
for \emph{any} $P \sle \I{\Stream\Base}$
there is a geometric theory $\th U$
on $\At$ which represents the \emph{frame} $\Open\restr P$.
For instance, with $P \deq \I{\lnot \Box a}$,
the isomorphism $\Open\restr P \cong \Geom(\At)/(\th T \cup \th U)$
lifts to homeomorphisms
\[
  \pt(\Open\restr P)
  ~\cong~
  \pt(\Geom(\At)/(\th T \cup \th U))
  ~\cong~
  \pt(\Geom(\At)/\Mod_{\th T}(\th U))
  ~\cong~
  \Mod_{\th{T}}(\th U)
\]

\noindent
But we have
$(P,\Open\restr P) \not\cong \Mod_{\th T}(\th U)$
since
$(P,\Open\restr P) \not\cong \pt(\Open\restr P)$
as $(P,\Open\restr P)$ is not sober
(Corollary~\ref{cor:frames:sub:notsober}).
\end{remark}

}	
\section{A Specification for the Denotation of Filter}
\label{sec:scott}

The core of this paper consists of the results presented
above
concerning $\LTL$ on streams.
%
However, the long term goal of this work is to reason on
input-output (negation-free) $\LTL$ properties of functions.
We now briefly sketch how our results can help
to handle our motivating example,
namely the $\filter$ function on streams.
This is mostly preliminary;
we leave as future work the elaboration of a general solution.

We work with the 
function $\I\filter$
of Remark~\ref{rem:prelim:filter}.
Fix a finite set $\Base$ and a Scott-continuous function
$p \colon \I\Base \to \I\Bool$.
Assume that for all $a \in \Base$, we have $p(a) \neq \bot_{\I\Bool}$.
Let $\Psi = \Psi_p$ as in Example~\ref{ex:ltl:nf}
and set $\Phi \deq \bigvee_{a \in \Base} a$.
The specification~\eqref{eq:intro:spec} for $\filter$
leads to the following specification for $\I\filter$:
\begin{equation}
\label{eq:scott:ltl}
  \forall \stream \in \I{\Stream\Base},~
  \text{$\stream$ total},
  \qquad
  \stream \forces \Box\Diam \Psi
  \quad\longimp\quad
  (\I\filter\ p\ \stream) \forces \Box \Phi
\end{equation}

\noindent
where we refrained from writing $\stream \forces \Box \Phi$
for the assumption that $\stream$ is total.

We use the notations of Remark~\ref{rem:prelim:filter}.
In particular,
$\I\filter p$ is the least fixpoint of the
Scott-continuous function $f_p \colon X \to X$,
where $X \deq \CPO\funct{\I{\Stream\Base},\I{\Stream\Base}}$.
In symbols, we have
$\I\filter p = \term{Y}(f_p) = \bigvee_{n \in \NN} f_p^n(\bot_X)$.

The standard method to reason on such fixpoints is the
rule of \emph{fixpoint induction}
(see e.g.~\cite[\S 6.2]{ac98book}).
This rule asserts that given a subset $\SP$ of a cpo $X$,
and a morphism $f \colon X \to_{\CPO} X$,
we have $\term{Y}(f) \in \SP$ provided
(i)
$\bot_X \in \SP$,
(ii)
$\SP$ is stable under sups of $\omega$-chains,
and (iii)
$f(x) \in \SP$ whenever $x \in \SP$.
In our case, the subset of interest is
\(
  \SP
  \deq
  \left\{ f \mid
  \text{$\stream$ total and $\stream \forces \Box\Diam\Psi$}
  ~\imp~
  f(\stream) \forces \Box\Phi
  \right\}
\).
But fixpoint induction cannot be applied 
since
$\bot_X \notin \SP$
(as $\bot_X$ takes any $\stream \in \I{\Stream\Base}$
to $\bot^\omega \not\forces \Box\Phi$).

We can proceed as follows,
with the help of~\S\ref{sec:geom:trans}.
Given $k, n \in \NN$ with $k \leq n$, let
\[
\begin{array}{r c l}
  \psi_{n,k}
& \deq
& \bigvee
  \left\{
  \bigwedge_{1 \leq j \leq k}
  \Next^{i_j}\th F\I\Psi
  \mid
  0 \leq i_1 < \dots < i_k < n
  \right\}
\\

  \varphi_k
& \deq
& \bigwedge_{m < k} \Next^m \th F\I\Phi
\end{array}
\]

\noindent
Note that
$\nu(\stream) \in \Mod\left(\th T\I{\Box\Diam\Psi} \right)$
if, and only if,
$(\forall k \in \NN)(\exists n \geq k)(\nu(\stream) \models \psi_{n,k})$.
It follows that
condition~\eqref{eq:scott:ltl}
can be obtained from the following.
\begin{equation}
\label{eq:scott:geom}
  \forall \stream \in \I{\Stream\Base},~
  \text{$\stream$ total},~
  \forall k \in \NN,~
  \forall n \geq k,
  \quad
  \nu(\stream) \models \psi_{n,k}
  \quad\longimp\quad
  \nu(\I\filter\ p\ \stream) \models \varphi_k
\end{equation}

Condition~\eqref{eq:scott:geom} is a consequence of Lemma~\ref{lem:scott} below.
The main inductive argument 
is encapsulated in
item (\ref{item:scott:ind}).

\begin{lemm}
\label{lem:scott}
Write $g_n$ for $f_p^n(\bot_X) \colon \I{\Stream\Base} \to_\CPO \I{\Stream\Base}$.
Let $\stream \in \I{\Stream\Base}$ be a total stream.
\begin{enumerate}[(1)]
\item
\label{item:scott:ind}
Assume $k \leq n$.
If $\nu(\stream) \models \psi_{n,k}$,
then $\nu(g_n(\stream)) \models \varphi_k$.

\item
Let $n,k \in \NN$.
If $\nu(g_n(\stream)) \models \varphi_k$,
then $\nu(\I\filter\ p\ \stream) \models \varphi_k$.
\end{enumerate}
\end{lemm}

\section{Conclusion}
\label{sec:conc}

In this paper, we conducted a semantic study of a logic $\LTL$
on a domain of streams $\I{\Stream\Base}$.
We showed that the negation-free formulae of $\LTL$ induce
sober subspaces of $\I{\Stream\Base}$,
and that this may fail in presence of negation.
We proposed an inductive translation of negation-free $\LTL$ to geometric logic.
This translation reflects the semantics of $\LTL$,
and we use it to prove that the denotation of $\filter$ 
satisfies the specification~\eqref{eq:intro:spec}.


\paragraph{Further Works.}
%
First, the logic $\LTL$ on $\I{\Stream\Base}$
deserves further studies, in particular regarding decidability and
possible axiomatizations.

\opt{short}{\renewcommand\fntext{For geometric logic, see~\cite[\S 5]{rs23full}
which completely axiomatizes the countable (e.g.\ $G_\delta$) case.}}
\opt{long}{\renewcommand\fntext{For geometric logic, see~\S\ref{sec:free}
which completely axiomatizes the countable (e.g.\ $G_\delta$) case.}}

We think an important next step would be to propose a refinement type
system in the spirit of~\cite{jr21esop}, but for an extension of $\PCF$ with streams.
More precisely, the system of~\cite{jr21esop} crucially relies on controlled
unfoldings of (formula level) fixpoints.
We think that our translation to geometric logic could provide a domain-theoretic
analogue for that, yielding a system grounded on DTLF
(in the form of~\cite[\S 4.3]{abramsky91apal} or e.g.~\cite[\S 10.5]{ac98book}).
This may rely on a deduction system for either $\LTL$ or geometric
logic.\opt{short,long}{\fn\ }%
In any case, we expect to need an analogue of the iteration terms of~\cite{jr21esop},
which actually could simulate (enough of) the infinitary aspects of geometric logic.
Also, an important task in this direction would be to formulate sufficiently general 
reasoning principles for program-level fixpoints.

\renewcommand\fntext{This corresponds to
``alternation depth $1$'' in~\cite[\S 2.2]{bw18chapter}.
See also~\cite[\S 7]{bs07chapter} and~\cite{sv10apal}.}

Further, we expect to handle alternation-free modal $\mu$-properties\fn\
on (finitary) polynomial types, thus targeting a system which as a whole
would be based on $\FPC$.
But polynomial types involve sums,
and sums are not universal in $\CPO$, in contrast with $\DCPO$
and with the category $\CPO_\bot$ of \emph{strict} functions.
We think of working with Call-By-Push-Value (CBPV)~\cite{levy03book,levy22siglog}
for the usual adjunction between $\DCPO$ and $\CPO_\bot$.
On the long run, it would be nice if this basis could extend to
enriched models of CBPV,
so as to handle further computational effects.
Print and global store are particularly relevant,
as an important trend in proving temporal properties
considers programs generating streams of events.
Major works in this line include
\cite{ssv08jfp,hc14lics,hl17lics,nukt18lics,kt14lics,ust17popl,nukt18lics,%
su23popl}.
In contrast with ours, these approaches are based on trace semantics
of syntactic expressions rather than denotational domains.%
\footnote{See e.g.~\cite[Theorem 4.1 (and Figure 6)]{nukt18lics}
or~\cite[Theorem 1 (and Definition 20 from the full version)]{su23popl}.}

In a different direction, we think our approach based on geometric
logic could extend to linear types~\cite{hjk00mscs},
for instance targeting systems like~\cite{nw03concur,winskel04llcs},
and relying on the categorical study of~\cite{bf06book}.

\paragraph{Acknowledgements.}
This work was partially supported by the ANR-21-CE48-0019 -- RECIPROG
and by the LABEX MILYON (ANR-10-LABX-0070) of Université de Lyon.
It started as a spin-off of
ongoing work
with Guilhem Jaber and Kenji Maillard.
G.\ Jaber proposed the $\filter$ function as a motivating example.
\opt{full,long}{Thomas Streicher pointed to us the reference~\cite{heckmann15mscs}.}%
\opt{short}{Thomas Streicher pointed to us the reference~\cite{heckmann15mscs}
(see \cite[\S 5]{rs23full}).}

\bibliographystyle{alpha-fr}
\clearpage
\bibliography{bibliographie}

\opt{full,long}{\newpage}
\opt{full,long}{\appendix}

\opt{long}{\input{app-intro}}	
\opt{long}{}	
\opt{full,long}{
\section{Proofs of~\S\ref{sec:prelim} (\nameref{sec:prelim})}
\label{sec:app:prelim}

\subsection{Proofs of~\S\ref{sec:prelim:domain} (\nameref{sec:prelim:domain})}
\label{sec:app:prelim:domain}

\subsubsection{Proof of Remark~\ref{rem:prelim:stream:domain-equation}}

We first discuss Remark~\ref{rem:prelim:stream:domain-equation}.
Consider the cpo $\I{\Stream\Base}$
equipped with the isomorphism
\[
  \I\Base \times \I{\Stream\Base}
  \to_{\CPO}
  \I{\Stream\Base}
  ,\quad
  (a,\stream)
  \mapsto
  a \cdot \stream
\]

\noindent
(where $(-) \times (-)$ is equipped with the pointwise order),
with inverse $\stream \mapsto (\stream(0), \stream\restr 1)$.
It is thus clear that $\I{\Stream\Base}$ is \emph{a} solution in $\CPO$
of the recursive domain equation
\[
\begin{array}{l l l}
  X
& \cong
& \I\Base
  \times
  X
\end{array}
\]

But in Remark~\ref{rem:prelim:stream:domain-equation},
we claimed that $\I{\Stream\Base}$ is actually
\emph{the solution of $X \cong \I\Base \times X$ in the usual sense}.
By this, we mean the following.
In presence of general recursive types, one has to solve equations
\begin{equation}
\label{eq:prelim:mixed}
\begin{array}{l l l}
  F(X,X)
& \cong
& X
\end{array}
\end{equation}

\noindent
where
\[
\begin{array}{l l r c l}
  F
& :
& \CPO^\op
  \times
  \CPO
& \longto
& \CPO
\end{array}
\]

\noindent
is a functor of mixed variance.
The usual solution, as presented in e.g.~\cite[\S 7.1]{ac98book},
is to replace $\CPO$ with the category $\CPO^{\mathrm{ip}}$
of \emph{injection-projection pairs} (\cite[Definition 7.1.8]{ac98book}),
and to replace $F$ with a \emph{covariant} functor
\(
  F^{\mathrm{ip}} \colon
  \CPO^{\mathrm{ip}} \times \CPO^{\mathrm{ip}}
  \to \CPO^{\mathrm{ip}}
\).
Then instead of~\eqref{eq:prelim:mixed} one solves
the following
\begin{equation}
\label{eq:prelim:ip}
\begin{array}{l l l}
  F^{\mathrm{ip}}(X,X)
& \cong
& X
\end{array}
\end{equation}

\noindent
in $\CPO^{\mathrm{ip}}$.
In turn, equation~\eqref{eq:prelim:ip}
is usually solved using colimits of $\omega$-chains
(\cite[Proposition 7.1.3]{ac98book}).
But by~\cite[Theorem 7.1.10]{ac98book},
colimits of $\omega$-chains in $\CPO^{\mathrm{ip}}$
can actually be computed from limits of $\omega^\op$-chains in $\CPO$.

We return to the case of streams, and consider the functor
\[
\begin{array}{*{7}{l}}
  F
& =
& \I\Base \times (-)
& :
& \CPO
& \longto
& \CPO
\end{array}
\]

\noindent
Then by~\cite[Definition 7.1.17 and Proof of Theorem 7.1.10]{ac98book},
we have to consider the limit of the $\omega^\op$-chain
\[
\begin{tikzcd}[column sep=large]
  \one
& F(\one)
  \arrow{l}[above]{\oc}
& F^2(\one)
  \arrow{l}[above]{F(\oc)}
& F^n(\one)
  \arrow[dashed]{l}
& F^{n+1}(\one)
  \arrow{l}[above]{F^n(\oc)}
& \phantom{F}
  \arrow[dashed]{l}
\end{tikzcd}
\]

\noindent
where $\one$ is the terminal cpo $\{\bot\}$.
Now, by~\cite[Proposition 7.1.13]{ac98book},
the limit of the above $\omega^\op$-chain is given by
\[
\begin{array}{c}
  \left\{
  \alpha \in \prod_{n \in \NN} F^n(\one)
  \mid
  \alpha(n) = F^n(\oc)(\alpha(n+1))
  \right\}
\end{array}
\]

\noindent
It is easy to see that $\I{\Stream\Base}$
is isomorphic to this limit.
The key is to define for each $n \in \NN$ an
isomorphism $\iota_n \colon \I\Base^n \to F^n(\one)$
with $\iota_0 = \id_\one \colon \one \to \one$
and
\[
\begin{array}{l l r c l}
  \iota_{n+1}
& :
& \I\Base^{n+1}
& \longto
& F^{n+1}(\one) = \I\Base \times F^n(\one)
\\

&
& (a_1,a_2,\dots,a_{n+1})
& \longmapsto
& (a_1, \iota_n(a_2,\dots,a_{n+1}))
\end{array}
\]

\noindent
and to observe that the following commutes
\[
\begin{tikzcd}
  \I\Base^{n+1}
  \arrow{d}[left]{(a_1,\dots,a_n,a_{n+1}) \mapsto (a_1,\dots,a_n)}
  \arrow{r}{\iota_{n+1}}
& F^{n+1}(\one)
  \arrow{d}{F^n(\oc)}
\\
  \I\Base^{n}
  \arrow{r}{\iota_{n}}
& F^n(\one)
\end{tikzcd}
\]

\subsubsection{Proof of Example~\ref{ex:prelim:stream:algebraic}}

Let $\Base$ be a set. We show that the cpo $\I{\Stream\Base}$
is algebraic, and that its finite elements are precisely those
of finite support.

\begin{lemm}
Given $\stream \in \I{\Stream\Base}$,
the set 
\(
 \left\{
   d \mid
   \text{$d$ of finite support and $d \leq_{\I{\Stream\Base}} \stream$}
  \right\}
\)
is directed.
\end{lemm}

\begin{proof}
First, the set is non-empty since
e.g.\ $\stream(0) \cdot \bot^\omega$
has finite (possibly empty) support and is
$\leq_{\I{\Stream\Base}} \stream$.

Let $d,d' \leq_{\I{\Stream\Base}} \stream$, with $d$ and $d'$
of finite support.
Since $\I\Base$ is flat, for all $n \in \supp(d) \cap \supp(d')$,
we have $d(n) = d'(n)$.
It follows that for each $n \in \NN$, the set
$\{d(n),d'(n)\}$ has a greatest element,
so that $d(n) \bigvee_{\I\Base} d'(n)$ exists in $\I\Base$.
We can thus define
\[
\begin{array}{*{5}{l}}
  d''
& :
& n \in \NN
& \longmapsto
& d(n) \bigvee_{\I\Base} d'(n)
\end{array}
\]

\noindent
Then $d''$ has finite support.
Moreover, we have $d'' \leq_{\I{\Stream\Base}} \stream$
and $d, d' \leq_{\I{\Stream\Base}} d''$.
\end{proof}

\begin{lemm}
Given $\stream \in \I{\Stream\Base}$,
we have
\[
\begin{array}{l l l}
  \stream
& =
  \bigvee
  \left\{
   d \mid
   \text{$d$ of finite support and $d \leq_{\I{\Stream\Base}} \stream$}
  \right\}
\end{array}
\]
\end{lemm}

\begin{proof}
It is clear that $\stream$ is an upper bound of
\(
 \left\{
   d \mid
   \text{$d$ of finite support and $d \leq_{\I{\Stream\Base}} \stream$}
  \right\}
\).

Let now $\streambis \in \I{\Stream\Base}$ such that
$d \leq_{\I{\Stream\Base}} \streambis$ for all $d$ of finite support
such that $d \leq_{\I{\Stream\Base}} \stream$.
We show that $\stream \leq_{\I{\Stream\Base}} \streambis$.
Let $n \in \NN$ and $d \deq \bot^n \cdot \stream(n) \cdot \bot^\omega$.
Then $d$ is of finite support and $d \leq_{\I{\Stream\Base}} \stream$.
Hence $d \leq_{\I{\Stream\Base}} \streambis$
and $\stream(n) \leq_{\I\Base} \streambis(n)$.
It follows that for all $n \in \NN$ we have $\stream(n) \leq_{\I\Base} \streambis(n)$,
and by definition we get
$\stream \leq_{\I{\Stream\Base}} \streambis$.
\end{proof}

\begin{lemm}
$\stream \in \I{\Stream\Base}$ is finite if, and only if, $\stream$ has finite support.
\end{lemm}

\begin{proof}
Assume first that $\stream$ is finite.
Since
\[
\begin{array}{l l l}
  \stream
& \leq
& \bigvee
  \left\{
   d \mid
   \text{$d$ of finite support and $d \leq_{\I{\Stream\Base}} \stream$}
  \right\}
\end{array}
\]

\noindent
where the set in the r.-h.s.\ is directed,
there is some $d$ of finite support such that
$d \leq_{\I{\Stream\Base}} \stream$
and
$\stream \leq_{\I{\Stream\Base}} d$.
Hence $\stream = d$ has finite support.

Conversely, assume $d$ has finite support, and let
$D \sle \I{\Stream\Base}$ be directed with $d \leq_{\I{\Stream\Base}} \bigvee D$.
Let $n \in \supp(d)$.
The set $D(n) = \{ \stream(n) \mid \stream \in D\}$ is directed,
and since $\I\Base$ is flat, this set a greatest element.
Hence $d(n) \leq_{\I\Base} \stream_n(n)$ for some $\stream_n \in D$.
Since $\supp(d)$ is finite and (again) since $D$ is directed,
we obtain that $d \leq_{\I{\Stream\Base}} \stream$ for some $\stream \in D$.
\end{proof}

\subsection{Proofs of~\S\ref{sec:prelim:ltl} (\nameref{sec:prelim:ltl})}
\label{sec:app:prelim:ltl}

\begin{lemm}[Remark~\ref{rem:ltl:fix}]
\label{lem:app:ltl:fix}
Given a complete lattice $L$ and a (monotone) function $f \colon L \to L$,
we write $\lfp(f)$ and $\gfp(f)$ for the least and the greatest
fixpoint of $f$, respectively.
\begin{enumerate}[(1)]
\item
Given $\LTL$ formulae $\Phi, \Psi$,
let $H_{\Phi,\Psi} \colon \Po(\I{\Stream\Base}) \to \Po(\I{\Stream\Base})$
take $\SP$ to $\I\Psi \cup (\I\Phi \cap \I\Next(\SP))$.
Then
\[
\begin{array}{l !{=} l !{\qquad\text{and}\qquad} l !{=} l}
  \I{\Phi \Ushort \Psi}
& \lfp(H_{\Phi,\Psi})

& \I{\Phi \Wshort \Psi}
& \gfp(H_{\Phi,\Psi})
\end{array}
\]

\item
Given a complete atomic Boolean algebra and a monotone $f \colon B \to B$,
we have
\[
\begin{array}{l !{=} l !{\quad\text{and}\quad} l !{=} l}
  \lfp(f)
& \lnot \gfp(b \mapsto \lnot f(\lnot b))

& \gfp(f)
& \lnot \lfp(b \mapsto \lnot f(\lnot b))
\end{array}
\]
\end{enumerate}
\end{lemm}

\begin{proof}
\hfill
\begin{enumerate}[(1)]
\item
The proof mimics that of~\cite[Lemmas 5.18 and 5.19]{bk08book}.
First note that we have
\[
\begin{array}{l l l}
  \Phi \Wshort \Psi
& \equiv
& \left(\Phi \Ushort \Psi\right)
  \lor
  \Box \Phi
\end{array}
\]

\noindent
and 
$\I{\Phi \Ushort \Psi} \sle \I{\Phi \Wshort \Psi}$.

We show that $\I{\Phi \Ushort \Psi}$
and $\I{\Phi \Wshort \Psi}$ are indeed fixpoints of $H_{\Phi,\Psi}$,
that is
\[
\begin{array}{r c l}
  \I{\Phi \Ushort \Psi}
& =
& \I\Psi \cup
  \left(
  \I\Phi \cap \I\Next\I{\Phi \Ushort \Psi}
  \right)
\\

  \I{\Phi \Wshort \Psi}
& =
& \I\Psi \cup
  \left(
  \I\Phi \cap \I\Next\I{\Phi \Wshort \Psi}
  \right)
\end{array}
\]

\noindent
The case of $\I{\Phi \Ushort \Psi}$ is a trivial unfolding of the definition:
given $\stream \in \I{\Stream\Base}$, we have
\[
\begin{array}{l l l}
  \stream \in \I{\Phi \Ushort \Psi}
& \text{iff}
& \text{there is $i \geq 0$ with}~
  \stream \restr 0,\dots, \stream \restr(i-1) \in \I\Phi
  ~\text{and}~
  \stream \restr i \in \I{\Psi}
\\

& \text{iff}
& \stream \in \I\Psi
  ~\text{or}~
  \left(
  \text{there is $i \geq 1$ with}~
  \stream \restr 0,\dots, \stream \restr(i-1) \in \I\Phi
  ~\text{and}~
  \stream \restr i \in \I{\Psi}
  \right)
\\

& \text{iff}
& \stream \in \I\Psi
  ~\text{or}~
  \left(
  \stream \in \I\Phi
  ~\text{and}~
  \text{there is $i \geq 1$ with}~
  \stream \restr 1,\dots, \stream \restr(i-1) \in \I\Phi
  ~\text{and}~
  \stream \restr i \in \I{\Psi}
  \right)
\\

& \text{iff}
& \stream \in \I\Psi
  ~\text{or}~
  \left(
  \stream \in \I\Phi
  ~\text{and}~
  \stream \restr 1 \in \I{\Phi \Ushort \Psi}
  \right)
\\

& \text{iff}
& \stream \in
  \I\Psi \cup
  \left(
  \I\Phi \cap \I\Next\I{\Phi \Ushort \Psi}
  \right)
\end{array}
\]

We now turn to $\I{\Phi \Wshort \Psi}$.
Since $\I{\Phi \Ushort \Psi} \sle \I{\Phi \Wshort \Psi}$,
and since $\I\Next$ is monotone (for inclusion),
we have
\[
\begin{array}{r c l}
  \I{\Phi \Ushort \Psi}
& \sle
& \I\Psi \cup
  \left(
  \I\Phi \cap \I\Next\I{\Phi \Wshort \Psi}
  \right)
\end{array}
\]

\noindent
Let $\stream \in \I{\Phi \Wshort \Psi} \setminus \I{\Phi \Ushort \Psi}$,
i.e.\ $\stream \in \I{\Box \Phi}$.
Since $\I{\Box \Phi} \sle \I{\Phi \Wshort \Psi}$,
we get
\[
\begin{array}{*{5}{l}}
  \stream
& \in
& \I\Phi \cap \I\Next\I{\Box\Phi}
& \sle
& \I\Psi \cup
  \left(
  \I\Phi \cap \I\Next\I{\Phi \Wshort \Psi}
  \right)
\end{array}
\]

\noindent
It follows that
\[
\begin{array}{r c l}
  \I{\Phi \Wshort \Psi}
& \sle
& \I\Psi \cup
  \left(
  \I\Phi \cap \I\Next\I{\Phi \Wshort \Psi}
  \right)
\end{array}
\]

\noindent
For the converse inclusion,
if $\stream \in \I\Psi$,
then we have
$\stream \in \I{\Phi \Ushort \Psi} \sle \I{\Phi \Wshort \Psi}$.
Assume now
$\stream \in \I\Phi \cap \I\Next\I{\Phi \Wshort \Psi}$.
If $\stream \in \I\Next\I{\Phi \Ushort \Psi}$,
then we obtain
$\stream \in \I{\Phi \Ushort \Psi} \sle \I{\Phi \Wshort \Psi}$.
Otherwise, we have $\stream \in \I\Next\I{\Box\Phi}$,
and since $\stream \in \I\Phi$, this gives
$\stream \in \I{\Box\Phi} \sle \I{\Phi \Wshort \Psi}$.

Hence $\I{\Phi \Ushort \Psi}$ and $\I{\Phi \Wshort \Psi}$
are indeed fixpoints of $H_{\Phi,\Psi}$.

We show that $\I{\Phi \Ushort \Psi}$ is the \emph{least} fixpoint of $H_{\Phi,\Psi}$.
Let $P \in \Po(\I{\Stream\Base})$ such that
$P = H_{\Phi,\Psi}(P)$.
We have to show that $\I{\Phi \Ushort \Psi} \sle P$.
But if $\stream \in \I{\Phi \Ushort \Psi}$,
there is some $i \geq 0$ such that
$\stream \restr 0,\dots, \stream \restr(i-1) \in \I\Phi$
and
$\stream \restr i \in \I{\Psi}$.
Since $P$ is a fixpoint of $H_{\Phi,\Psi}$,
we get $\stream \restr i \in P$.
Again since $P$ is a fixpoint of $H_{\Phi,\Psi}$,
we obtain
$\stream \restr (i-1),\dots, \stream \restr 0 \in P$.
Hence $\stream = \stream \restr 0 \in P$ and we are done.

It remains to show that 
$\I{\Phi \Wshort \Psi}$ is the greatest fixpoint of $H_{\Phi,\Psi}$.
Let $Q \in \Po(\I{\Stream\Base})$ such that
\[
\begin{array}{r c l}
  Q
& =
& \I\Psi \cup
  \left(
  \I\Phi \cap \I\Next Q
  \right)
\end{array}
\]

\noindent
We have to show that $Q \sle \I{\Phi \Wshort \Psi}$.
Let $\stream \in Q$.
If $\stream \in \I{\Box\Phi}$ then we are done.
Otherwise, there is a \emph{least} $i \geq 0$ such that
$\stream \restr i \notin \I{\Phi}$.
Hence $\stream \restr 0,\dots \stream \restr (i-1) \in \I{\Phi}$.
Since $\stream \in Q = H_{\Phi,\Psi}(Q)$, it follows that
$\stream \restr 1,\dots \stream \restr i \in Q$.
Hence $\stream \restr i \in \I\Psi$ since
$\stream \restr i \notin \I\Phi$.
It follows that
$\stream \in \I{\Phi \Ushort \Psi} \sle \I{\Phi \Wshort \Psi}$.

\item
Let $f \colon B \to B$ be monotone with $B$ a complete atomic Boolean algebra.
By the Knaster-Tarski Fixpoint Theorem
(see e.g.~\cite[2.35]{dp02book}), we have
\[
\begin{array}{r c l}
  \lfp(f)
& =
& \bigvee \left\{
  a \in B \mid f(a) \leq a
  \right\}
\\

  \gfp(f)
& =
& \bigwedge \left\{
  a \in B \mid a \leq f(a)
  \right\}
\end{array}
\]

Let $g$ be the monotone function $B \to B$ which takes $a$ to $\lnot f(\lnot a)$.
We show
\[
\begin{array}{l l l}
  \gfp(f)
& =
& \lnot \lfp(g)
\end{array}
\]

\noindent
The other equation is then obtained by duality.
Since $B$ is complete and atomic, it follows from
e.g.~\cite[Theorem 10.24]{dp02book} that
$\lnot \bigvee \SP = \bigwedge \{\lnot s \mid s \in S \}$
for every $\SP \sle B$.
We then compute
\[
\begin{array}{l l l}
  \lnot \lfp(g)
& =
& \lnot
  \bigvee \left\{ a \mid g(a) \leq a \right\}
\\

& =
& \bigwedge \left\{ \lnot a \mid g(a) \leq a \right\}
\\

& =
& \bigwedge \left\{ \lnot a \mid \lnot f(\lnot a) \leq a \right\}
\\

& =
& \bigwedge \left\{ \lnot a \mid \lnot a \leq f(\lnot a) \right\}
\\

& =
& \bigwedge \left\{ b \mid b \leq f(b) \right\}
\\

& =
& \gfp(f)
\end{array}
\]
\qedhere
\end{enumerate}
\end{proof}

\begin{lemm}[Lemma~\ref{lem:ltl:up}]
\label{lem:app:ltl:up}
If $\Phi$ is negation-free then $\I\Phi$
is upward-closed in $\I{\Stream\Base}$
(if $\stream \in \I\Phi$ and $\stream \leq_{\I{\Stream\Base}} \streambis$
then $\streambis \in \I\Phi$).
\end{lemm}

\begin{proof}
We reason by induction $\Phi$.
Since upward-closed sets are stable under (arbitrary)
unions and intersections, we just have to consider the cases of
atomic formulae $a \in \Base$ and of modalities.

Concerning atomic formulae, if $\stream \in \I{a}$, then $\stream(0) = a$.
Since $\I\Base$ is a flat cpo, we have $\streambis(0) = a$
for all $\streambis \geq_{\I{\Stream\Base}} \stream$.
Hence $\I{a}$ is upward-closed.

The cases of modalities
$\Next\Phi$, $(\Phi \Ushort \Psi)$ and $(\Phi \Wshort \Psi)$
follow from the induction hypothesis and
the fact that
$\stream \leq_{\I{\Stream\Base}} \streambis$
implies
$\stream\restr i \leq_{\I{\Stream\Base}} \streambis \restr i$
for all $i \in \NN$.
\end{proof}

}
\opt{full,long}{
\section{Proofs of~\S\ref{sec:frames} (\nameref{sec:frames})}
\label{sec:app:frames}

\subsection{Proofs of~\S\ref{sec:frames:sub} (\nameref{sec:frames:sub})}
\label{sec:app:frames:sub}

Let $(X,\Open)$ be a topological space,
and fix $P \sle X$.
Write $\incl \colon (P,\Open\restr P) \emb (X,\Open)$
for the subspace inclusion,
and let $\ladj\incl = \Open(\incl) \colon \Open \quot \Open\restr P$
be the induced surjective frame morphism.

The following simple observation is used repeatedly below.
\begin{lemm}
\label{lem:app:frames:sub:base}
%
Given $x \in P$, we have 
$\Filt_x = \Filt^{\Open\restr P}_x \comp \ladj\incl$
where
\(
  \Filt^{\Open\restr P}_x
  \deq
  \left\{ (U \cap P) \in \Open\restr P \mid x \in U \cap P \right\}
\).
\end{lemm}

\begin{proof}
Given $U \in \Open$, if $U \in \Filt$ then $x \in U$,
so that $x \in U \cap P$ and $\ladj\incl(U) = (U \cap P) \in \Filt^{\Open\restr P}_x$.
Conversely, if $\ladj\incl(U) \in \Filt^{\Open\restr P}_x$, then in particular
$x \in U$ so that $U \in \Filt_x$.
\end{proof}

\begin{lemm}[Lemma~\ref{lem:frames:sub:soberfilt}]
\label{lem:app:frames:sub:soberfilt}
Assume that $(X,\Open)$ is sober. Then the following are equivalent.
\begin{enumerate}[(i)]
\item
$(P,\Open\restr P)$ is sober.

\item
For each $x \in X$, we have $x \in P$ if, and only if,
there is some $\Filtbis \in \pt(\Open\restr P)$
such that
$\Filt_x = \Filtbis \comp \ladj\incl$.
\[
\begin{tikzcd}[row sep=tiny]
  \Open
  \arrow[twoheadrightarrow]{rr}{\ladj\incl}
  \arrow{dr}[below]{\Filt_x}
&
& \Open\restr P
  \arrow[dashed]{dl}{\Filtbis}
\\
& \two
\end{tikzcd}
\]
\end{enumerate}
\end{lemm}

\begin{proof}
Let $(X,\Open)$ be sober, and fix some $P \sle X$.
%
%
In view of Lemma~\ref{lem:app:frames:sub:base},
we just have to show that the 
following are equivalent.
\begin{enumerate}[(i)]
\item
\label{item:app:frames:sub:soberfilt:sober}
The space $(P,\Open\restr P)$ is sober.

\item
\label{item:app:frames:sub:soberfilt:triangle}
For each $x \in X$, we have
$x \in P$ whenever
there is some $\Filtbis \in \pt(\Open\restr P)$
such that
$\Filt_x = \Filtbis \comp \ladj\incl$.
\end{enumerate}

We discuss each implication separately.
\begin{description}
\item[\(
\text{(\ref{item:app:frames:sub:soberfilt:sober})}
\imp
\text{(\ref{item:app:frames:sub:soberfilt:triangle})}
\).]

Assume that $(P,\Open\restr P)$ is sober, and let $x \in X$.
Let $\Filtbis \in \pt(\Open\restr P)$
such that $\Filt_x = \Filtbis \comp \ladj\incl$.
We have to show that $x \in P$.
Since $(P,\Open\restr P)$ is sober, we have
$\Filtbis = \Filt^{\Open\restr P}_y$ for some unique $y \in P$.
But we have seen above that
$\Filt_y = \Filt^{\Open\restr P}_y \comp \ladj\incl$.
Hence $\Filt_x = \Filt_y$ and $x = y$ since $X$ is sober.

\item[\(
\text{(\ref{item:app:frames:sub:soberfilt:triangle})}
\imp
\text{(\ref{item:app:frames:sub:soberfilt:sober})}
\).]

Assume condition (\ref{item:app:frames:sub:soberfilt:triangle})
and let $\Filtbis \in \pt(\Open\restr P)$.
We have to show that $\Filtbis = \Filt_x^{\Open\restr P}$
for a unique $x \in P$.

Since $X$ is sober, there is some $x \in X$
such that $\Filt_x = \Filtbis \comp \ladj\incl$.
Condition (\ref{item:app:frames:sub:soberfilt:triangle})
implies that $x \in P$, so that
$\Filt_x = \Filt^{\Open\restr P}_x \comp \ladj\incl$.
But since $\ladj\incl$ is a surjective frame morphism,
it is in particular an epimorphism in $\Frm$.
Hence
$\Filtbis \comp \ladj\incl = \Filt^{\Open\restr P}_x \comp \ladj\incl$
implies $\Filtbis = \Filt^{\Open\restr P}_x$.

Moreover, if $\Filtbis = \Filt_y^{\Open\restr P}$
for some $y \in P$,
then since
$\Filt_y = \Filt^{\Open\restr P}_y \comp \ladj\incl$,
we have $\Filt_x = \Filt_y$.
Hence $y = x$ since $X$ is sober.
\qedhere
\end{description}
\end{proof}


Write $j \colon \Open \to \Open$ for the nucleus
induced by the surjective frame morphism
$\ladj\incl \colon \Open \quot \Open\restr P$,
and write $\widetilde P$ for the frame of $j$-fixpoints
($\widetilde P = \{ U \in \Open \mid j(U) = U \}$).

\begin{lemm}[Remark~\ref{rem:frames:sub:nucleus}]
\label{lem:app:frames:sub:nucleus}
Given an open $U \in \Open$ of $X$, we have
\[
\begin{array}{l l l}
  j(U)
& =
& \bigcup \left\{ V \in \Open \mid
  V \cap P = U \cap P
  \right\}
\end{array}
\]
\end{lemm}

\begin{proof}
Fix $U \in \Open$.
By definition, we have
\[
\begin{array}{l l l}
  j(U)
& =
& \radj\incl\left( \ladj\incl(U) \right)
\\

& =
& \bigcup \left\{ V \in \Open \mid
  \ladj\incl(V) \sle \ladj\incl(U)
  \right\}
\\

& =
& \bigcup \left\{ V \in \Open \mid
  V \cap P \sle U \cap P
  \right\}
\end{array}
\]

Hence, if $V \in \Open$ is such that $V \cap P = U \cap P$,
then in particular $V \cap P \sle U \cap P$
and thus $V \sle j(U)$.
It follows that
\(
  \bigcup \left\{ V \in \Open \mid
  V \cap P = U \cap P
  \right\}
  \sle
  j(U)
\).

It remains to show that
\(
  j(U)
  \sle
  \bigcup \left\{ V \in \Open \mid
  V \cap P = U \cap P
  \right\}
\).
Since $j(U) \in \Open$, we are done if we show that
$j(U) \cap P = U \cap P$.
But we have
\(
  j(U) \cap P
  =
  \bigcup \left\{ V \cap P \mid
  V \cap P \sle U \cap P
  \right\}
\)
so that 
$j(U) \cap P \sle U \cap P$.
Since $U \sle j(U)$, we obtain
$j(U) \cap P = U \cap P$,
as required.
\end{proof}

Recall that given $x \in \Open$,
we set $\widetilde x = X \setminus \clos{\{x\}} = X \setminus \down x$,
where $y \in \down x$ iff $y \leq_\Open x$,
with $\leq_\Open$
the specialization (pre)order of $(X,\Open)$.

\begin{lemm}[Remark~\ref{rem:frames:sub:widetilde}]
\label{lem:app:frames:sub:widetilde}
Given $x \in X$ and $U \in \Open$, we have $U \sle \widetilde x$
if, and only if, $x \notin U$.
\end{lemm}

\begin{proof}
If $U \sle \widetilde x$, then we obviously have $x \notin U$
since $x \notin \widetilde x$.
Assume conversely that $U \not\sle \widetilde x$.
Hence there is some $y \in U$ such that $y \notin \widetilde x$.
This implies $y \leq_\Open x$ with $y \in U$,
so that $x \in U$.
\end{proof}

\begin{lemm}[Lemma~\ref{lem:frames:sub:widefilt}]
\label{lem:app:frames:sub:widefilt}
Given $x \in X$, let
$\widetilde\Filt_x \deq \{ U \in \widetilde P \mid x \in U \}$.
Then
$\widetilde\Filt_x \in \pt(\widetilde P)$
if, and only if,
$\widetilde x \in \widetilde P$.
\end{lemm}

\begin{proof}
Assume first that $\widetilde x \in \widetilde P$.
It is clear that $\widetilde\Filt_x$ is upward-closed and
stable under finite intersections.
Let $\SP \sle \widetilde P$ such that for all $U \in \SP$,
we have $U \sle \widetilde x$.
Then $\bigcup \SP \sle \widetilde x$
and $j(\bigcup \SP) \sle j(\widetilde x) = \widetilde x$
(since $\widetilde x \in \widetilde P$).
Hence, if $x \in j(\bigcup \SP)$, we have $j(\bigcup \SP) \not\sle \widetilde x$,
and there is some $U \in \SP$ such that $U \not\sle \widetilde x$,
i.e.\ $x \in U$.

Assume conversely that $\widetilde\Filt_x \in \pt(\widetilde P)$.
Let $\SP \deq \{ V \in \widetilde P \mid V \sle \widetilde x\}$.
Then for all $V \in \SP$, we have $x \notin V$ and thus
$V \notin \widetilde\Filt_x$.
Hence $j(\bigcup\SP) \notin \widetilde\Filt_x$, so that
$x \notin j(\bigcup\SP)$ and $j(\bigcup \SP) \sle \widetilde x$.
But this implies $j(\widetilde x) \sle \widetilde x$
since $\bigcup \SP = \widetilde x$.
\end{proof}

It is well-known that $x \in P$ implies $\widetilde x \in \widetilde P$
(\cite[Remark VI.1.3.1]{pp12book}).

\begin{prop}[Proposition~\ref{prop:frames:sub:tilde}]
\label{prop:app:frames:sub:tilde}
Assume that $(X,\Open)$ is sober.
Then the following are equivalent.
\begin{enumerate}[(i)]
\item
\label{item:app:frames:sub:tilde:sober}
$(P,\Open\restr P)$ is sober.

\item
\label{item:app:frames:sub:tilde:tilde}
For each $x \in X$, we have $x \in P$ if, and only if,
$\widetilde x \in \widetilde P$.
\end{enumerate}
\end{prop}

\begin{proof}
We prove each implication separately.
\begin{description}
\item[\(
\text{(\ref{item:app:frames:sub:tilde:tilde})}
\imp
\text{(\ref{item:app:frames:sub:tilde:sober})}
\).]

We show that $(P,\Open\restr P)$ is sober
whenever $\widetilde x \in \widetilde P$ implies $x \in P$
for all $x \in X$.
By Lemma~\ref{lem:app:frames:sub:base} and
Lemma~\ref{lem:app:frames:sub:soberfilt} (i.e.\ Lemma~\ref{lem:frames:sub:soberfilt}),
it is sufficient to show that given $x \in X$,
we have $x \in P$ whenever there is some $\Filtbis \in \pt(\Open\restr P)$
such that $\Filt_x = \Filtbis \comp \ladj\incl$.

Assume that for all $x \in X$,
we have $x \in P$ whenever $\widetilde x \in \widetilde P$.
Let $x \in X$ such that 
$\Filt_x = \Filtbis \comp \ladj\incl$
for some $\Filtbis \in \pt(\Open\restr P)$.
We have to show $x \in P$.
Assume toward a contradiction that $x \notin P$.
By our assumption,
this implies that $\widetilde x \notin \widetilde P$.
Hence there is some $y \in j(\widetilde x) \setminus \widetilde x$.
In turn, there is some $V \in \Open$ such that
$y \in V$ and $V \cap P = \widetilde x \cap P$.
Since $y \notin \widetilde x = X \setminus \down x$,
we have $y \leq_\Open x$, and since $y \in V$,
we get $x \in V$.
It follows that $V \in \Filt_x$,
so that $\ladj\incl(V) = (V \cap P) \in \Filtbis$.
But since $V \cap P = \widetilde x \cap P = \ladj\incl(\widetilde x)$,
we thus get $\widetilde x \in \Filt_x$,
a contradiction since $x \notin \widetilde x$.

\item[\(
\text{(\ref{item:app:frames:sub:tilde:sober})}
\imp
\text{(\ref{item:app:frames:sub:tilde:tilde})}
\).]

Assume that $(P,\Open\restr P)$ is sober,
and let $x \in X$ such that $\widetilde x \in \widetilde P$.
We have to show that $x \in P$.
We apply
Lemma~\ref{lem:app:frames:sub:soberfilt} (i.e.\ Lemma~\ref{lem:frames:sub:soberfilt}).
Actually, we are going to construct a point
$\widetilde\Filt \in \pt(\widetilde P)$
such that $\Filt_x = \widetilde\Filt \comp j$.
Since $j = \radj\incl \comp \ladj\incl$
(where $\radj\incl$ is the upper adjoint of $\ladj\incl$),
Lemma~\ref{lem:frames:sub:quot}
then gives the result.

Let $\widetilde\Filt = \widetilde\Filt_x$ be the set of all $U \in \widetilde P$
such that $x \in U$.
Since $\widetilde x \in \widetilde P$,
we have $\widetilde\Filt_x \in \pt(\widetilde P)$
by
Lemma~\ref{lem:app:frames:sub:widefilt} (i.e. Lemma~\ref{lem:frames:sub:widefilt}).
In order to obtain $x \in P$,
it thus remains to show that
\[
\begin{array}{l l l}
  \Filt_x
& =
& \widetilde\Filt_x \comp j
\end{array}
\]

\noindent
But given $U \in \Filt_x$,
we have $x \in U \sle j(U)$,
so that $j(U) \in \widetilde\Filt_x$.
Conversely, let $U \in \Open$
such that $U \notin \Filt_x$.
We have $x \notin U$
and thus $U \sle \widetilde x$.
But then $j(U) \sle j(\widetilde x) = \widetilde x$.
Hence $x \notin j(U)$ and $j(U) \notin \widetilde\Filt_x$.
\qedhere
\end{description}
\end{proof}

\begin{lemm}[Lemma~\ref{lem:frames:sub:notsober}]
\label{lem:app:frames:sub:notsober}
Assume that $(X,\Open)$ is sober
and let $x \in X$ such that for all $U \in \Open$
with $x \in U$,
we have $U \setminus \{x\} \notin \Open$.
Set $P \deq X \setminus \{x\}$.
Then $(P,\Open\restr P)$ is not sober.
\end{lemm}

\begin{proof}
Let $\incl \colon (P,\Open\restr P) \emb (X,\Open)$ be the subspace inclusion.
Since $x \notin P$, by
By Lemma~\ref{lem:app:frames:sub:base} and
Lemma~\ref{lem:app:frames:sub:soberfilt} (i.e.\ Lemma~\ref{lem:frames:sub:soberfilt}),
it is sufficient to show that
$\Filt_x = \Filtbis \comp \ladj\incl$
for some $\Filtbis \in \pt(\Open\restr P)$.
We appeal to Lemma~\ref{lem:frames:sub:quot},
and instead provide a $\widetilde\Filt \in \pt(\Open\restr P)$
such that $\Filt_x = \widetilde\Filt \comp j$
where $j \colon \Open \to \Open$ is the nucleus induced by $\incl$.
Let $\widetilde\Filt = \widetilde\Filt_x$ be the set of all $U \in \widetilde P$
such that $x \in U$.

Note that for each $U \in \Open$, we have
\[
\begin{array}{l l l}
  j(U)
& =
& \bigcup \left\{
  V \in \Open \mid
  V \setminus \{x\} = U \setminus \{x\}
  \right\}
\end{array}
\]

We first show that $\widetilde x \in \widetilde P$.
This amounts to showing that $j(\widetilde x) \sle \widetilde x$,
i.e.\ that $x \notin j(\widetilde x)$.
Assume toward a contradiction that $x \in j(\widetilde x)$.
Hence there is some $V \in \Open$ such that $x \in V$
and
\(
  V \setminus \{x\}
= (X \setminus \clos{\{x\}}) \setminus \{x\}
\).
But
\(
  (X \setminus \clos{\{x\}}) \setminus \{x\}
= X \setminus \clos{\{x\}}
\)
is open, while $V \setminus \{x\}$ is not,
a contradiction.

We thus get $\widetilde\Filt_x \in \pt(\widetilde P)$ by
Lemma~\ref{lem:app:frames:sub:widefilt} (i.e. Lemma~\ref{lem:frames:sub:widefilt}).
We are left with proving
\[
\begin{array}{l l l}
  \Filt_x
& =
& \widetilde\Filt_x \comp j
\end{array}
\]

\noindent
Given $U \in \Filt_x$,
we have $x \in U \sle j(U)$ and thus $j(U) \in \widetilde\Filt_x$.
Conversely, let $U \in \Open$ such that $x \in j(U)$.
Hence there is some $V \in \Open$ such that $x \in V$
and $V \setminus \{x\} = U \setminus \{x\}$.
Since $(V \setminus \{x\}) \notin \Open$,
we have $(U \setminus \{x\}) \notin \Open$.
Hence $x \in U$ and $U \in \Filt_x$.
\end{proof}

}
\opt{full,long}{
\section{Proofs of~\S\ref{sec:geom} (\nameref{sec:geom})}
\label{sec:app:geom}

\subsection{Proofs of~\S\ref{sec:geom:th} (\nameref{sec:geom:th})}
\label{sec:app:geom:th}

\begin{lemm}[Remark~\ref{rem:geom:conn}]
\label{lem:app:geom:conn}
\hfill
\begin{enumerate}[(1)]
\item
Given $(\varphi_i \mid i \in I)$ with
$\varphi_i = \bigvee \left\{ \gamma_{i,j} \mid j \in J_i \right\}$,
let
\(
  \bigvee_{i \in I} \varphi_i
  \deq
  \bigvee \left\{ \gamma_{i,j} \mid \text{$i \in I$ and $j \in J_i$}\right\}
\).
Then
\[
\begin{array}{l l l}
  \nu \models \bigvee_{i \in I} \varphi_i
& \text{iff}
& \text{there exists $i \in I$ such that $\nu \models \varphi_i$}
\end{array}
\]

\item
Given $\varphi = \bigvee_{i \in I} \gamma_i$
and $\psi = \bigvee_{j \in J} \gamma'_j$,
let
$\varphi \land \psi \deq \bigvee_{(i,j) \in I \times J} \gamma_i \land \gamma'_j$.
Then
\[
\begin{array}{l l l}
  \nu
  \models
  \varphi \land \psi
& \text{iff}
& \text{$\nu \models \varphi$ and $\nu \models \psi$}
\end{array}
\]
\end{enumerate}
\end{lemm}

\begin{proof}
\hfill
\begin{enumerate}[(1)]
\item
Since
\[
\begin{array}{l l l}
  \nu \models \bigvee_{i \in I} \varphi_i
& \text{iff}
& \text{there are $i \in I$ and $j \in J_i$ such that $\nu \models \gamma_{i,j}$}
\\

& \text{iff}
& \text{there is $i \in I$ such that $\nu \models \varphi_i$}
\end{array}
\]

\item
Since
\[
\begin{array}{l l l}
  \nu \models \varphi \land \psi
& \text{iff}
& \text{there is $(i,j) \in I \times J$ such that $\nu \models \gamma_i \land \gamma'_j$}
\\

& \text{iff}
& \text{there is $(i,j) \in I \times J$ such that
  $\nu \models \gamma_i$ and $\nu \models \gamma'_i$}
\\

& \text{iff}
& (\text{there is $i \in I$ such that $\nu \models \gamma_i$})
  ~\text{and}~
  (\text{there is $j \in J$ such that $\nu \models \gamma'_j$})
\\

& \text{iff}
& \text{$\nu \models \varphi$ and $\nu \models \psi$}
\end{array}
\]
\end{enumerate}
\end{proof}

\subsection{Proofs of~\S\ref{sec:geom:spaces} (\nameref{sec:geom:spaces})}
\label{sec:app:geom:spaces}

Let $X = (X,\leq_X)$ be an algebraic dcpo.
Recall the  geometric theory $\th T(X)$
over $\At = \Fin(X)$, namely
\[
\begin{array}{c}

  d \thesis d'
\quad\text{(if $d' \leq_X d$)}

\qquad\quad

  \thesis \bigvee \Fin(X)

\qquad\quad

  d \land d'
  \thesis
  \bigvee \left\{
  d'' \in \Fin(X) \mid \text{$d \leq_X d''$ and $d' \leq_X d''$}
  \right\}
\end{array}
\]

\noindent
where $d,d' \in \Fin(X)$.

Recall also that given $x \in X$,
$\nu(x) \colon \At \to \two$ is the characteristic function
of $\{d \in \Fin(X) \mid d \leq_X x\}$.

\begin{prop}[Proposition~\ref{prop:geom:algdcpo:scott}]
\label{prop:app:geom:algdcpo:scott}
Let $X$ be an algebraic dcpo.
The bijection $x \mapsto \nu(x)$
of Proposition~\ref{prop:geom:algdcpo} extends to an homeomorphism from
$X$ to $\Mod(\th T(X))$.
\end{prop}

\begin{proof}
By~\cite[Proposition 1.21(1)]{ac98book},
the finite elements of $\Idl(\Fin(X))$ are the principal ideals
(i.e.\ those of the form $\down_{\Fin(X)} d = \{d' \in \Fin(X) \mid d' \leq_X d\}$).
Hence,
the order-isomorphism $X \cong \Idl(\Fin(X))$
of Lemma~\ref{lem:geom:algdcpo}
takes a basic Scott-open $\up d \in \Open(X)$
to the basic Scott-open
$\Up ( \down_{\Fin(X)} d ) \in \Open(\Idl(\Fin(X)))$,
where
\[
\begin{array}{l l l}
  \Up \left( \down_{\Fin(X)} d \right)
& =
& \left\{
  J \in \Idl(\Fin(X)) \mid \down_{\Fin(X)} d \sle J
  \right\}
\end{array}
\]

Now, writing $\nu \colon \At \to \two$ for the characteristic
function of $J \in \Idl(\Fin(X))$,
we have $\down_{\Fin(X)} d \sle J$ if, and only if,
$\nu \models d$.
In other words, 
under Lemma~\ref{lem:geom:algdcpo},
the basic opens of $X$ correspond exactly to the $\gmod_{\Mod(\th T(X))}(d)$, 
for $d$ an atomic proposition over $\At = \Fin(X)$.
This directly extends to Scott-opens $U \in \Open(X)$
on the one hand, and
opens $\gmod_{\Mod(\th T(X))}(\varphi) \in \Open(\Mod(\th T(X)))$
on the other.
\end{proof}

\subsubsection{Proof of Theorem~\ref{theo:geom:pt}}
We shall now prove Theorem~\ref{theo:geom:pt}.

\begin{theo}[Theorem~\ref{theo:geom:pt}]
\label{theo:app:geom:pt}
Let $\th{T}$ be a geometric theory over $\At$.
The function taking 
$\nu \in \Mod(\th T)$
to
$\{ \class\varphi_{\Mod(\th{T})} \mid \nu \models \varphi \}$
is an homeomorphism
from $\Mod(\th T)$ to $\pt(\Geom(\At)/\Mod(\th{T}))$.
\end{theo}

Fix a geometric theory $\th{T}$ over $\At$.
For notational simplicity, we let $M \deq \Mod(\th{T})$ and $L \deq \Geom(\At)/M$.
But beware that the proof of
Theorem~\ref{theo:app:geom:pt} (i.e.\ Theorem~\ref{theo:geom:pt})
relies on the theory $\th{T}$.

Define
\[
\begin{array}{l l r c l}
  f
& :
& (\At \to \two)
& \longto
& \Po(L)
\\

&
& \nu
& \longmapsto
& \left\{ \class\varphi_M \mid \nu \models \varphi \right\}
\end{array}
\]

\noindent
Note that $f$ is injective since $f(\nu) = f(\nu')$
implies that for all $p \in \At$ we have
$\nu \models p$ if, and only if, $\nu' \models p$.

\begin{lemm}
Given $\nu \colon \At \to \two$, we have
$\nu \in M$ if, and only if, $f(\nu) \in \pt(L)$.
\end{lemm}

\begin{proof}
Assume first that $\nu \in M$. We show that $f(\nu) \in \pt(L)$.
First, if $\class\varphi_M \leq_M \class\psi_M$ and
$\class\varphi_M \in f(\nu)$,
then $\nu \models \varphi$.
Hence $\nu \models \psi$ since $\gmod_M(\varphi) \sle \gmod_M(\psi)$
and $\nu \in M$.
It follows that $\class\varphi_M \in f(\nu)$.
Second, assume $\class\varphi_M, \class\varphi_M \in f(\nu)$.
Then $\nu \models \varphi$ and $\nu \models \psi$.
It follows that $\nu \models \varphi \land \psi$,
and by Lemma~\ref{lem:geom:spaces:frame} we get
$\class\varphi_M \wedge \class\psi_M \in f(\nu)$.
Finally, assume $\bigvee_{i \in I} \class{\varphi_i}_M \in f(\nu)$.
By Lemma~\ref{lem:geom:spaces:frame} we get that
$\nu \models \bigvee_{i \in I} \varphi_i$.
Hence for some $i \in I$ we have $\nu \models \varphi_i$,
and thus $\class{\varphi_i}_M \in f(\nu)$.

Conversely, assume that $f(\nu) \in \pt(L)$. We show that $\nu \in M$.
Let $\varphi \thesis \psi$ be a sequent of $\th{T}$.
Since $\varphi \preceq_M \psi$, we have
$\class\psi_M \in f(\nu)$ whenever $\class\varphi_M \in f(\nu)$.
Hence, we have $\nu \models \psi$ whenever $\nu \models \varphi$.
\end{proof}

Given $\Filt \in \pt(L)$,
let $\nu_\Filt$ be the valuation which takes $p \in \At$ to $1$
iff $\class p_M \in \Filt$.

\begin{lemm}
\label{lem:app:geom:modpt}
For each $\varphi \in \Geom(\At)$, we have
$\class\varphi_M \in \Filt$ if, and only if,
$\nu_\Filt \models \varphi$.
\end{lemm}

\begin{proof}
We first show by induction on $\gamma \in \Conj(\At)$
that $\nu_{\Filt} \models \gamma$ exactly when
$\class{\bigvee\{\gamma\}}_M \in \Filt$
(thus making explicit that $\bigvee\{\gamma\}$ is
the conjunctive formula $\gamma \in \Conj(\At)$
seen as a geometric formula).
\begin{description}
\item[Case of $p \in \At$.]
Since by definition of $\nu_\Filt$, we have
$\nu_\Filt \models p$ if, and only if,
$\class{p}_M = \class{\bigvee\{p\}}_M \in \Filt$.

\item[Case of $\true$.]
On the one hand, we have $\nu_\Filt \models \true$.
On the other hand, we have $\class{\bigvee\{\true\}}_M \in \Filt$
by Lemma~\ref{lem:geom:spaces:frame}.

\item[Case of $\gamma_1 \land \gamma_2$.]
We have $\nu_\Filt \models \gamma_1 \land \gamma_2$ iff
$\nu_\Filt \models \gamma_1$
and
$\nu_\Filt \models \gamma_2$.
By induction hypothesis, for $i = 1,2$ we have
$\nu_\Filt \models \gamma_i$ iff
$\class{\bigvee\{\gamma_i\}}_M \in \Filt$.
On the other hand,
by Lemma~\ref{lem:geom:spaces:frame}
we have
\[
\begin{array}{l l l}
  \class{\bigvee\{\gamma\}}_M \land \class{\bigvee\{\gamma'\}}_M
& =
& \class{\bigvee\{\gamma \land \gamma' \}}_M
\end{array}
\]

\noindent
Hence $\class{\bigvee\{\gamma \land \gamma' \}}_M \in \Filt$
if, and only if,
$\class{\bigvee\{\gamma\}}_M, \class{\bigvee\{\gamma'\}}_M \in \Filt$.
It follows that $\nu_\Filt \models \gamma_1 \land \gamma_2$
iff $\class{\bigvee\{\gamma \land \gamma' \}}_M \in \Filt$.
\end{description}

We now consider the case of $\varphi = \bigvee_{i \in I} \gamma_i$.
Note that for \emph{every} $\nu$ of $\At$, we have
\[
\begin{array}{l l l}
  \nu \models \varphi
& \text{iff}
& \text{there is $i \in I$ such that $\nu \models \gamma_i$}
\\

& \text{iff}
& \text{there is $i \in I$ such that $\nu \models \bigvee\{\gamma_i\}$}
\\

& \text{iff}
& \nu \models \bigvee_{i \in I} \bigvee\{\gamma_i\}
\end{array}
\]

\noindent
where $\bigvee_{i \in I} \bigvee\{\gamma_i\}$ is the operation on
\emph{geometric} formulae of
Lemma~\ref{lem:app:geom:conn} (i.e.\ Remark~\ref{rem:geom:conn}).
Hence
by Lemma~\ref{lem:geom:spaces:frame} we have
\[
\begin{array}{l l l}
  \class\varphi_M
& =
& \bigvee_{i \in I} \class{\bigvee\{\gamma_i \}}_M
\end{array}
\]

\noindent
We finally conclude with the following, which relies
on the above inductive property on conjunctive formulae:
\[
\begin{array}{l l l}
  \nu_\Filt \models \varphi
& \text{iff}
& \text{there is $i \in I$ such that $\nu_\Filt \models \gamma_i$}
\\

& \text{iff}
& \text{there is $i \in I$ such that $\class{\bigvee\{\gamma_i\}}_M \in \Filt$}
\\

& \text{iff}
& \class{\varphi}_M \in \Filt
\end{array}
\]
\end{proof}

%
%
%
%
%
%
%

Hence, given $\Filt \in \pt(L)$ we have
\[
\begin{array}{l l l}
  f(\nu_\Filt)
& =
& \left\{
  \class\varphi_M \mid
  \nu_\Filt \models \varphi
  \right\}
\\

& =
& \left\{
  \class\varphi_M \mid
  \class\varphi_M \in \Filt
  \right\}
\\

& =
& \Filt
\end{array}
\]

\noindent
It follows that we have a bijection
\[
\begin{array}{l l r c l}
  f
& :
& M
& \longto
& \pt(L)
\\

&
& \nu
& \longmapsto
& \left\{ \class\varphi_M \mid \nu \models \varphi \right\}
\end{array}
\]

\noindent
We can now prove
Theorem~\ref{theo:app:geom:pt} (i.e.\ Theorem~\ref{theo:geom:pt}).

\begin{proof}[Proof Theorem~\ref{theo:app:geom:pt}]
It remains to prove that $f \colon M \to \pt(\Geom(\At)/M)$
is an homeomorphism.

Write $g$ for the frame isomorphism
$(\Geom(\At)/M,\leq_M) \to (\Open(M),\sle)$
induced by 
$\gmod_M \colon \Geom(\At) \to \Open(M)$.
Note that
\[
\begin{array}{*{5}{l}}
  g(\class\varphi_M)
& =
& \gmod_M(\varphi)
& =
& \left\{
  \nu \in M \mid \nu \models \varphi
  \right\}
\end{array}
\]

Recall from~\S\ref{sec:frames:frames}
the unit at $(M,\Open(M))$ of the adjunction $\Open \adj \pt$,
namely
\[
\begin{array}{l l r c l}
  \eta_M
& :
& M
& \longto
& \pt(\Open(M))
\\

&
& \nu
& \longmapsto
& \left\{ U \in \Open(M) \mid \nu \in U \right\}
\end{array}
\]

\noindent
Since $g$ is an isomorphism, we have
\[
\begin{array}{l l l}
  \eta_M(\nu)
& =
& \left\{ g(\class\varphi_M) \mid \nu \in g(\class\varphi_M) \right\}
\\

& =
& \left\{ g(\class\varphi_M) \mid \nu \models \varphi \right\}
\end{array}
\]

The underlying function of $\eta_M$
thus
factors as the composite $\pt(g) \comp f$,
where $\pt(g)$ stands for
the underlying bijection of the homeomorphism
$\pt(g) \colon \pt(\Open(M)) \to \pt(\Geom(\At)/M)$
(recall the contravariant action of $\pt \colon \Frm^\op \to \Top$).
It follows that the underlying function of $\eta_M$ is a bijection
as a composition of two bijections.
But by~\cite[\S II.1.6]{johnstone82book}, $\eta_M$ is then automatically
an homeomorphism.
It follows that $f = \pt(g)^{-1} \comp \eta_M$ is an homeomorphism.
\end{proof}

\subsubsection{Proof of Proposition~\ref{prop:geom:sub}}

We now turn to Proposition~\ref{prop:geom:sub}.

\begin{prop}[Proposition~\ref{prop:geom:sub}]
\label{prop:app:geom:sub}
Given geometric theories $\th{T}$ and $\th{U}$ on $\At$,
the space $\Mod_{\th{T}}(\th{U})$ is 
\emph{equal}
to the subspace induced by the inclusion $\Mod(\th{T} \cup \th{U}) \sle \Mod(\th{T})$.
\end{prop}

\begin{proof}
Write $\Open$ for the topology $\Open(\Mod(\th{T}))$
and let $P$ be the subset $\Mod(\th{T} \cup \th{U})$ of $\Mod(\th{T})$.
We just have to check that $\Open(\Mod_{\th{T}}({\th{U}}))$
is the subspace topology $\Open \restr P$.
We have
\[
\begin{array}{l l l}
  \Open\restr P
& =
& \left\{ V \cap P \mid V \in \Open \right\}
\\


& =
& \left\{
  \gmod_{\Mod(\th{T})}(\varphi) \cap P
  \mid \varphi \in \Geom(\At)
  \right\}
\\

& =
& \left\{
  \gmod_{\Mod(\th{T})}(\varphi) \cap \Mod(\th{T} \cup \th{U})
  \mid \varphi \in \Geom(\At)
  \right\}
\end{array}
\]

\noindent
On the other hand, for each $\varphi \in \Geom(\At)$,
\[
\begin{array}{l l l}
  \gmod_{\Mod(\th{T})}(\varphi) \cap \Mod(\th{T} \cup \th{U})
& =
& \left\{
  \nu \in \Mod(\th{T}) \mid \nu \models \varphi
  \right\}
  \cap \Mod(\th{T} \cup \th{U})
\\


& =
& \left\{
  \nu \in \Mod(\th{T} \cup \th{U}) \mid \nu \models \varphi
  \right\}
\\

& =
& \gmod_{\Mod(\th{T} \cup \th{U})}(\varphi)
\end{array}
\]

\noindent
It follows that
\[
\begin{array}{l l l}
  \Open\restr P
& =
& \left\{
  \gmod_{\Mod(\th{T})}(\varphi) \cap \Mod(\th{T} \cup \th{U})
  \mid \varphi \in \Geom(\At)
  \right\}
\\

& =
& \left\{
  \gmod_{\Mod(\th{T} \cup \th{U})}(\varphi)
  \mid \varphi \in \Geom(\At)
  \right\}
\\

& =
& \Open(\Mod_{\th{T}}(\th{U}))
\end{array}
\]
\end{proof}
\subsection{Proofs of~\S\ref{sec:geom:op} (\nameref{sec:geom:op})}
\label{sec:app:geom:op}

Let $(\th T_i \mid i \in I)$ be theories, all over $\At$,
with $\th T_i = \{ \psi_{i,j} \thesis \varphi_{i,j} \mid j \in J_i\}$.
\begin{enumerate}[(1)]
\item
If $I$ is finite, we let
\(
  \bigcurlyvee_{i \in I} \th T_i
  \deq
  \left\{
  \bigwedge_{i \in I} \psi_{i,f(i)}
  \thesis
  \bigvee_{i \in I} \varphi_{i,f(i)}
  \mid
  f \in \prod_{i \in I} J_i
  \right\}
\).

\item
If $I$ is infinite, and all $\th T_i$'s are antecedent-free,
we let
\(
  \bigcurlyvee_{i \in I} \th T_i
  \deq
  \left\{
  \thesis
  \bigvee_{i \in I} \varphi_{i,f(i)}
  \mid
  f \in \prod_{i \in I} J_i
  \right\}
\).
\end{enumerate}

\begin{prop}[Proposition~\ref{prop:geom:op}]
\label{prop:app:geom:op}
In both cases above, we have
(using the Axiom of Choice when $I$ is infinite)
\[
\begin{array}{l l l}
  \Mod\left(
  \bigcurlyvee_{i \in I} \th T_i
  \right)
& =
& \bigcup_{i \in I} \Mod(\th T_i)
\end{array}
\]
\end{prop}

\begin{proof}
We discuss each case separately.
\begin{enumerate}[(1)]
\item
Assume that $I$ is a finite set.

We first show that $\bigcup_{i \in I}\Mod(\th T_i)$ is included in
$\Mod(\bigcurlyvee_{i \in I} \th T_i)$.
Let $\nu \in \Mod(\th T_k)$ for some $k \in I$.
Thus, given $f \in \prod_{i \in I} J_i$,
if $\nu \models \bigwedge_{i \in I} \psi_{i,f(i)}$,
then in particular $\nu \models \psi_{k,f(k)}$.
Hence $\nu \models \varphi_{k,f(k)}$ since $\nu \in \Mod(\th T_k)$.
It follows that $\nu \models \bigvee_{i \in I}\varphi_{i,f(i)}$.

We now show the converse inclusion.
Let $\nu$
such that $\nu \notin \bigcup_{i \in I}\Mod(\th T_i)$.
Hence, for all $i \in I$ we have $\nu \notin \Mod(\th T_i)$.
It follows that for all $i \in I$ there is some $j \in J_i$ such that
$\nu$ is not a model of the sequent
$\psi_{i,j} \thesis \varphi_{i,j}$.
Since $I$ is finite, this yields some $f \in \prod_{i \in I} J_i$
such that for all $i \in I$,
$\nu$ is not a model of the sequent
$\psi_{i,f(i)} \thesis \varphi_{i,f(i)}$
(see e.g.~\cite[\S 5]{jech06book}).
This implies that $\nu$ is not a model of the sequent
$\bigwedge_{i \in I} \psi_{i,f(i)} \thesis \bigvee_{i \in I}\varphi_{i,f(i)}$.
Hence 
$\nu \notin \Mod(\bigcurlyvee_{i \in I} \th T_i)$.

\item
The case when $I$ is infinite (and all $\th T_i$'s are antecedent-free)
is proven similarly, excepted that now, for the inclusion 
$\Mod(\bigcurlyvee_{i \in I} \th T_i) \sle \bigcup_{i \in I}\Mod(\th T_i)$
we use the full Axiom of Choice 
to obtain a suitable $f \in \prod_{i \in I}J_i$
(see e.g.~\cite[\S 5]{jech06book}).
\qedhere
\end{enumerate}
\end{proof}

\subsubsection{Proofs of~\S\ref{sec:geom:trans} (\nameref{sec:geom:trans})}
\label{sec:app:geom:trans}

\begin{lemm}[Lemma~\ref{lem:geom:trans:base}]
\label{lem:app:geom:trans:base}
Fix set $\Base$.
\begin{enumerate}[(1)]
\item
\label{item:app:geom:trans:base:next}
The map $\O\I\Next \colon \Po(\I{\Stream\Base}) \to \Po(\I{\Stream\Base})$
preserves all unions and all intersections.

\item
\label{item:app:geom:trans:base:fix}
Given $\LTL$ formulae $\Phi, \Psi$,
let $H_{\Phi,\Psi} \colon \Po(\I{\Stream\Base}) \to \Po(\I{\Stream\Base})$
take $\SP$ to $\I\Psi \cup (\I\Phi \cap \I\Next(\SP))$.
Then
\[
\begin{array}{l !{=} l !{\qquad\text{and}\qquad} l !{=} l}
  \I{\Phi \Ushort \Psi}
& \bigcup_{n \in \NN} H^n_{\Phi,\Psi}(\I\False)

& \I{\Phi \Wshort \Psi}
& \bigcap_{n \in \NN} H^n_{\Phi,\Psi}(\I\True)  
\end{array}
\]
\end{enumerate}
\end{lemm}

\begin{proof}
Recall that $\I\Next$ takes $\SP \in \Po(\I{\Stream\Base})$ to 
$\{ \stream \mid \stream \restr 1 \in \SP\}$.
\begin{enumerate}[(1)]
\item
Let $(\SP_i \mid i \in I)$ be a set of subsets of $\I{\Stream\Base}$.
Given $\stream \in \I{\Stream\Base}$, we have
\[
\begin{array}{l l l}
  \stream \in \I\Next\left(\bigcup_{i \in I} \SP_i \right)
& \text{iff}
& \stream \restr 1 \in \bigcup_{i \in I} \SP_i
\\

& \text{iff}
& \text{there is $i \in I$ such that $\stream \restr 1 \in \SP_i$}
\\

& \text{iff}
& \text{there is $i \in I$ such that $\stream \in \I\Next \SP_i$}
\\

& \text{iff}
& \stream \in \bigcup_{i \in I} \I\Next \SP_i
\end{array}
\]

\noindent
Preservation of intersections is similar.

\item
Fix $\LTL$ formulae $\Phi$ and $\Psi$.
We have seen in Remark~\ref{rem:ltl:fix} that
$\I{\Phi \Ushort \Psi}$ and $\I{\Phi \Wshort \Psi}$
are respectively the least and the greatest fixpoints of
$H_{\Phi,\Psi} \colon \Po(\I{\Stream\Base}) \to \Po(\I{\Stream\Base})$.

Now, it follows from item~(\ref{item:app:geom:trans:base:next})
that $H_{\Phi,\Psi}$ preserves all unions and all intersections.
In particular, $H_{\Phi,\Psi}$ is Scott-continuous
($(\Po(\I{\Stream\Base}),\sle)$ is a complete lattice and thus in
particular a cpo),
and~\cite[Theorem 8.15]{dp02book} gives
\[
\begin{array}{l l l}
  \I{\Phi \Ushort \Psi}
& =
& \bigcup_{n \in \NN} H^n_{\Phi,\Psi}(\I\False)
\end{array}
\]

\noindent
The case of $\I{\Phi \Wshort \Psi}$ is obtained dually
(since $H_{\Phi,\Psi}$ preserves all intersections, it is
a Scott-continuous endo-function on the cpo $(\I{\Stream\Base},\sge)$).
\qedhere
\end{enumerate}
\end{proof}

\begin{remark}
\label{rem:app:geom:trans:next}
Note that given $d \in \Fin(\I{\Stream\Base})$,
we have $(\bot \cdot d) \leq_{\I{\Stream\Base}} \stream$
if, and only if, $d \leq_{\I{\Stream\Base}} \stream\restr 1$.
Hence $\I\Next(\up d) = \up(\bot \cdot d)$.

Now, recall from Example~\ref{ex:topo:stream}
that Scott-opens $U \sle \I{\Stream\Base}$ are unions of sets
of the form $\up d$ for $d \in \Fin(\I{\Stream\Base})$.
Hence, it follows from
Lemma~\ref{lem:app:geom:trans:base}(\ref{item:app:geom:trans:base:next})
that $\I\Next(U)$ is Scott-open whenever so is $U$.

A second application of
Lemma~\ref{lem:app:geom:trans:base}(\ref{item:app:geom:trans:base:next})
implies that $\I\Next(\SP)$ is a countable intersection of Scott-opens
whenever so is $\SP$.
\end{remark}

\begin{lemm}[Remark~\ref{rem:geom:trans:class}]
\label{rem:app:geom:trans:class}
Fix a set $\Base$.
\hfill
\begin{enumerate}[(1)]
\item
\label{item:app:geom:trans:class:open}
If $\Phi_1 \in G$, then $\I{\Phi_1}$ is Scott-open in $\I{\Stream\Base}$.

\item
If $\Phi_2 \in G_\delta$, then $\I{\Phi_2}$ is a countable intersection
of Scott-opens (i.e.\ a $G_\delta$ subset of $\I{\Stream\Base}$).
\end{enumerate}
\end{lemm}

\begin{proof}
We handle each case separately.
\begin{enumerate}[(1)]
\item
We reason by induction on $\Phi_1 \in G$.
In the case of the atomic formula $a \in \Base$, note that
$\I{a}$ is the basic open set
with $\up (a \cdot \bot^\omega)$.
For the propositional connectives, use the induction hypothesis
and the stability of open sets under (finite) unions and intersections.
The case of $\Next$ follows from Remark~\ref{rem:app:geom:trans:next}.

It remains to deal with $\Phi_1 \Ushort \Psi_1$.
Assume $\I{\Phi_1}$ and $\I{\Psi_1}$ Scott-open.
By
Lemma~\ref{lem:app:geom:trans:base}(\ref{item:app:geom:trans:base:fix}),
we have
\[
\begin{array}{l l l}
  \I{\Phi \Ushort \Psi}
& =
& \bigcup_{n \in \NN} H^n_{\Phi,\Psi}(\I\False)
\end{array}
\]

\noindent
where $H_{\Phi,\Psi} \colon \Po(\I{\Stream\Base}) \to \Po(\I{\Stream\Base})$
takes $\SP$ to $\I\Psi \cup (\I\Phi \cap \I\Next(\SP))$.
Note that $H_{\Phi,\Psi}(U)$ is Scott-open whenever so is $U$.
Since $\I\False = \emptyset$ is Scott-open,
it follows by induction on $n \in \NN$ that each
$H^n_{\Phi,\Psi}(\I\False)$ is Scott-open.
Hence $\I{\Phi \Ushort \Psi}$ is Scott-open.

\item
We reason by induction on $\Phi_2 \in G_\delta$.
The argument is similar to that of item~(\ref{item:app:geom:trans:class:open})
using that an open set is (trivially) an countable intersection
of opens, and that countable intersections of opens are stable under
finite unions and intersections.

In the case of $\Phi_2 \Wshort \Psi_2$, by 
Lemma~\ref{lem:app:geom:trans:base}(\ref{item:app:geom:trans:base:fix})
we have
\[
\begin{array}{l l l}
  \I{\Phi \Wshort \Psi}
& =
& \bigcap_{n \in \NN} H^n_{\Phi,\Psi}(\I\True)
\end{array}
\]

\noindent
Reasoning similarly as for item~(\ref{item:app:geom:trans:class:open}),
since $\I\True = \I{\Stream\Base}$ is open, we get
that $H^n_{\Phi,\Psi}(\I\True)$
is a countable intersection of opens for all $n \in \NN$.
But a countable intersection of countable intersections is a countable
intersection. Hence $\I{\Phi \Wshort \Psi}$ is a countable intersection of
opens.
\qedhere
\end{enumerate}
\end{proof}


\begin{lemm}[Remark~\ref{rem:geom:trans:equiv}]
\label{lem:app:geom:trans:equiv}
Recall that $\Phi \equiv \Psi$ means $\I\Phi = \I\Psi$.
Given $\LTL$ formulae $\Phi$ and $\Psi$, we have
\[
\begin{array}{l !{\equiv} l !{\qquad} l !{\equiv} l !{\qquad} l !{\equiv} l}
  \Next\False
& \False

& \Next(\Phi \lor \Psi)
& \Next\Phi \lor \Next\Psi

& \Next(\Phi \Ushort \Psi)
& (\Next\Phi) \Ushort (\Next\Psi)
\\

  \Next\True
& \True

& \Next(\Phi \land \Psi)
& \Next\Phi \land \Next \Psi

& \Next(\Phi \Wshort \Psi)
& (\Next\Phi) \Wshort (\Next\Psi)
\end{array}
\]
\end{lemm}

\begin{proof}
First, 
Lemma~\ref{lem:app:geom:trans:base}(\ref{item:app:geom:trans:base:next})
directly yields the laws
\[
\begin{array}{l !{\equiv} l !{\qquad} l !{\equiv} l}
  \Next\False
& \False

& \Next(\Phi \lor \Psi)
& \Next\Phi \lor \Next\Psi
\\

  \Next\True
& \True

& \Next(\Phi \land \Psi)
& \Next\Phi \land \Next \Psi
\end{array}
\]

We discuss
\[
\begin{array}{l !{\equiv} l !{\qquad\text{and}\qquad} l !{\equiv} l}
  \Next(\Phi \Ushort \Psi)
& (\Next\Phi) \Ushort (\Next\Psi)

& \Next(\Phi \Wshort \Psi)
& (\Next\Phi) \Wshort (\Next\Psi)
\end{array}
\]

\noindent
It follows from
Lemma~\ref{lem:app:geom:trans:base}(\ref{item:app:geom:trans:base:next})
that $\I\Next H_{\Phi,\Psi}(\SP) = H_{\Next\Phi,\Next\Psi}(\I\Next\SP)$.
Then, by induction on $n \in \NN$ we obtain
\[
\begin{array}{*{5}{l}}
  \I\Next H^n_{\Phi,\Psi}(\I\False)
& =
& H^n_{\Next\Phi,\Next\Psi}(\I{\Next\False})
& =
& H^n_{\Next\Phi,\Next\Psi}(\I\False)
\\

  \I\Next H^n_{\Phi,\Psi}(\I\True)
& =
& H^n_{\Next\Phi,\Next\Psi}(\I{\Next\True})
& =
& H^n_{\Next\Phi,\Next\Psi}(\I\True)
\end{array}
\]

\noindent
Hence, using
Lemma~\ref{lem:app:geom:trans:base}(\ref{item:app:geom:trans:base:next})
and
Lemma~\ref{lem:app:geom:trans:base}(\ref{item:app:geom:trans:base:fix}),
we conclude that
\[
\begin{array}{*{7}{l}}
  \I{\Next(\Phi \Ushort \Psi)}
& =
& \I\Next \bigcup_{n \in \NN} H^n_{\Phi,\Psi}(\I\False)
& =
& \bigcup_{n \in \NN} H^n_{\Next\Phi,\Next\Psi}(\I\False)
& =
& \I{(\Next\Phi) \Ushort (\Next\Psi)}
\\

  \I{\Next(\Phi \Wshort \Psi)}
& =
& \I\Next \bigcap_{n \in \NN} H^n_{\Phi,\Psi}(\I\True)
& =
& \bigcap_{n \in \NN} H^n_{\Next\Phi,\Next\Psi}(\I\True)
& =
& \I{(\Next\Phi) \Wshort (\Next\Psi)}
\end{array}
\]
\end{proof}

}
\opt{full}{
\section{Proofs of~\S\ref{sec:free} (\nameref{sec:free})}
\label{sec:app:free}

\begin{lemm}[Lemma~\ref{lem:free:frame}]
\label{lem:app:free:frame}
$(\Geom(\At)/\th{T}, \leq_{\th{T}})$ is a frame.
\end{lemm}
\begin{proof}
The argument is mostly a direct inspection of Figure~\ref{fig:free}.

First, we have that $\class{\bigvee_i \varphi_i}_{\th{T}}$ is an upper bound of
$(\class{\varphi_i}_{\th{T}})_i$
since
$\varphi_i \thesis \bigvee_i \varphi_i$ for all $i$
(rule $\rn{$\bigvee$-R}$).
It is a least upper bound since $\bigvee_i \varphi_i \thesis \psi$
whenever $\varphi_i \thesis \psi$ for all $i$
(rule $\rn{$\bigvee$-L}$).

Similarly, $\class{\varphi \land \psi}_{\th{T}}$
is a lower bound of $\class\varphi_{\th{T}}$ and $\class\psi_{\th{T}}$
since $\varphi \land \psi \thesis \varphi$
and $\varphi \land \psi \thesis \psi$
(rules $\rn{$\land$-L$_1$}$ and $\rn{$\land$-L$_2$}$).
It is a greatest lower bound since $\theta \thesis \varphi \land \psi$
whenever $\theta \thesis \varphi$ and $\theta \thesis \psi$
(rule $\rn{$\land$-R}$).

Moreover, the rule $\rn{$\true$-R}$ yields
that $\class\varphi_{\th{T}} \leq_{\th{T}} \class\true_{\th{T}}$
for all $\varphi$.

Hence,
$(\Geom(\At)/\th{T}, \leq_{\th{T}})$ is a complete lattice
whose sups and binary infs are respectively given by
\[
\begin{array}{r c l !{=} l}
  (\class{\varphi_i}_{\th{T}})_i
& \longmapsto
& \bigvee_i \class{\varphi_i}_{\th{T}}
& \class{\bigvee_i \varphi_i}_{\th{T}}
\\

  \class\varphi_{\th{T}}
  ,~
  \class\psi_{\th{T}}
& \longmapsto
& \class\varphi_{\th{T}} \land \class\psi_{\th{T}}
& \class{\varphi \land \psi}_{\th{T}}
\end{array}
\]

%

It remains to prove frame distributivity,
namely
\[
\begin{array}{l l l}
  \bigvee_{i \in I} \left(
  \class\psi_{\th{T}} \land \class{\varphi_i}_{\th{T}}
  \right)
& =
& \class\psi_{\th{T}} \land \bigvee_{i \in I}\class{\varphi_i}_{\th{T}}
\end{array}
\]

\noindent
We reason on the syntax of geometric formulae.
Assume
$\psi = \bigvee \{\gamma'_k \mid k \in K \}$
and
$\varphi_i = \bigvee \{ \gamma_{i,j} \mid j \in J_i \}$
for each $i \in I$.
Then, unfolding the notations of
Remark~\ref{rem:geom:conn} (i.e.\ Lemma~\ref{lem:app:geom:conn}),
we have
\[
\begin{array}{r c l}
  \psi \land \bigvee_{i \in I} \varphi_i
& =
& \left( \bigvee \left\{ \gamma'_k \mid k \in K \right\} \right)
  \land
  \left(
  \bigvee \left\{
  \gamma_{i,j} \mid \text{$i \in I$ and $j \in J_i$}
  \right\}
  \right)
\\

& =
& \bigvee \left\{
  \gamma'_k \land \gamma_{i,j}
  \mid
  \text{$k \in K$, $i \in I$ and $j \in J_i$}
  \right\}
\end{array}
\]

\noindent
On the other hand
\[
\begin{array}{r c l}
  \psi \land \varphi_i
& =
& \bigvee \left\{
  \gamma'_k \land \gamma_{i,j}
  \mid
  \text{$k \in K$ and $j \in J_i$}
  \right\}
\\

  \bigvee_{i \in I} \left( \psi \land \varphi_i \right)
& =
& \bigvee \left\{
  \gamma'_k \land \gamma_{i,j}
  \mid
  \text{$i \in I$, $k \in K$ and $j \in J_i$}
  \right\}
\end{array}
\]

\noindent
It follows that
\[
\begin{array}{l l l}
  \bigvee_{i \in I} \left(
  \psi \land \varphi_i
  \right)
& =
& \psi \land \bigvee_{i \in I}\varphi_i
\end{array}
\]

\noindent
and we are done.
\end{proof}

For some proofs in this~\S\ref{sec:app:free},
we make explicit that 
$\bigvee\{\gamma\}$ is the conjunctive formula $\gamma$
seen as a conjunctive formula.
The following observation will be useful several times.

\begin{remark}
\label{rem:app:free:geom}
Given a theory $\th{T}$ and
given $\varphi = \bigvee_{i \in I} \gamma_i$,
we have
\[
\begin{array}{l l l}
  \class{\varphi}_{\th{T}}
& =
& \bigvee_{i \in I} \class{\bigvee\{\gamma_i\}}_{\th{T}}
\end{array}
\]
\end{remark}

\begin{proof}
Note that given $\SP \sle \Conj(\At)$,
making explicit the $\bigvee\{\gamma\}$'s
in the rules $\rn{$\bigvee$-L}$ and $\rn{$\bigvee$-R}$
leads to the following instances
\[
\begin{array}{c}

\dfrac{\text{for all $\gamma \in \SP$, $\bigvee\{\gamma\} \thesis_{\th T} \psi$}}
  {\bigvee \SP \thesis_{\th T} \psi}

\qquad\qquad

\dfrac{}
  {\bigvee\{\gamma\} \thesis_{\th T} \bigvee \SP}
~\text{(if $\gamma \in \SP$)}

\end{array}
\]

\noindent
Hence by
Lemma~\ref{lem:app:free:frame} (i.e.\ Lemma~\ref{lem:free:frame})
we have
\[
\begin{array}{l l l}
  \class{\varphi}_{\th{T}}
& =
& \bigvee_{i \in I} \class{\bigvee\{\gamma_i\}}_{\th{T}}
\end{array}
\]
\end{proof}

The following simple property is mentioned in the text of~\S\ref{sec:free}.

\begin{lemm}
\label{lem:app:free:sound}
Let $\th{T}$ be a geometric theory over $\At$ and let $M \deq \Mod(\th{T})$.
Then for all $\varphi,\psi \in \Geom(\At)$, we have
\[
\begin{array}{l l l}
  \psi \thesis_{\th{T}} \varphi
& \longimp
& \psi \preceq_M \varphi
\end{array}
\]
\end{lemm}

\begin{proof}
The proof is a simple induction on $\psi \thesis_{\th{T}} \varphi$,
using Remark~\ref{rem:geom:conn} (i.e.\ Lemma~\ref{lem:app:geom:conn})
for the logical rules.
The case of the rule $\rn{Th}$
follows from the fact that
$(\psi \thesis \varphi) \in \th{T}$ implies
$\psi \preceq_M \varphi$.
The cases of $\rn{Ax}$ and $\rn{Cut}$
follow from the fact that
both $\thesis_{\th{T}}$ and $\preceq_M$ are preorders.
\end{proof}

We now turn Lemma~\ref{lem:free:filtmod}.
Let $\th{T}$ be a theory over $\At$.
Let $\Filt$ be a point
of $\Geom(\At)/\th{T}$
such that $\class\psi_{\th{T}} \in \Filt$ and $\class\varphi_{\th{T}} \notin \Filt$.
Let $\nu \colon \At \to \two$ take $p \in \At$ to $1$ iff
$\class{p}_{\th{T}} \in \Filt$.

\begin{lemm}[Lemma~\ref{lem:free:filtmod}]
\label{lem:app:free:filtmod}
Let $\Filt$ and $\nu$ as above.
Then for every $\theta \in \Geom(\At)$,
we have $\nu \models \theta$ if, and only if, $\class\theta_{\th{T}} \in \Filt$.

In particular, $\nu$ is a model of $\th{T}$ with
$\nu \models \psi$ and $\nu \not\models \varphi$.
\end{lemm}

\begin{proof}
We first show that 
for every $\theta \in \Geom(\At)$,
we have $\nu \models \theta$ if, and only if, $\class\theta_{\th{T}} \in \Filt$.
The proof is similar to that of
Lemma~\ref{lem:app:geom:modpt} (\S\ref{sec:app:geom:spaces}).
We first show by induction on $\gamma \in \Conj(\At)$
that $\nu \models \gamma$ if, and only if
$\class{\bigvee\{\gamma\}}_{\th{T}} \in \Filt$
(thus making explicit that $\bigvee\{\gamma\}$ is
the conjunctive formula $\gamma \in \Conj(\At)$
seen as a geometric formula).
\begin{description}
\item[Case of $p \in \At$.]
Since by definition of $\nu$, we have
$\nu \models p$ if, and only if
$\class{p}_{\th{T}} = \class{\bigvee\{p\}}_{\th{T}} \in \Filt$.

\item[Case of $\true$.]
On the one hand, we have $\nu \models \true$.
On the other hand, we have $\class{\bigvee\{\true\}}_{\th{T}} \in \Filt$
by Lemma~\ref{lem:app:free:frame} (i.e.\ Lemma~\ref{lem:free:frame}).

\item[Case of $\gamma_1 \land \gamma_2$.]
We have $\nu \models \gamma_1 \land \gamma_2$ iff
$\nu \models \gamma_1$
and
$\nu \models \gamma_2$.
By induction hypothesis, for $i = 1,2$ we have
$\nu \models \gamma_i$ iff
$\class{\bigvee\{\gamma_i\}}_{\th{T}} \in \Filt$.
On the other hand, by
Lemma~\ref{lem:app:free:frame} (i.e.\ Lemma~\ref{lem:free:frame})
we have
\[
\begin{array}{l l l}
  \class{\bigvee\{\gamma\}}_{\th{T}} \land \class{\bigvee\{\gamma'\}}_{\th{T}}
& =
& \class{\bigvee\{\gamma \land \gamma' \}}_{\th{T}}
\end{array}
\]

\noindent
Hence $\class{\bigvee\{\gamma \land \gamma' \}}_{\th{T}} \in \Filt$
if, and only if,
$\class{\bigvee\{\gamma\}}_{\th{T}}, \class{\bigvee\{\gamma'\}}_{\th{T}} \in \Filt$.
It follows that $\nu \models \gamma_1 \land \gamma_2$
iff $\class{\bigvee\{\gamma \land \gamma' \}}_{\th{T}} \in \Filt$.
\end{description}

We now consider the case of $\varphi = \bigvee_{i \in I} \gamma_i$.
By Remark~\ref{rem:app:free:geom} we have
\[
\begin{array}{l l l}
  \class{\varphi}_{\th{T}}
& =
& \bigvee_{i \in I} \class{\bigvee\{\gamma_i\}}_{\th{T}}
\end{array}
\]

\noindent
Then, using the above inductive property on conjunctive formulae,
we get
\[
\begin{array}{l l l}
  \nu \models \varphi
& \text{iff}
& \text{there is $i \in I$ such that $\nu \models \gamma_i$}
\\

& \text{iff}
& \text{there is $i \in I$ such that $\class{\bigvee\{\gamma_i\}}_{\th{T}} \in \Filt$}
\\

& \text{iff}
& \class{\varphi}_{\th{T}} \in \Filt
\end{array}
\]

For the second part of the statement, given
$(\theta_1 \thesis \theta_2) \in \th{T}$
we have
$\class{\theta_1}_{\th{T}} \leq_{\th{T}} \class{\theta_2}_{\th{T}}$.
Hence, if $\nu \models \theta_1$,
then $\class{\theta_1}_{\th{T}} \in \Filt$,
so that 
$\class{\theta_2}_{\th{T}} \in \Filt$
and $\nu \models \theta_2$.
It follows that $\nu$ is a model of $\th{T}$.

Finally, we get $\nu \models \psi$ and $\nu \not\models \varphi$
since $\class\psi_{\th{T}} \in \Filt$
and
$\class\varphi_{\th{T}} \notin \Filt$
by assumption.
\end{proof}

\subsection{Proof of Proposition~\ref{prop:free:frame:free}}

\begin{prop}[Proposition~\ref{prop:free:frame:free}]
\label{prop:app:free:frame:free}
$(\Geom(\At)/\th{E},\leq_{\th{E}})$
(together with the function $p \in \At \mapsto \class{p}_{\th{T}}$)
is a free frame on $\At$.
\end{prop}

Fix a set $\At$.
Consider a frame $L$ and a function $f \colon \At \to \Forget L$.
We have to show that
there is a unique frame morphism
$\ladj f \colon \Geom(\At)/\th{E} \to L$
such that the following commutes.
\[
\begin{tikzcd}[row sep=tiny, column sep=large]
  \At
  \arrow{rr}{f}
  \arrow{dr}[below]{p \mapsto \class{p}_{\th{E}}}
&
& \Forget L
\\
& \Forget (\Geom(\At)/\th{E})
  \arrow[dashed]{ur}[below]{\Forget \ladj f}
\end{tikzcd}
\]

We first extend $f$ to the function
\[
\begin{array}{*{5}{l}}
  g_0
& :
& \Conj(\At)
& \longto
& \Forget L
\end{array}
\]

\noindent
defined by induction on $\gamma \in \Conj(\At)$ as follows:
\[
\begin{array}{r c l}
  g_0(p)
& \deq
& f(p)
\\

  g_0(\true)
& \deq
& \top_L
\\

  g_0(\gamma \land \gamma')
& \deq
& g_0(\gamma) \land_L g_0(\gamma')
\end{array}
\]

\noindent
We then extend $g_0$ to a function
\[
\begin{array}{*{5}{l}}
  g
& :
& \Geom(\At)
& \longto
& \Forget L
\end{array}
\]

\noindent
with
\[
\begin{array}{l l l}
  g\left( \bigvee \SP \right)
& \deq
& \bigvee_L \left\{ g_0(\gamma) \mid \gamma \in \SP\right\}
\end{array}
\]

\noindent
where $\SP \sle \Conj(\At)$.

\begin{lemm}
\label{lem:app:free:frame:free:conn}
With the notation of Remark~\ref{rem:geom:conn} (i.e.\ Lemma~\ref{lem:app:geom:conn}),
we have
\begin{enumerate}[(i)]
\item $g(\true) = \top_L$
\item $g(\varphi \land \psi) = g(\varphi) \land_L g(\psi)$
\item
\(
  g\left( \bigvee \left\{\varphi_i \mid i \in I\right\} \right)
  =
  \bigvee_L \left\{ g(\varphi_i) \mid i \in I \right\}
\)
\end{enumerate}
\end{lemm}

\begin{proof}
We discuss each case separately.
\begin{enumerate}[(i)]
\item
Recall that we write $\true$ for the geometric formula $\bigvee\{\true\}$.
Then we are done since $g(\bigvee\{\true\}) = \bigvee_L\{g_0(\true)\} = \top_L$.

\item
Write $\varphi = \bigvee_{i \in I}\gamma_i$
and $\psi = \bigvee_{j \in J} \gamma'_j$.
We have
\[
\begin{array}{r c l}
  g(\varphi)
& =
& \bigvee_L \left\{ g_0(\gamma_i) \mid i \in I \right\}
\\

  g(\psi)
& =
& \bigvee_L \left\{ g_0(\gamma'_j) \mid j \in J \right\}
\end{array}
\]

\noindent
and
\[
\begin{array}{r c l}
  \varphi \land \psi
& =
& \bigvee \left\{
    \gamma_i \land \gamma'_j \mid \text{$i \in I$ and $j \in J$} \right\}
\\

  g(\varphi \land \psi)
& =
& \bigvee_L \left\{
    g_0(\gamma_i) \land_L g_0(\gamma'_j) \mid \text{$i \in I$ and $j \in J$} \right\}
\end{array}
\]

\noindent
Hence, using frame distributivity in $L$ twice we get
\[
\begin{array}{l l l}
  g(\varphi \land \psi)
& =
& g(\varphi) \land_L g(\psi)
\end{array}
\]

\item
For each $i \in I$,
write $\varphi_i = \bigvee \{\gamma_{i,j} \mid \text{$i \in I$ and $j \in J_i$}\}$.
Then we are done since
\[
\begin{array}{r c l}
  g(\varphi_i)
& =
& \bigvee_L \left\{
  g_0(\gamma_{i,j}) \mid j \in J_i
  \right\}
\\

  g\left( \bigvee_{i \in I} \varphi_i \right)
& =
& \bigvee_L \left\{
  g_0(\gamma_{i,j}) \mid \text{$i \in I$ and $j \in J_i$}
  \right\}
\end{array}
\]
\qedhere
\end{enumerate}
\end{proof}

Using Lemma~\ref{lem:app:free:frame:free:conn},
a straightforward induction on derivations in $\thesis_{\th{E}}$
shows that
\[
\begin{array}{l l l}
  \psi \thesis_{\th{E}} \varphi
& \longimp
& g(\psi) \leq_L g(\varphi)
\end{array}
\]

\noindent
It follows that $\class{\varphi}_{\th{E}} = \class{\psi}_{\th{E}}$
implies $g(\varphi) = g(\psi)$.
This yields our function
\[
\begin{array}{l l r c l}
  \ladj f
& :
& \Geom(\At)/\th{E}
& \longto
& L
\\

&
& \class\varphi_{\th{E}}
& \longmapsto
& g(\varphi)
\end{array}
\]

\noindent
Lemma~\ref{lem:app:free:frame:free:conn}
implies that $\ladj f$ is a frame morphism.
Moreover, given $p \in \At$ we have
\[
\begin{array}{l l l}
  f(\class{p}_{\th E})
& =
& g(p)
\\

& =
& g(\bigvee\{p\})
\\

& =
& g_0(p)
\\

& =
& f(p)
\end{array}
\]

We can now conclude the proof of Proposition~\ref{prop:app:free:frame:free}
(i.e.\ Proposition~\ref{prop:free:frame:free}).

\begin{proof}[Proof of Proposition~\ref{prop:app:free:frame:free}]
It remains to show that $\ladj f$ is the unique frame morphism
such that $\ladj f(\class{p}_{\th{E}}) = f(p)$ for all $p \in \At$.
Let $h \colon \Geom(\At)/\th{E} \to L$ be a frame morphism such that
$h(\class{p}_{\th{E}}) = f(p)$ for all $p \in \At$.

We show that $h = f$.
We make explicit that $\bigvee\{\gamma\}$ is
the conjunctive formula $\gamma \in \Conj(\At)$
seen as a geometric formula.

We first show
by induction on $\gamma \in \Conj(\At)$ that
\[
\begin{array}{l l l}
  h(\class{\bigvee\{\gamma\}}_{\th{E}})
& =
& g_0(\gamma)
\end{array}
\]
\begin{description}
\item[Case of $p \in \At$.]

Since $h(\class{\bigvee\{p\}}_{\th E}) = f(p)$ by assumption on $h$.

\item[Case of $\true$.]
Since $h(\class{\bigvee\{\true\}}_{\th E}) = \top_L$
as $h$ is a frame morphism.

\item[Case of $\gamma \land \gamma'$.]
Note that we have
\[
\begin{array}{l l l}
  \bigvee\{\gamma\} \land \bigvee\{\gamma'\}
& =
& \bigvee\{\gamma \land \gamma'\}
\end{array}
\]

\noindent
Hence, the result follows from the induction hypothesis
and the fact that $h$ is a frame morphism.
\end{description}

\noindent
Now, given $\varphi = \bigvee_i \gamma_i$,
since $h$ is a frame morphism,
it follows from Remark~\ref{rem:app:free:geom}
that 
\[
\begin{array}{l l l}
  h(\class\varphi_{\th{E}})
& =
& h\left(
  \bigvee \left\{ \class{\bigvee\{\gamma_i\}}_{\th{E}} \mid i \in I\right\}
  \right)
\\

& =
& \bigvee_L
  \left\{
  h\left( \class{\bigvee\{\gamma_i\}}_{\th{E}} \right) \mid i \in I
  \right\}
\\

& =
& \bigvee_L
  \left\{ g_0(\gamma_i) \mid i \in I \right\}
\\

& =
& g\left( \varphi \right)
\\

& =
& f\left( \class\varphi_{\th{E}} \right)
\end{array}
\]
\end{proof}
\subsection{Proof of Proposition~\ref{prop:free:quot}}

\begin{prop}[Proposition~\ref{prop:free:quot}]
\label{prop:app:free:quot}
Given geometric theories $\th{T}, \th{U}$ over $\At$,
let
\(
  R
  \deq
  \{ (\class\varphi_{\th{T}} , \class\psi_{\th{T}}) 
  \mid (\varphi \thesis \psi) \in \th{T} \cup \th{U}\}
\).
Then $\Geom(\At)/(\th{T} \cup \th{U})$
is isomorphic to the quotient of $\Geom(\At)/\th{T}$
by $\sim_{\clos R}$.
\end{prop}

Fix geometric theories $\th{T}$ and $\th{U}$ over $\At$,
and let $R$ be as in the statement.
Consider
\[
\begin{array}{l l l}
  \widetilde R
& \deq
& \left\{
  (\class\varphi_{\th{T}} , \class\psi_{\th{T}}) 
  \mid
  \varphi \thesis_{\th{T} \cup \th{U}} \psi
  \right\}
\end{array}
\]

We are going to show that $\widetilde R = \clos R$,
i.e.\ that $\widetilde R$ is the least congruence preorder containing $R$.
It is trivial that $R \sle \widetilde R$.

\begin{lemm}
$\widetilde R$ is a congruence preorder on $\Geom(\At)/\th{T}$.
\end{lemm}

\begin{proof}
We first check that $\widetilde R$ is a preorder.
Its reflexive since $\varphi \thesis_{\th{T} \cup \th{U}} \varphi$.
We now prove that $\widetilde R$ is transitive.
Let $C \mathrel{\widetilde R} D$
and $D \mathrel{\widetilde R} E$.
Hence there are $\varphi,\theta$ such that
$C = \class\varphi_{\th{T}}$,
$D = \class\theta_{\th{T}}$
and $\varphi \thesis_{\th{T} \cup \th{U}} \theta$.
Similarly, there are $\theta',\psi$ such that
$D = \class{\theta'}_{\th{T}}$,
$E = \class\psi_{\th{T}}$
and
$\theta' \thesis_{\th{T} \cup \th{U}} \psi$.
Since $\class\theta_{\th{T}} = \class{\theta'}_{\th{T}}$,
we get (say) $\theta \thesis_{\th{T}} \theta'$.
Hence
$\theta \thesis_{\th{T} \cup \th{U}} \psi$
and thus
$C \mathrel{\widetilde R} E$.

Moreover, if $\class\varphi_{\th{T}} \leq_{\th{T}} \class\psi_{\th{T}}$,
then $\varphi \thesis_{\th{T}} \psi$,
so that $\varphi \thesis_{\th{T} \cup \th{U}} \psi$
and
$\class\varphi_{\th{T}} \mathrel{\widetilde R} \class\psi_{\th{T}}$.

Let $\psi$ and $(\varphi_i)_{i \in I}$
such that for all $i \in I$, we have
\[
\begin{array}{l l l}
  \class{\varphi_i}_{\th{T}}
& \mathrel{\widetilde R}
& \class\psi_{\th{T}}
\end{array}
\]

Since ${\thesis_{\th{T}}} \sle {\thesis_{\th{T} \cup \th{U}}}$,
this implies that for all $i \in I$, we have
$\varphi_i \thesis_{\th{T} \cup \th{U}} \psi$.
Hence using the rule $\rn{$\bigvee$-L}$, we get
$\bigvee_{i \in I}\varphi_i \thesis_{\th{T} \cup \th{U}} \psi$,
and it follows that
\[
\begin{array}{l l l}
  \bigvee_{i \in I}\class{\varphi_i}_{\th{T}}
& \mathrel{\widetilde R}
& \class\psi_{\th{T}}
\end{array}
\]

Note that for every $\varphi$, we have $\varphi \thesis_{\th{T}} \true$,
and thus
\[
\begin{array}{l l l}
  \class{\varphi}_{\th{T}}
& \mathrel{\widetilde R}
& \class\true_{\th{T}}
\end{array}
\]

Finally, assume
\[
\begin{array}{l l l !{\qquad\text{and}\qquad} l l l}
  \class{\theta}_{\th{T}}
& \mathrel{\widetilde R}
& \class\varphi_{\th{T}}

& \class{\theta}_{\th{T}}
& \mathrel{\widetilde R}
& \class\psi_{\th{T}}
\end{array}
\]

\noindent
Again since ${\thesis_{\th{T}}} \sle {\thesis_{\th{T} \cup \th{U}}}$,
this implies
$\theta \thesis_{\th{T} \cup \th{U}} \varphi$
and
$\theta \thesis_{\th{T} \cup \th{U}} \psi$.
Hence
with the rule $\rn{$\land$-R}$ we get
$\theta \thesis_{\th{T} \cup \th{U}} \varphi \land \psi$
and thus
\[
\begin{array}{l l l}
  \class{\theta}_{\th{T}}
& \mathrel{\widetilde R}
& \class{\varphi}_{\th{T}}
  \land
  \class{\psi}_{\th{T}}
\end{array}
\]
\end{proof}

\begin{lemm}
Let $Q$ be a congruence preorder on $\Geom(\At)/\th{T}$
such that $R \sle Q$.
Then $\widetilde R \sle Q$.
\end{lemm}

\begin{proof}
We show that if $\varphi \thesis_{\th{T} \cup \th{U}} \psi$ then
$(\class\varphi_{\th{T}}, \class\psi_{\th{T}}) \in Q$.
We reason by induction on the derivation of
$\varphi \thesis_{\th{T}\cup \th{U}} \psi$.
\begin{description}
\item[Case of $\rn{Th}$.]

Since $R \sle Q$.

\item[Cases of $\rn{Ax}$ and $\rn{Cut}$]

Since $Q$ is a preorder.

\item[Case of $\rn{$\true$-R}$.]

Since $Q$ is a preorder containing $\leq_{\th{T}}$.

\item[Case of $\rn{$\land$-L$_1$}$ and $\rn{$\land$-L$_2$}$.]

Since
\(
  \class{\varphi_1 \land \varphi_2}_{\th{T}}
  \leq_{\th{T}}
  \class{\varphi_i}_{\th{T}}
\)
while $Q$ contains $\leq_{\th{T}}$.

\item[Case of $\rn{$\bigvee$-L}$.]

Similar.

\item[Case of $\rn{$\land$-R}$.]

Since $Q$ is a congruence preorder, we have
\[
\begin{array}{l l l}
  \class{\theta}_{\th{T}}
& Q
& \class{\varphi}_{\th{T}}
  \land
  \class{\psi}_{\th{T}}
\end{array}
\]

\noindent
whenever
\[
\begin{array}{l l l !{\qquad\text{and}\qquad} l l l}
  \class{\theta}_{\th{T}}
& Q
& \class\varphi_{\th{T}}

& \class{\theta}_{\th{T}}
& Q
& \class\psi_{\th{T}}
\end{array}
\]

\noindent
Then conclude with the induction hypothesis.

\item[Case of $\rn{$\bigvee$-L}$.]

Similar.

\item[Case of $\rn{Dist}$.]

Since by frame distributivity in $\Geom(\At)/\th{T}$,
we have
\[
\begin{array}{l l l}
  \class{\psi \land \bigvee_{i \in I}\varphi_i}_{\th{T}}
& =
& \class{\bigvee_{i \in I}(\psi \land \varphi_i)}_{\th{T}}
\end{array}
\]
\qedhere
\end{description}
\end{proof}

Hence $\widetilde R = \clos R$, the least congruence preorder
containing $R$.
We can now conclude the proof of Proposition~\ref{prop:app:free:quot},
(i.e.\ Proposition~\ref{prop:free:quot}).

\begin{proof}[Proof of Proposition~\ref{prop:app:free:quot}]
Let $\sim$ be the equivalence relation induced by $\widetilde R$.
We have to show that $\Geom(\At)/(\th{T} \cup \th{U})$ is isomorphic
to the quotient of $\Geom(\At)/\th{T}$ by $\sim$.

Recall that ${\thesis_{\th{T}}} \sle {\thesis_{\th{T} \cup \th{U}}}$.
Note that given $\varphi,\psi \in \Geom(\At)$,
we have $\class\varphi_{\th{T}} \sim \class\psi_{\th{T}}$
precisely when $\varphi \thesis_{\th{T} \cup \th{U}} \psi$ and
$\psi \thesis_{\th{T} \cup \th{U}} \varphi$.
In other words, for all $\varphi,\psi \in \Geom(\At)$,
we have
\[
\begin{array}{l l l}
  \class\varphi_{\th{T}} \sim \class\varphi_{\th{T}}
& \text{if, and only if,}
& \class\varphi_{\th{T} \cup \th{U}} = \class\varphi_{\th{T} \cup \th{U}}
\end{array}
\]

\noindent
and we are done.
\end{proof}

}	
\opt{full,long}{
\section{Proofs of~\S\ref{sec:scott} (\nameref{sec:scott})}
\label{sec:app:scott}

Fix a finite set $\Base$ and a Scott-continuous
$p \colon \I\Base \to \I\Bool$
with $p(a) \neq \bot_{\I\Bool}$ if $a \in \Base$.
Let $\Psi = \Psi_p$ as in Example~\ref{ex:ltl:nf}
and let $\Phi \deq \bigvee_{a \in \Base} a$.

Recall from Remark~\ref{rem:prelim:filter}
that
$\I\filter p = \term{Y}(f_p) = \bigvee_{n \in \NN} f_p^n(\bot_X)$
where
\[
\begin{array}{*{7}{l}}
  f_p
& \deq
& \lambda g.\lambda \stream.~
  \term{if}~ p (\stream(0))
  ~\term{then}~ \stream(0) \cdot g(\stream \restr 1)
  ~\term{else}~ g(\stream \restr 1)
& :
& X
& \longto_{\CPO}
& X
\end{array}
\]

\noindent
and
where $X$ is the cpo
$\I{\Stream\Base} \to_{\CPO} \I{\Stream\Base}$.

Recall also the geometric formulae
\[
\begin{array}{l !{\qquad} r c l}
& \psi_{n,k}
& \deq
& \bigvee
  \left\{
  \bigwedge_{1 \leq j \leq k}
  \Next^{i_j}\th F\I\Psi
  \mid
  0 \leq i_1 < \dots < i_k < n
  \right\}
\\

  \text{and}
& \varphi_k
& \deq
& \bigwedge_{m < k} \Next^m \th F\I\Phi
\end{array}
\]
\noindent
where $k \leq n$.

We begin with the following property, which is stated in the text
of~\S\ref{sec:scott}.
Let $\stream \in \I{\Stream\Base}$.

\begin{lemm}
\label{lem:app:scott:mod}
We have
$\nu(\stream) \in \Mod\left(\th T\I{\Box\Diam\Psi} \right)$
if, and only if,
$(\forall k \in \NN)(\exists n \geq k)(\nu(\stream) \models \psi_{n,k})$.
\end{lemm}

\begin{proof}
Let $\stream \in \I{\Stream\Base}$ and write $\nu$ for $\nu(\stream)$.
Recall from Example~\ref{ex:geom:trans:infty}
that 
\[
\begin{array}{l l l}
  \th T\I{\Box\Diam\Psi}
& =
& \left\{
  \thesis
  \bigwedge_{n \leq N}
  \bigvee_{m \in \NN}
  \Next^n \th F\I{\Psi}
  \lor
  \Next^{n+1} \th F\I{\Psi}
  \lor
  \dots
  \lor
  \Next^{n+m} \th F\I{\Psi}
  \mid
  N \in \NN
  \right\}
\end{array}
\]

Assume first 
$\nu \in \Mod\left(\th T\I{\Box\Diam\Psi} \right)$
and let $k \in \NN$.
Since
\[
\begin{array}{l l l}
  \nu
& \models
& \bigvee_{m \in \NN}
  \th F\I{\Psi}
  \lor
  \Next \th F\I{\Psi}
  \lor
  \dots
  \lor
  \Next^{m} \th F\I{\Psi}
\end{array}
\]

\noindent
we get some $i_1 \geq 0$ such that
$\nu \models \Next^{i_1} \th F\I{\Psi}$.
Since
\[
\begin{array}{l l l}
  \nu
& \models
& \bigvee_{m \in \NN}
  \Next^{i_1+1}\th F\I{\Psi}
  \lor
  \Next^{i_1+2} \th F\I{\Psi}
  \lor
  \dots
  \lor
  \Next^{i_1+m+1} \th F\I{\Psi}
\end{array}
\]

\noindent
we get some $i_2 > i_1$ such that
$\nu \models \Next^{i_2} \th F\I{\Psi}$.
Iterating this up to $i_k > \dots > i_2 > i_1$ yields the result
with $n \deq i_k+1$.

For the converse, using the simplifications
mentioned in Example~\ref{ex:geom:trans:uall},
it is sufficient to show that $\nu$ is a model of
\[
  \left\{
  \thesis
  \bigvee_{m \in \NN}
  \Next^N \th F\I{\Psi}
  \lor
  \Next^{N+1} \th F\I{\Psi}
  \lor
  \dots
  \lor
  \Next^{N+m} \th F\I{\Psi}
  \mid
  N \in \NN
  \right\}
\]

\noindent
Let $N \in \NN$.
Given $k > N$, by assumption there is some $n \geq k$
such that $\nu \models \psi_{n,k}$.
Hence, there are $0 \leq i_1 < \dots < i_k < n$
such that $\nu \models \Next^{i_j} \th F\I{\Psi}$
for all $j=1,\dots,k$.
In particular, we have $\nu \models \Next^{i_k} \th F\I{\Psi}$.
But since $0 \leq i_1 < \dots < i_k$,
we necessarily have $i_k \geq k-1$, and
since $k > N$ we get $k-1 \geq N$.
It follows that $i_k \geq N$, and we are done.
\end{proof}

Let us now explain how
\begin{equation}
\tag{\ref{eq:scott:ltl}}
  \forall \stream \in \I{\Stream\Base},~
  \text{$\stream$ total},
  \qquad
  \stream \forces \Box\Diam \Psi
  \quad\longimp\quad
  \I\filter\ p\ \stream \forces \Box \Phi
\end{equation}

\noindent
can be obtained from
\begin{equation}
\tag{\ref{eq:scott:geom}}
  \forall \stream \in \I{\Stream\Base},~
  \text{$\stream$ total},~
  \forall k \in \NN,~
  \forall n \geq k,
  \quad
  \nu(\stream) \models \psi_{n,k}
  \quad\longimp\quad
  \nu(\I\filter\ p\ \stream) \models \varphi_k
\end{equation}

\noindent
First, it follows from
Lemma~\ref{lem:geom:trans:cor}
and
Theorem~\ref{thm:geom:trans:cor}
that
condition \eqref{eq:scott:ltl}
amounts to
\[
  \forall \stream \in \I{\Stream\Base},~
  \text{$\stream$ total},
  \qquad
  \nu(\stream) \in \Mod(\th T\I{\Box\Diam \Psi})
  \quad\longimp\quad
  (\forall k \in \NN)\left(\nu(\I\filter\ p\ \stream) \models \varphi_k \right)
\]

\noindent
that is
\[
  \forall \stream \in \I{\Stream\Base},~
  \text{$\stream$ total},~
  \forall k \in \NN,
  \qquad
  \nu(\stream) \in \Mod(\th T\I{\Box\Diam \Psi})
  \quad\longimp\quad
  \nu(\I\filter\ p\ \stream) \models \varphi_k
\]

\noindent
Now, if
$\nu(\stream) \in \Mod(\th T\I{\Box\Diam \Psi})$,
then by 
Lemma~\ref{lem:app:scott:mod},
for all $k \in \NN$ we have
$(\exists n \geq k)(\nu(\stream) \models \psi_{n,k})$.
Hence, condition~\eqref{eq:scott:ltl}
follows from
\[
  \forall \stream \in \I{\Stream\Base},~
  \text{$\stream$ total},~
  \forall k \in \NN,
  \qquad
  (\exists n\geq k)\left(\nu(\stream) \models \psi_{n,k} \right)
  \quad\longimp\quad
  \nu(\I\filter\ p\ \stream) \models \varphi_k
\]

\noindent
and the latter is equivalent to
condition~\eqref{eq:scott:geom}.
Condition~\eqref{eq:scott:geom}
is a direct consequence of the following.

\begin{lemm}[Lemma~\ref{lem:scott}]
\label{lem:app:scott}
Write $g_n$ for $f_p^n(\bot_X) \colon \I{\Stream\Base} \to_\CPO \I{\Stream\Base}$.
Let $\stream \in \I{\Stream\Base}$
be a total stream.
\begin{enumerate}[(1)]
\item
Assume $k \leq n$.
If $\nu(\stream) \models \psi_{n,k}$,
then $\nu(g_n(\stream)) \models \varphi_k$.

\item
Let $n,k \in \NN$.
If $\nu(g_n(\stream)) \models \varphi_k$,
then $\nu(\I\filter\ p\ \stream) \models \varphi_k$.
\end{enumerate}
\end{lemm}

\begin{proof}
\hfill
\begin{enumerate}[(1)]
\item
We show by induction on $n \in \NN$ that for all
$\stream \in \I{\Stream\Base}$ and all $k \leq n$,
we have
\[
\begin{array}{l l l}
  \nu(\stream) \models \psi_{n,k}
& \longimp
& \nu(g_n(\stream)) \models \varphi_k
\end{array}
\]

\begin{description}
\item[Base case $n = 0$.]
In this case we have also $k = 0$.
But $\varphi_0 = \true$ and we are done.

\item[Induction step.]
Note that
\[
\begin{array}{*{7}{l}}
  g_{n+1}(\stream)
& =
& \term{if}~ p (\stream(0))
  ~\term{then}~ \stream(0) \cdot g_n(\stream \restr 1)
  ~\term{else}~ g_n(\stream \restr 1)
\end{array}
\]

Let $k \leq n+1$ and assume
$\nu(\stream) \models \psi_{n+1,k}$.
Hence there are $0 \leq i_1 < \dots < i_k < n+1$
such that $\nu(\stream) \models \Next^{i_j}\th F\I\Psi$
for all $j = 1,\dots,k$.

If $i_1 = 0$, then we have $p(\stream(0)) = \term{tt}$
and $g_{n+1}(\stream) = \stream(0) \cdot g_n(\stream \restr 1)$.
Moreover, since $\stream \restr 1 \models \psi_{n,k-1}$,
the induction hypothesis gives
$\nu(g_n(\stream\restr 1)) \models \varphi_{k-1}$.
Since $\stream\restr 1 \models \Phi$,
we obtain
$\nu(\stream(0) \cdot g_n(\stream\restr 1)) \models \varphi_{k}$
and we are done.

Otherwise, we have $i_1 > 0$.
Hence $\nu(\stream \restr 1) \models \psi_{n,k}$,
so that $\nu(g_n(\stream\restr 1)) \models \varphi_{k}$
by induction hypothesis.

Then if $p(\stream(0)) = \term{tt}$,
we have $g_{n+1}(\stream) = \stream(0) \cdot g_n(\stream \restr 1)$.
Since $\stream\restr 1 \models \Phi$,
we obtain
$g_{n+1}(\stream) \models \varphi_{k+1}$.
In particular, 
$\nu(g_{n+1}(\stream)) \models \varphi_k$
and we are done.

If $p(\stream(0)) = \term{ff}$,
then
$g_{n+1}(\stream) = g_n(\stream \restr 1)$
and we are done.

Note that the case of $p(\stream(0)) = \bot_{\I\Bool}$
cannot happen since $\stream$ is total and since
we assumed $p(a) \neq \bot_{\I\Bool}$ for all $a \in \Base$.
\end{description}

\item
Recall that $\I\filter p$ is the sup of the chain
\(
  \bot_X
  \leq_X
  f_p(\bot_X)
  \leq_X
  \cdots
  \leq_X
  f_p^n(\bot_X)
  \leq
  \cdots
\),
so that $g_n = f_p^n(\bot_X) \leq_X \I\filter p$.
Hence $g_n(\stream) \leq_{\I{\Stream\Base}} \I\filter\ p\ \stream$.

On the other hand,
it follows from
Proposition~\ref{prop:geom:algdcpo:scott}
(i.e.\ Proposition~\ref{prop:app:geom:algdcpo:scott})
that the set of all $\stream \in \I{\Stream\Base}$
such that $\nu(\stream) \models \varphi_k$
is Scott-open
and thus upward-closed.
Hence
$\nu(\I\filter\ p\ \stream)$ is a model of $\varphi_k$
whenever so is $\nu(g_n(\stream))$.
\qedhere
\end{enumerate}
\end{proof}

}
\opt{long}{}	

\opt{draft}{\input{draftnotes}}

\opt{full,long}{\newpage}
\opt{full,long}{\tableofcontents}

\end{document}